\documentclass[a4paper,onecolumn,11pt,accepted=2026-05-12]{quantumarticle}
\pdfoutput=1
\usepackage[utf8]{inputenc}
\usepackage[english]{babel}
\usepackage[T1]{fontenc}
\usepackage{amsmath}
\usepackage{hyperref}

\usepackage{tikz}
\usepackage{lipsum}
\usepackage{etoolbox}
\usepackage{enumerate}
\usepackage{multirow}
\usepackage[dvipsnames]{xcolor}
\usepackage{graphicx}
\usepackage{xfrac}
\usepackage{bm}
\usepackage{mathtools}
\usepackage{graphicx}
\usepackage{epstopdf}
\usepackage{csquotes}
\usepackage{subfig}
\usepackage{mathtools}
\usepackage{amssymb, amsmath,amsthm, amsfonts}
\usepackage{ifthen}
\usepackage{tikz}
\usepackage{cite}
\usepackage{enumerate}
\usepackage{verbatim}
\usepackage{pifont}
\usepackage{bm}  
\usepackage{stackengine}
\usepackage{bbm}

\newtheorem{theorem}{Theorem}
\newtheorem{cor}[theorem]{Corollary}
\newtheorem{lemma}[theorem]{Lemma}
\newtheorem{prop}[theorem]{Proposition}

\newtheorem{remark}{Remark}

\newtheorem{definition}{Definition}

\newtheorem{assump}{Assumption}

\newcommand{\bra}[1]{\langle#1|}
\newcommand{\ket}[1]{|#1\rangle}

\DeclareMathOperator{\Tr}{Tr}

\newcommand{\norm}[1]{\left\Vert#1\right\Vert}
\newcommand{\abs}[1]{\left\vert#1\right\vert}

\newcommand{\sD}{\mathsf{D}}

\newcommand{\cF}{\mathcal{F}}

\newcommand{\cH}{\mathcal{H}}

\newcommand{\cK}{\mathcal{K}}
\newcommand{\cL}{\mathcal{L}}
\newcommand{\cM}{\mathcal{M}}
\newcommand{\cN}{\mathcal{N}}

\newcommand{\cP}{\mathcal{P}}

\newcommand{\cS}{\mathcal{S}}

\newcommand{\cU}{\mathcal{U}}
\newcommand{\cV}{\mathcal{V}}
\newcommand{\cW}{\mathcal{W}}
\newcommand{\cX}{\mathcal{X}}

\newcommand{\cZ}{\mathcal{Z}}

\newcommand{\EE}{\mathbb{E}}

\newcommand{\HH}{\mathbb{H}}

\newcommand{\NN}{\mathbb{N}}

\newcommand{\PP}{\mathbb{P}}

\newcommand{\RR}{\mathbb{R}}

\newcommand{\ind}{\mathbbm 1}

\allowdisplaybreaks
\begin{document}

\title{Performance Guarantees  for Quantum Neural Estimation of Entropies}

\author{Sreejith Sreekumar}
\affiliation{Laboratoire Des Signaux Et Systèmes (L2S), CNRS, CentraleSupélec,  University of Paris-Saclay}
\orcid{0000-0003-3615-8961}
\email{sreejith.sreekumar@centralesupelec.fr}
\author{Ziv Goldfeld}
\affiliation{School of Electrical and Computer Engineering, Cornell University}
\email{goldfeld@cornell.edu}
\orcid{0000-0003-3406-3950}
\author{Mark M.~Wilde}
\affiliation{School of Electrical and Computer Engineering, Cornell University}
\email{wilde@cornell.edu}
\orcid{0000-0002-3916-4462}

\maketitle

\begin{abstract}
Estimating quantum entropies and divergences is an important problem in quantum physics, information theory, and machine learning. Quantum neural estimators (QNEs), which utilize a hybrid classical-quantum architecture, have recently emerged as an appealing computational framework for estimating these measures. Such estimators combine classical neural networks with parametrized quantum circuits, and their deployment typically entails tedious tuning of hyperparameters controlling the sample size, network architecture, and circuit topology. This work initiates the study of formal guarantees for QNEs of measured (R\'enyi) relative entropies in the form of non-asymptotic error risk bounds. We further establish exponential tail bounds showing that the error is sub-Gaussian and thus sharply concentrates about the ground truth value.  
For an appropriate sub-class of density operator pairs on a space of dimension~$d$ with bounded Thompson metric, our theory establishes a copy complexity of $O(|\Theta(\cU)|d/\epsilon^2)$ for QNE with a quantum circuit parameter set $\Theta(\cU)$, which has minimax optimal dependence on the accuracy~$\epsilon$.  Additionally, if the density operator pairs are permutation invariant, we improve the dimension dependence above to $O(|\Theta(\cU)|\mathrm{polylog}(d)/\epsilon^2)$. Our theory aims to facilitate principled implementation of QNEs for measured relative entropies and guide hyperparameter tuning in practice. 
\end{abstract}

\tableofcontents

\section{Introduction}

Entropy, a concept rooted in thermodynamics and statistical mechanics, plays a fundamental role in quantifying the amount of uncertainty or disorder in physical systems. Von Neumann and Shannon introduced these concepts to  quantum mechanics and classical  information theory, now referred to as the von Neumann entropy~\cite{NeumannThermodynamikGesamtheiten} and Shannon entropy~\cite{Shannon1948ACommunication}, respectively. R\'{e}nyi generalized this concept axiomatically, showing that any uncertainty measure satisfying a set of natural axioms can be characterized in terms of a single parameter family of entropy measures now known as R\'{e}nyi entropies~\cite{Renyi1961OnInformation}. Since their seminal works, related measures of information and discrepancy, such as mutual information, relative entropy~\cite{Kullback1951OnSufficiency,U62}, and their R\'{e}nyi counterparts \cite{Petz1985Quasi-entropiesAlgebra,Petz1986Quasi-entropiesSystems,muller2013quantum,wilde2014strong}, have found operational interpretations in classical and quantum hypothesis testing \cite{Petz1985Quasi-entropiesAlgebra,ogawa-nagaoka-2000,ogawa-hayashi-2004,Nussbaum-Szkola-2006,Audenaert-2008,BBH-2021}, channel coding, data compression, and other related tasks. 

Given the significance of entropic measures in characterizing the fundamental limits of quantum information systems, computing them is an important problem. However, often the underlying state is unknown, which necessitates estimating these quantities based on measurements of the quantum system. To this end, several estimators have been proposed in both the classical and quantum settings \cite{Li2017QuantumEstimation,Subasi2019EntanglementCircuit,AISW-2020, Subramanian2019QuantumStates,Gilyen2019DistributionalWorld,Gur2021SublinearEntropy, Wang2022QuantumEntropies,Wang2022NewDistances,Gilyen2022ImprovedEstimation,Wang2022QuantumSystems,
 Wang2023QuantumEstimation,Wang-Zhang-Li-2024,Wang-Zhang-2025,lu2025estimatingquantumrelativeentropies}. Recently,  quantum neural estimation (QNE), 
employing a hybrid classical-quantum architecture composed of neural nets and quantum circuits, has been proposed for estimating  quantum entropies and measured relative entropies \cite{QNE-24} (see also \cite{Shin2024,lee2023estimating}). QNE utilizes the variational form of the entropic measure given as the supremum over Hermitian operators of a trace functional with respect to the underlying states \cite{Petz1988,petz2007quantum,Berta2015OnEntropies}. By parametrizing the eigenvalues and eigenvectors of the Hermitian operator by the neural net and quantum circuit, respectively,  the QNE leverages their expressivity to construct scalable  estimators. 
In~\cite{QNE-24}, it was shown via numerical simulations that the QNE is consistent.
\subsection{Summary of Contributions}
In this paper, we provide the first theoretical guarantees for QNE in the form of sharp bounds on the expected absolute error. The bounds depend on the number of quantum state copies the algorithm uses, parameters defining the class of underlying density operators, and the classical-quantum circuit parameters. Furthermore, we also derive concentration inequalities for the absolute error about the ground truth, which leads to an upper bound on the copy complexity of QNE, i.e., the minimum number of copies of the quantum states required for estimating measured relative entropies.

\begin{table}
\centering
\begin{tabular}{ |c|c|c |}
 \hline
 R\'enyi order and state class&  Minimax risk upper bound     &  Minimax risk lower bound   \\
 \hline
 $\alpha=1$ and $\cS \subseteq \cS_d(b)$   & $O_b\!\left(\delta+\sqrt{\frac{d\,\Theta(\delta,\cS)}{n}}\right)$    & $\Omega_b\!\left(\frac{1}{\sqrt{n}}+\frac{d}{n \log d}\right)$ \\[4pt]
 \hline
$0 <\alpha  \neq 1 $ and  $\cS \subseteq \cS_d(b)$     & $O_{\alpha,b}\!\left(\delta +   \sqrt{\frac{d\,\Theta(\delta,\cS)}{n}}\right)$    & $\Omega_{\alpha,b}\!\left(\frac{1}{\sqrt{n}}+\frac{d}{n \log d}\right)$  \\[4pt]
 \hline
\end{tabular}
\caption{Minimax risk bounds for estimation of measured relative entropies over the class  $\cS \subseteq \cS_d(b)$.  Here, $b$ is a uniform upper bound on the Thompson metric  between density operator pairs in $\cS_d(b)$, $\delta$ is an approximation error parameter for PQC parametrized by $\Theta(\delta,\cS)$, and $n$ is the number of copies of quantum states on a Hilbert space of dimension $d$  used by the QNE. The upper bound holds for any $\cS \subseteq \cS_d(b)$ while the lower bound holds for any $\cS \subseteq \cS_d(b)$ containing all state pairs in any fixed eigen-basis, when  $d b \lesssim n \log d\,$ and $b \geq (\log d)^2 \vee 2$.}
\label{table:summ-main-cont}
\end{table}
To summarize our main contributions, let $\cS_d(b)$ be the set of density operator pairs $(\rho,\sigma)$ for a Hilbert space of dimension $d$, such that their Thompson metric \cite{thompson_1963} is bounded from above by $\log b$.  Assuming that the neural net and quantum circuit are sufficiently expressive and have bounded statistical complexity (see Assumption~\ref{Assump1} and Section~\ref{Sec:PQC}), we show the following: 
\begin{enumerate}[(i)]
    \item We derive an upper bound on the expected  absolute error achieved by QNE for estimating the measured relative entropy (Theorem~\ref{Thm:minimaxrisk-QNE}) and the measured R\'{e}nyi relative  entropy (Theorem~\ref{Thm:minimaxrisk-QNE-Renyi}) over the class $\cS_d(b)$. Further, we establish an exponential deviation inequality that shows that the absolute error has sub-Gaussian tails.
    
    \item To obtain concrete bounds, we specialize the aforementioned results to shallow neural nets with sufficiently large width (Corollary~\ref{Cor-shallownn} and Corollary \ref{Cor-shallownn-Renyi}) and deep neural nets with a sufficient number of layers (Proposition~\ref{Prop:deepnn} and Proposition \ref{Prop:deepnn-Renyi}). An upper bound on the absolute error achieved by a QNE with a shallow net over the class $\cS_d(b)$ is illustrated in Table \ref{table:summ-main-cont}, along with a lower bound on the minimax risk (valid for any estimator) over the same class.  
    
    \item Let $d=m^N$, with $N$ denoting the number of qudits, and let $\rho,\sigma$ be density operators defined on $N$ qudits. While the QNE has a worst-case copy complexity that scales exponentially in $N$, we show that for a subclass (which depends on the quantum circuit) of permutation invariant states across $N$ qudits within $\cS_d(b)$,  the algorithm has a much better polynomial dependence (see Proposition~\ref{Prop:copcomp-perminv}). This is possible due to Schur--Weyl duality from representation theory, which implies that the number of degrees of freedom for such states scales polynomially in $N$.
\end{enumerate}

\subsection{Literature Review and Related Work}
Measured relative entropy was introduced in  \cite{Donald1986,P09} and since then has found operational significance in terms of characterizing the Stein's and strong converse exponents in centralized, adversarial, and distributed hypothesis testing (see e.g.,  \cite{Hiai-Petz-1991,Hayashi_2002, Mosonyi2015QuantumHT,Brando2020AdversarialHT, RSB-IT-2025,SHCB-2025}). 
There is an extensive literature on estimation of entropic, information and divergence measures in the classical setting. Given this, we focus primarily on the finite-dimensional quantum setting. 
The reader is referred to, e.g.,  \cite{Verdu-2019,SK-2025} and references therein, for a detailed account of relevant works in the classical setting. In the quantum setting,  a pertinent approach for estimating such measures is to use quantum tomography to learn the entire density matrix, from which the desired entropic quantity of interest can be computed.  The copy complexity of this approach is linear in the dimension of the state (or, equivalently, exponential in the number of qudits) \cite{HHJWY16,wright2016learn,OW16,Yuen2023improvedsample}.

A possibly more efficient approach is to estimate the entropies directly, without estimating the underlying states.  In this regard, several estimators have been proposed and their computational complexities  investigated under different input models \cite{Li2017QuantumEstimation,Subasi2019EntanglementCircuit,AISW-2020, Subramanian2019QuantumStates,Gilyen2019DistributionalWorld,Gur2021SublinearEntropy, Wang2022QuantumEntropies,Wang2022NewDistances,Gilyen2022ImprovedEstimation,Wang2022QuantumSystems,
Wang2023QuantumEstimation,Wang-Zhang-Li-2024,Wang-Zhang-2025}. In particular, \cite{AISW-2020} provided  algorithms for estimating the von Neumann entropy and R\'enyi entropies, given access to independent copies of the input quantum state,  which achieves the copy complexity (sub-linear in dimension) for integer orders greater than one.   Entropy estimation under the quantum query model, i.e., where one has access to an oracle that prepares the input quantum state, was explored in~\cite{Subramanian2019QuantumStates,Gur2021SublinearEntropy}. These works also established  sub-linear (in dimension) query complexity bounds for estimating the von Neumann  and R\'{e}nyi entropies.

In contrast to the  estimators mentioned above, QNE proposed in \cite{QNE-24} is based on  variational methods, which have become  a prominent research area in quantum computing and machine learning \cite{Biamonte2017QuantumLearning,Cerezo2020variationalquantum,Cerezo2021VariationalAlgorithms,Jaeger2023}. Specifically, QNEs utilize a hybrid quantum-classical algorithm that extends the approach from \cite{NWJ10,pmlr-v70-arora17a,Zhang-2018,belghazi-2018,SS-2021-aistats,SG-2021-neural,Birell-2021} to the quantum framework. The appeal of such estimators stems from their scalability to high-dimensional problems and large datasets, as well as their compatibility with modern gradient-based optimization techniques.  This makes them particularly relevant in the quantum context, where the dimension of the underlying Hilbert space scales exponentially in the number of qudits. Although QNE was shown to perform empirically well in \cite{QNE-24}, theoretical performance guarantees are still lacking,  which this work aims to advance upon. We note that variational methods have also been used in \cite{VQA-entropies-fidelity-2021} for estimating fidelity, entropies and Fisher information, in \cite{Shin2024} for estimating quantum mutual information and in \cite{lee2023estimating} for estimating entanglement entropy, which indicates their versatility.

In the framework of \cite{QNE-24} for estimating measured relative entropies, the optimization is performed over Hermitian operators of the form $H(\theta,\phi) = U(\bm{\theta})\,\mathrm{diag}(\Lambda_\phi)\,U(\bm{\theta})^\dagger$, where a parameterized quantum circuit (PQC) $U(\bm{\theta})$ specifies the eigenvectors and a neural network defines the eigenvalues $\Lambda_\phi$. Although this hybrid parameterization effectively captures the optimization of the variational formula $\mathrm{Tr}[H\rho] + 1 - \mathrm{Tr}[e^{H}\sigma]$ for measured relative entropy (and a related formula for measured R\'enyi relative entropy), it remains susceptible to barren plateaus \cite{McClean_2018,Cerezo2021CostDependent,Holmes2022Trainability}. In particular, if the PQC is sufficiently deep and unstructured, $U(\bm{\theta})$ approximates a unitary two-design, causing $\mathrm{Tr}[H\rho]$ and $\mathrm{Tr}[e^{H}\sigma]$ to concentrate around their means and the gradient variances to vanish exponentially in the number of qubits. This leads to vanishing gradients in the quantum parameters $\bm{\theta}$ and, indirectly, to weakened gradient signals for the classical parameters $\phi$. Barren plateaus can be mitigated by employing shallow or symmetry-informed ansätze \cite{Meyer2023ExploitingSymmetry,Nguyen2024TheoryEquivariantQNN}, initializing $U(\bm{\theta})$ near the identity \cite{Grant2019Initialization}, or adopting layer-wise or hybrid training strategies that preserve gradient signal throughout optimization \cite{Skolik2021Layerwise}. Importantly, despite its variational structure, the QNE framework is not susceptible to the classical-simulability arguments discussed in~\cite{Cerezo2025ProvableAbsence}, since it assumes direct access to quantum samples of the states $\rho$ and $\sigma$ rather than classical descriptions or measurement data. The learning task thereby exploits intrinsically quantum input data, placing it outside the regime of efficiently classically simulable models identified in~\cite{Cerezo2025ProvableAbsence}.
\subsection*{Paper Organization}
The rest of the paper is organized as follows. Section~\ref{Sec:Preliminaries} introduces the notation and preliminary concepts along with QNE. Section~\ref{Sec:minimaxriskbnds} obtains performance guarantees in terms of expected absolute error bounds and  deviation inequalities for estimating measured relative entropies using QNE. Section~\ref{Sec:copycomp} discusses the copy complexity achieved by QNE. Section~\ref{Sec:proofs} presents proofs of the main results, with proofs of the auxiliary results deferred to the Appendix. Finally, Section~\ref{Sec:concl-remarks}  concludes the paper with suggestions of avenues for future research.  
\section{Preliminaries} \label{Sec:Preliminaries}
\subsection*{Notation}
Throughout, we consider a $d$-dimensional complex Hilbert space $\HH_d$. The set of (linear) operators from $\HH_d$  to itself is denoted by  $\cL(\HH_d)$. The set of density operators (or quantum states) on $\HH_d$, i.e., positive semi-definite operators with unit trace is denoted by $\mathbb{S}_d $. 
$\operatorname{Tr}[\cdot]$ and $\norm{\cdot}_p$ for $p \geq 1$ signify the trace operation   and Schatten $p$-norm, respectively. For normal operators $A$ (which includes Hermitian and unitary operators), $\norm{A}_{\infty}=\norm{A}$, where $\norm{A}$ represents  the operator norm of $A$. The usual norm associated with the Lebesgue space of $p$-integrable functions ($p  \geq 1$) on $\cX \subseteq \RR^d$ with respect to a measure $\mu$ is denoted by $\| \cdot \|_{p,\mu}$, while $\|\cdot\|_{\infty,\cX}$ designates the standard sup-norm, i.e., $\|f\|_{\infty,\cX}=\sup_{x \in \cX} \abs{f(x)}$.  
$\ind_{\cX}$ denotes the indicator of a set $\cX$.
For reals $a,b$, we define $a \wedge b \coloneqq \min\{a,b\}$ and $a \vee b \coloneqq \max\{a,b\}$. Further, $a\lesssim_x b$ indicates that $a \leq c_x b$ for a  constant $c_x>0$ that depends only on $x$, and  $a\gtrsim_x b$ means that $b \lesssim_x a$. 
\subsection{Measured  Relative Entropy and Measured R\'enyi Relative  Entropy}\label{Sec:measured-rel-entropies}
The \emph{measured relative entropy} between $\rho,\sigma \in \mathbb{S}_d$ is defined as~\cite{Donald1986,P09}
\begin{equation}
\sD_{\mathsf{M}}(\rho\Vert\sigma)\coloneqq\sup_{\cZ,\left\{  M_z\right\}  _{z\in\cZ}}
\sum_{z\in\cZ}\operatorname{Tr}[M_z\rho]\log \!\left(  \frac{\operatorname{Tr}
[M_z\rho]}{\operatorname{Tr}[M_z\sigma]}\right),
\label{eq:measured-rel-def}
\end{equation}
where the supremum is over all finite sets $\cZ$ and positive operator-valued measures\footnote{A POVM $\left\{  M_z\right\}_{z\in\cZ}$  on $\HH_d$ indexed by a discrete $\cZ$ is a set of positive semi-definite operators satisfying $\sum_{z \in \cZ} M_z = I$.} (POVMs) $\left\{  M_z\right\}_{z\in\cZ}$ indexed by~$\cZ$.
The measured relative entropy can be expressed in terms of the following variational expression~\cite[Lemma~1]{Berta2015OnEntropies}: 
\begin{equation}
\sD_{\mathsf{M}}(\rho\Vert\sigma)  =\sup_{H}\operatorname{Tr}[H\rho
]-\operatorname{Tr}[e^H\sigma]+1,\label{eq:rel_ent_variational} 
\end{equation}
where the supremum is over the class of Hermitian operators $H$ on $\HH_d$.

The \emph{measured R\'{e}nyi relative  entropy} of order $\alpha \in (0, 1) \cup (1,\infty)$ between $\rho,\sigma \in \mathbb{S}_d$ is defined as follows~\cite[Eqs.~(3.116)--(3.117)]{F96}:
\begin{align}
\sD_{\mathsf{M},\alpha}(\rho\Vert\sigma) \coloneqq 
\sup_{\cZ,\left\{  M_z\right\}  _{z\in\cZ}
} \frac{1}{\alpha-1}\log  \sum_{z\in\cZ} \operatorname{Tr}[M_z\rho]^{\alpha}\operatorname{Tr}[M_z\sigma]^{1-\alpha} ,
\label{eq:renyi_rel_ent}
\end{align}
where the supremum is over all finite sets $\cZ$ and POVMs  $\left\{  M_z\right\}_{z\in\cZ}$.
 Recall that it satisfies the  variational expression~\cite[Lemma~3]{Berta2015OnEntropies}:
 \begin{align}
\sD_{\mathsf{M},\alpha}(\rho\Vert\sigma)  =\sup_{H}\frac{\alpha}{\alpha-1}\log  \operatorname{Tr}[e^{\left(  \alpha-1\right)  H}\rho
]-\log  \operatorname{Tr}[e^{\alpha H}\sigma].
\label{eq:renyi_rel_ent_variational}
 \end{align}
When $\rho,\sigma>0$, the suprema in~\eqref{eq:rel_ent_variational} and~\eqref{eq:renyi_rel_ent_variational} are achieved for some $H_{\alpha}^{\star}(\rho,\sigma)$~\cite[Lemma~1 and Lemma~3]{Berta2015OnEntropies} ($H_{1}^{\star}(\rho,\sigma)$ corresponds to the optimizer of~\eqref{eq:rel_ent_variational}). Specifically, the  proof of~\cite[Lemma~1 and Theorem~2]{Berta2015OnEntropies} shows that
\begin{align}
&H_{\alpha}^{\star}(\rho,\sigma)= \sum_{i=1}^d\lambda_{i,\alpha}^{\star}  P_{i,\alpha}^{\star}, \quad \mbox{ with } \lambda_{i,\alpha}^{\star} \coloneqq \log  \!\left(\frac{\operatorname{Tr}[P_{i,\alpha}^{\star} \rho]}{\operatorname{Tr}[P_{i,\alpha}^{\star} \sigma]}\right),\label{eq:opteigvalues}     
\end{align}
 where $\{P_{i,\alpha}^{\star}\}_{i=1}^d$ denotes a set of orthonormal rank-one projectors from the spectral decomposition of $H_{\alpha}^{\star}(\rho,\sigma)$. When $\alpha=1$, we henceforth simply write 
 $\{\lambda_{i}^{\star}\}_{i=1}^d$, $\{P_{i}^{\star}\}_{i=1}^d$, and $H^{\star}(\rho,\sigma)$. 
 Note that while $\{P_{i,\alpha}^{\star}\}_{i=1}^d$ 
 need not be unique in general, $H_{\alpha}^{\star}(\rho,\sigma)$ is unique (see Appendix~\ref{Sec:optmeasrelent}). 
 Also, although a closed form expression for $H_{\alpha}^{\star}(\rho,\sigma)$ is not known, it can be shown   by setting matrix (Fr\'{e}chet) derivatives  with respect to $H$ in~\eqref{eq:rel_ent_variational} and~\eqref{eq:renyi_rel_ent_variational}  to zero  that 
$H_{\alpha}^{\star}(\rho,\sigma)$  satisfies 
 \begin{align}
\int_{0}^1 e^{t(\alpha-1) H_{\alpha}^{\star}(\rho,\sigma)} \rho e^{(1-t) (\alpha-1) H_{\alpha}^{\star}(\rho,\sigma)} dt=\int_{0}^1 e^{t \alpha H_{\alpha}^{\star}(\rho,\sigma)} \sigma e^{(1-t)\alpha H_{\alpha}^{\star}(\rho,\sigma)} dt.\label{eq:optrenmeasrelent}
 \end{align}
 Furthermore, we note from the proof of~\cite[Lemma~1 and Theorem~2]{Berta2015OnEntropies}  that the projective POVMs formed by the optimizers $\{P_{i}^{\star}\}_{i=1}^d$ and $\{P_{i,\alpha}^{\star}\}_{i=1}^d$  turn out to be optimal measurements achieving the suprema in~\eqref{eq:measured-rel-def} and~\eqref{eq:renyi_rel_ent}, respectively.

 There exists an alternative variational form for measured relative entropies, as given in \cite[Lemma~1 and Lemma~2]{Berta2015OnEntropies}:
\begin{align}
\sD_{\mathsf{M}}(\rho\Vert\sigma)&=\sup_{\omega>0} \Tr[\rho\log\omega]+1-\Tr[\sigma\omega],\label{eq:measured-rel-ent-opt} \\
    \sD_{\mathsf{M},\alpha}(\rho\Vert\sigma)&= \frac{1}{\alpha-1} \log Q_{\alpha}(\rho\|\sigma), \quad \alpha \in (0,1) \cup (1,\infty), \label{eq:varformrenentmain2}
\end{align}
where 
\begin{align}\label{eq:varformrenent2}
Q_{\alpha}(\rho\|\sigma)=\begin{cases}
\inf_{\omega>0}\big\{ \alpha\Tr[\rho\omega]+\left(1-\alpha\right)\Tr[\sigma\omega^{\frac{\alpha}{\alpha-1}}]\big\}, & \qquad\text{ for }\alpha  \in\left(0,\frac{1}{2}\right),\\
\inf_{\omega>0}\big\{ \alpha\Tr[\rho\omega^{1-\frac{1}{\alpha}}]+\left(1-\alpha\right)\Tr[\sigma\omega]\big\}, & \qquad\text{ for }\alpha  \in\left[\frac{1}{2},1\right),\\
\sup_{\omega>0}\big\{ \alpha\Tr[\rho\omega^{1-\frac{1}{\alpha}}]+\left(1-\alpha\right)\Tr[\sigma\omega]\big\}, & \qquad\text{ for }\alpha  \in\left(1,\infty\right).
\end{cases}
\end{align}
The advantage of these forms is that \eqref{eq:measured-rel-ent-opt} and \eqref{eq:varformrenentmain2} are (strictly) convex optimization problems  over the cone of positive definite matrices when $\rho,\sigma>0$, which  can be utilized to obtain a necessary and sufficient condition that is satisfied by $H_{\alpha}^{\star}(\rho,\sigma)$.  For instance,
  setting the matrix derivatives of the expressions on the right-hand side (RHS) of~\eqref{eq:varformrenent2}  to zero yields that, for $\alpha \in [\sfrac{1}{2},1) \cup (1,\infty)$, the unique optimizer  $\omega_{\alpha}^{\star}(\rho,\sigma)$ $\big(= e^{\alpha H_{\alpha}^{\star}(\rho,\sigma)}\big)$ satisfies 
 \begin{align}
\sigma & =\left(\frac{\alpha}{1-\alpha}\right)\frac{\sin\!\left(\left(\frac{1-\alpha}{\alpha}\right)\pi\right)}{\pi}\int_{0}^{\infty}\ t^{\frac{\alpha-1}{\alpha}}\left(\omega_{\alpha}^{\star}(\rho,\sigma)+tI\right)^{-1}\rho\left(\omega_{\alpha}^{\star}(\rho,\sigma)+tI\right)^{-1}dt.\label{eq:optimal-omega-alpha-gen}
\end{align}
Note that  $\rho$ and $\sigma$ above  are related by a completely positive map involving $H_{\alpha}^{\star}(\rho,\sigma)=\log \omega_{\alpha}^{\star}(\rho,\sigma)$.
Details of the above  are provided in  Appendix~\ref{Sec:optmeasrelent}. See also~\cite{huang2025semidefiniteoptimizationmeasuredrelative} for how to compute the measured (R\'enyi) relative entropies by means of semi-definite optimization. 
\subsection{Minimax Estimation of Measured Relative Entropies}
Consider the problem of estimating measured relative entropies $\sD_{\mathsf{M},\alpha}(\rho\Vert\sigma)$ and $\sD_{\mathsf{M}}(\rho\Vert\sigma)$ based on $n$ copies of both $\rho$ and $\sigma$. We are interested in the decay rate of the QNE absolute error with $n$, and how it compares to the performance of an optimal algorithm for estimating entropies. It is well known from classical results that estimating entropies, in general, suffers from the so-called curse of dimensionality~\cite{Paninski-2003}, in that, the estimation error is $\Omega\big(n^{-\frac 1d}\big)$. In other words, estimating entropies is hard in high dimensions without additional assumptions that restrict the class of density operators.  

The performance of an estimator is quantified by the minimax risk, namely, the error of the best estimator uniformly over a class of density operators. Formally, the minimax risk for measured relative entropy over a class $\cS$  is defined as follows:
\begin{align}
R_n^{\star}(\cS)\coloneqq \inf_{\hat\sD_n} \sup_{(\rho,\sigma) \in \cS}   \EE\!\left[ \abs{\sD_{\mathsf{M}}(\rho\Vert\sigma) -\hat\sD_n(\rho^{\otimes n},\sigma^{\otimes n})}\right], \label{eq:minimaxrisk}
\end{align}
where the infimum is over all estimators $\hat\sD_n$, i.e., a map $(\rho^{\otimes n},\sigma^{\otimes n})\mapsto \hat\sD_n(\rho^{\otimes n},\sigma^{\otimes n})$, which is allowed to perform quantum operations (apply quantum channels, measurements, etc.) and the expectation is over any randomness resulting from such quantum operations. Note that the minimax risk is a property of the estimation setting---the best possible performance uniformly over the class---as opposed to the specific estimator. The minimax risk for measured R\'{e}nyi relative entropy, $R_{n,\alpha}^{\star}(\cS)$,  is defined analogously by replacing $\sD_{\mathsf{M}}(\rho\Vert\sigma)$ with $\sD_{\mathsf{M},\alpha}(\rho\Vert\sigma)$ in~\eqref{eq:minimaxrisk}. 

Recalling the Thompson metric
$T(\rho,\sigma)\coloneqq \log \!\big(\big\|\sigma^{-1/2}\rho \sigma^{-1/2}\big\| \vee  \big\|\rho^{-1/2} \sigma \rho^{-1/2}\big\|\big)$ ~\cite{thompson_1963} ($\big\|\sigma^{-1/2} \rho \sigma^{-1/2}\big\| \wedge \big\|\rho^{-1/2} \sigma \rho^{-1/2}\big\| \geq 1$ with equality iff $\rho=\sigma$), in what follows, we are interested in the minimax risk over a subclass $\cS \subseteq \cS_d(b)$, where  for $b \geq 1$, 
  \begin{align}
      \cS_d(b)\coloneqq \left\{(\rho,\sigma): 0<\rho,\sigma \in \mathbb{S}_d, T(\rho,\sigma)\leq \log b\right\}. \label{eq:densityopclass}
  \end{align}
Note that  $T(\rho,\sigma)$ being bounded constrains the minimal eigenvalue of both $\rho$ and $\sigma$ to be bounded away from zero. In fact, it is equivalent to $D_{\max}(\rho\|\sigma) \vee D_{\max}(\sigma\|\rho) \leq \log b$, where $D_{\max}(\rho\|\sigma)$ denotes the max-relative entropy between $\rho$ and $\sigma$~\cite{Datta-2009} (see Appendix \ref{Sec:Lem:bndtraceratio-proof}).  
We will use the fact that the Thompson metric contracts under the application of a quantum channel \cite[Lemma 7]{Datta-2009}, i.e., $T(\rho,\sigma) \geq T\big(\cN(\rho),\cN(\sigma)\big) $ for every quantum channel $\cN$.  
\subsection{Quantum Neural Estimator for Estimating Measured Relative Entropies}\label{Sec:QNE}
The QNE proposed in~\cite{QNE-24} has a classical neural class $\cF$ and parametrized quantum circuit (PQC) $\cU(\Theta)$, indexed by a set $\Theta$. It estimates $\sD_{\mathsf{M}}(\rho\Vert\sigma) $  as follows:
\begin{subequations}\label{eq:QNE:optim}
 \begin{align}
\hat\sD_{\mathsf{M},\Theta,\cF }^n \coloneqq \sup_{\bm{\theta} \in \Theta}\hat\sD_{\mathsf{M},\bm{\theta},\cF}^n,   \label{eq:QNE-outer_opt}
\end{align}
where 
\begin{align}
\hat\sD_{\mathsf{M},\bm{\theta},\cF}^n&\coloneqq  \sup_{\substack{f \in \cF}}\frac{1}{n}\sum_{\ell=1}^{n}f(i_\ell(\bm{\theta}))-\frac{1}{n}\sum_{\ell=1}^{n}
e^{f(j_\ell(\bm{\theta}))}+1,\label{eq:QNE-inner_opt}
\end{align}   
\end{subequations}
and
$i_\ell(\bm{\theta}) \sim p_{\bm{\theta}}^\rho $, $j_\ell(\bm{\theta}) \sim q_{\bm{\theta}}^\sigma$, with
\begin{subequations}\label{eq:measdist}
 \begin{align}
   p_{\bm{\theta}}^\rho(i)&\coloneqq  \operatorname{Tr}\!\left[ U(\bm{\theta}) |i\rangle\!\langle
i|U^{\dag}(\bm{\theta})\rho\right], \\
   q_{\bm{\theta}}^\sigma(i)&\coloneqq \operatorname{Tr}\!\left[ U(\bm{\theta}) |i\rangle\!\langle
i|U^{\dag}(\bm{\theta})\sigma\right].  
 \end{align}   
\end{subequations}
 The analogous QNE estimate for measured R\'{e}nyi relative  entropy is given by
 \begin{align}
\hat{\sD}_{\mathsf{M},\alpha,\Theta, \cF}^n \coloneqq \sup_{\bm{\theta} \in \Theta} \hat{\sD}_{\mathsf{M},\alpha,\bm{\theta}, \cF}^n, \notag
\end{align}
where 
\begin{align}
\hat{\sD}_{\mathsf{M},\alpha,\bm{\theta}, \cF}^n\coloneqq 
\frac{\alpha}
{\alpha-1} \log  \frac{1}{n}\sum_{k=1}^{n}e^{(\alpha-1)f(i_k(\bm{\theta}))}
-\log \frac{1}{n}\sum_{k=1}^{n}
e^{\alpha f(j_k(\bm{\theta}))}, \notag
\end{align}
and $i_\ell(\bm{\theta}) \sim p_{\bm{\theta}}^\rho $, $j_\ell(\bm{\theta}) \sim q_{\bm{\theta}}^\sigma$ with $p_{\bm{\theta}}^\rho $ and $q_{\bm{\theta}}^\sigma$ as given in
\eqref{eq:measdist}. In practice, $\hat\sD_{\mathsf{M},\Theta,\cF }^n$ and $\hat{\sD}_{\mathsf{M},\alpha,\Theta, \cF}^n$ are computed via alternating optimization, by performing  gradient ascent iteratively over the neural net class $f \in \cF$ (for a fixed $\bm{\theta}$) followed by that over $\Theta$.

For flexibility of the framework, we will also allow  applying  a fixed quantum channel to each copy of $\rho$ and $\sigma$ before performing measurements using the PQC; i.e., $\rho$ and $\sigma$ in~\eqref{eq:measdist} are replaced by $\cN(\rho)$ and $\cN(\sigma)$, respectively, for a channel $\cN$, with  $\cU(\Theta)$ adjusted to be a set of unitary operators on the channel's output space. The rationale behind this is to leverage any additional known structure on the quantum states to reduce the copy complexity of QNE.  For instance, when $\rho$ and $\sigma$ are permutationally invariant, it is known from Schur--Weyl duality that the number of degrees of freedom in such states scales polynomially in the number of qubits (as opposed to exponential in the general case). In this case,  the complexity of QNE can be reduced by applying a pre-processing channel $\cN$ chosen based on the Schur transform and running QNE on its output space. We note that there are various efficient quantum circuit implementations \cite{bacon2006quantum,harrow2005applications,Kirby2018schur,Krovi2019efficienthigh} for the Schur transform. In the general case (without any additional known structure), we take the pre-processing channel to be the identity channel.

Our objective is to quantify the performance of QNE in terms of $d$, $n$, and the parameters defining $\cF$, $\cU$, and $\cS$. To that end, we split the error of QNE into two components\footnote{In practice, there is also an optimization error component, which arises when the suprema in~\eqref{eq:QNE:optim} are only approximately computed, e.g., by gradient-based methods. This can be incorporated as an additive factor in the overall error. A detailed analysis however has to take into account the specific parameterizations of the classical neural net and quantum circuit. Hence, we ignore this aspect in our analysis and leave its quantification as part of future work.}: \textit{approximation} and \textit{statistical}. The approximation error depends on the ability of the classical neural net and PQC  to approximate, respectively, the eigenvalues and eigenvectors of the optimizer $H^{\star}(\rho,\sigma)$ given in~\eqref{eq:opteigvalues}. The statistical error, on the other hand,  is the random component of the overall error that depends on the complexity of the classical-quantum model used by the QNE. We briefly note that a QNE  for measured relative entropies may also be constructed by using the variational forms given in~\eqref{eq:measured-rel-ent-opt} and~\eqref{eq:varformrenentmain2} (by estimating $Q_{\alpha}(\rho\|\sigma)$ when $\alpha\neq 1$ and substituting in~\eqref{eq:varformrenentmain2}). An analysis of this estimator could be performed similarly to that of~\eqref{eq:QNE:optim},  with the resulting performance guarantees differing only in  terms of constants that depend on $b$ and $\alpha$. Hence, we omit it. 
\subsection{Relevant Concepts from Probability Theory}
Here, we introduce some  notions from probability theory (see e.g. \cite{VanHandel-book}) which will be useful in our analysis of QNE performance.   We will require that the neural class used in QNE has bounded complexity (see Assumption \ref{Assump1} below), as quantified using the notion of covering entropy \cite{tikhomirov1993varepsilon} of a function class in a suitable metric. 
\begin{definition}
[Covering  entropy] \label{cov-pack-num}
Let $(\Theta,\mathsf{d})$ be a metric space.  A set $\Theta'\subseteq\Theta$ is an $\epsilon$-covering of $(\Theta,\mathsf{d})$  if, for every $\theta \in \Theta$, there exists $\theta' \in \Theta' $ such that $\mathsf{d}(\theta,\theta')\leq \epsilon$.  The $\epsilon$-covering number is defined as
\begin{align}
N(\epsilon,\Theta,\mathsf{d})\coloneqq \inf \left\{|\Theta'|:\,\Theta' \mbox{ is an } \epsilon \mbox{-covering of }\Theta \right\},  \notag
\end{align}
and the $\epsilon$-covering entropy is $\log  N(\epsilon,\Theta,\mathsf{d})$.
\end{definition}
To obtain a handle on the statistical component of the overall error, we will upper bound it in terms of suprema of separable sub-Gaussian processes indexed by neural class and quantum circuit parameters. 
\begin{definition}[Separable sub-Gaussian process]
    A real-valued stochastic process $\left(X_{\theta}\right)_{\theta \in \Theta}$ on a metric space $(\Theta,\mathsf{d})$ is  sub-Gaussian if it is centered, i.e., $\EE[X_{\theta}]=0$ for all $\theta \in \Theta$, and $\mathbb{E}\big[e^{t(X_{\theta}-X_{\tilde{\theta}})}\big] \leq e^{\frac 12 t^2\mathsf{d}(\theta,\tilde{\theta})^2},\quad \forall ~\theta,\tilde{\theta} \in \Theta,~t \geq 0$. A stochastic process $(X_{\theta})_{\theta \in \Theta}$ is separable if there exists a countable subset $\Theta_0 \subseteq \Theta$ and a probability null set $\cN$ such that for every $\omega \notin \cN$ and $\theta \in \Theta$, there exists a sequence $\left(\theta_n\right)_{n \in \NN}$ in $\Theta_0$ such that $\mathsf{d}(\theta_n,\theta) \rightarrow 0$ and $X_{\theta_n}(\omega) \rightarrow X_{\theta}(\omega)$.
\end{definition} 
\subsection{Approximating Functions and Operators in Quantum Neural Estimation}\label{Sec:approxfunc-op}
At the core of our framework is the question of approximating the eigenvalues and eigenvectors of Hermitian operators. In order for this approximation to hold uniformly over a class of density operators of interest, we require the neural class and quantum circuit to be sufficiently expressive. On the other hand, the complexity of these approximating classes should be balanced so that the statistical component of the estimation error is minimized. To this end, we assume that the neural class satisfies the (mild) high-level assumptions stated below, which, in particular, are satisfied by feedforward neural nets with finite width and depth, bounded parameters and common continuous activations such as rectified linear (ReLU) and sigmoid. After stating these for the neural class, we provide some basic definitions for quantum circuits. 
\subsubsection{Classical Neural Networks}
To handle the approximation error, we will employ bounds from approximation theory for neural networks. For this purpose, we  embed the set $\{1,\ldots,d\}$ indexing the eigenvalues  as points into the Euclidean space $\RR^{m}$, and we consider the neural net approximation of a function whose values at these points are equal to the eigenvalues of $H^{\star}(\rho,\sigma)$ or $H_{\alpha}^{\star}(\rho,\sigma)$. To account for a broad class of embeddings (an injective map) and classical neural nets, we formulate a general framework and state the assumptions  under which our bounds hold. Later, we consider natural instances of embeddings and neural nets that satisfy these assumptions. To obtain a handle on the statistical error, we require that the neural class  has bounded complexity, quantified in terms of the covering entropy as stated below.

For a positive integer $\ell \leq d$, let $e_m\colon \{1, \ldots,d\} \mapsto  \RR^{m}$ denote an embedding, and let $\cX\coloneqq \{e_m(i)\}_{i=1}^d$ be the set of embedded points. Let $\cF_{b,\cX}$ denote the class of continuous functions $f\colon \RR^{m} \mapsto \RR$ such that $\norm{f}_{\infty,\cX} \leq \log  b $, which we will simply denote by $\cF_b$. 
\begin{assump}[Assumptions on neural  class]\label{Assump1} 
The neural class $\cF$ is a closed subset of $\cF_b$ and   satisfies 
\unskip
   \begin{subequations} \label{eq:coventcond}%
       \begin{align}
    \log   N (2\epsilon \log  b,\cF, \norm{\cdot}_{\infty,\cX}) \leq \bar \kappa(\cF,\epsilon) , \qquad \forall \epsilon \in (0,  1], \label{eq:coventropy}
    \end{align}
   where   $\bar \kappa(\cF,\epsilon) \geq 0$ is such that 
    \begin{align}
 \kappa(\cF)\coloneqq   \int_{0}^{1} \sqrt{\bar \kappa(\cF,\epsilon)} \  d\epsilon <\infty.\label{eq:finentrint}
  \end{align}
   \end{subequations}
 \end{assump}
\subsubsection{Parametrized Quantum Circuits}\label{Sec:PQC}
For a fixed enumeration $\mathbf{P}\coloneqq\{P_i^{\star}\}_{i=1}^d$ of the set of orthogonal rank-one projectors in the spectral decomposition of $H_{\alpha}^{\star}(\rho,\sigma)$ satisfying~\eqref{eq:opteigvalues} and~\eqref{eq:optrenmeasrelent}, let $U^{\star}(\mathbf{P})$ be the unitary such that $U^{\star}(\mathbf{P})\ket{i}\!\bra{i}\big(U^{\star}(\mathbf{P})\big)^{\dag}=P_i^{\star}$ for all $i$. 
For $\cS \subseteq \mathbb{S}_d \times \mathbb{S}_d$, let 
\begin{align}
\cU_{\alpha}^{\star}(\cS)\coloneqq \{U^{\star}(\mathbf{P}): \mathbf{P} \mbox{ is a given enumeration of  eigenprojectors of } H_{\alpha}^{\star}(\rho,\sigma)\mbox{ for }(\rho,\sigma) \in \cS  \}. \label{eq:unitcirccls}  
\end{align} 
Note that although the RHS of~\eqref{eq:unitcirccls} implicitly depends on the specific enumeration $\mathbf{P}$ of the eigen-projectors (which is not necessarily unique) of $H_{\alpha}^{\star}(\rho,\sigma)$ for $(\rho,\sigma) \in \cS$,   we suppress this dependence for simplicity. 
Let $\cU_{\alpha}(\delta,\cS)$  denote a $\delta$-covering of $\cU_{\alpha}^{\star}(\cS)$ in the operator norm.  
In other words, $\cU_{\alpha}(\delta,\cS)$  denotes a PQC composed of a (finite) set of  unitary operators,  parametrized by $\bm{\theta} \in  \Theta_{\alpha}(\delta,\cS)$  such that  for every $U^{\star} \in \cU_{\alpha}^{\star}(\cS)$, there exists $U(\bm{\theta})$ satisfying
\begin{align}
 \norm{U^{\star}-U(\bm{\theta})} \leq \delta. \label{eq:unitaryapprox}   \end{align}
 For all $\cS$, an application of the Solovay--Kitaev theorem (see Appendix~\ref{Sec:eff-imp-unit})  guarantees that such a $\cU_{\alpha}(\delta,\cS)$ exists. 
 Set $ \cU^{\star}(\cS) \coloneqq \cU_{1}^{\star}(\cS)$, $\cU(\delta,\cS) \coloneqq \cU_{1}(\delta,\cS)$ and $\Theta(\delta,\cS) \coloneqq \Theta_{1}(\delta,\cS)$. 
In practice, PQCs generated by commonly employed  ans\"{a}tze such as the hardware efficient ansatz\cite{kandala2017hardware}, alternating layered ansatz \cite{Cerezo2021CostDependent},    or alternating operator ansatz for quantum approximate optimization algorithm (QAOA) could be used for QNE, among which the latter is known to be universal for emulating a general quantum circuit given sufficient circuit length \cite{lloyd2018universalqaoa,morales2020universality}. 
Since the construction of the PQC is not our focus, we henceforth assume that an efficient parametrized quantum circuit architecture implementing $\cU_{\alpha}(\delta,\cS)$ is given.
 \section{Performance Guarantees for Quantum Neural Estimation}\label{Sec:minimaxriskbnds}
 In this section, we obtain upper bounds on the expected absolute error achieved by the QNE for estimating measured relative entropy, along with a sub-Gaussian  concentration inequality for this error. 
Before doing so, let us establish a few definitions. 
Let  $\Lambda_{\alpha}^{\star}(\rho,\sigma)$  be the set of eigenvalues of $H_{\alpha}^{\star}(\rho,\sigma)$, where $H_{\alpha}^{\star}(\rho,\sigma)$ is a Hermitian operator satisfying~\eqref{eq:opteigvalues} and~\eqref{eq:optrenmeasrelent}. For a given set $ \cS$ of density operators and a set $\cF$ of real-valued functions supported on $\cX$, set 
   \begin{align}
     \varepsilon_{\alpha}(\cF,\cS) \coloneqq \min_{f \in \cF} \sup_{(\rho,\sigma) \in \cS}\max_{\{\lambda_{i,\alpha}^{\star}\}_{i=1}^d }\max_{i \in [1:d]} \abs{f(e_m(i))-\lambda_{i,\alpha}^{\star}},\label{eq:approxcond} 
   \end{align} 
   where the inner maximization is over all possible enumerations $\{\lambda_{i,\alpha}^{\star}\}_{i=1}^d$ of $\Lambda_{\alpha}^{\star}(\rho,\sigma)$. Set $\varepsilon(\cF,\cS) \coloneqq   \varepsilon_{1}(\cF,\cS)$. These quantities provide a uniform upper bound  (over the class $\cS$) on the component of error arising due to neural network approximation of any enumeration of the eigenvalues of $H_{\alpha}^{\star}(\rho,\sigma)$, and will play a role in the performance guarantees for QNE given below. 
 \subsection{Measured Relative Entropy}
We start with our main result for measured relative entropy QNE.
\begin{theorem}[Performance guarantees for measured relative entropy]\label{Thm:minimaxrisk-QNE}
Consider a QNE for estimating $\sD_{\mathsf{M}}(\rho\Vert\sigma)$ with classical neural net $\cF \subseteq \cF_b$ satisfying Assumption~\ref{Assump1} and a PQC $\cU(\delta,\cS)$ parametrized by $\Theta(\delta,\cS)$ for $\cS \subseteq \cS_d(b)$.  Setting $\xi_{\cF,\cS,\delta,b}\coloneqq \varepsilon(\cF,\cS)(b+1)+ 2\delta   (b+\log  b )$, we have
\begin{align}
 &\sup_{(\rho,\sigma) \in \cS}   \EE\!\left[ \abs{\sD_{\mathsf{M}}(\rho\Vert\sigma) -\hat\sD_{\mathsf{M},\Theta(\delta,\cS),\cF}^n}\right]  \leq \xi_{\cF,\cS,\delta,b}+48(b+1) (\log  b) \mspace{2 mu} \kappa(\cF)\mspace{1 mu} n^{-\frac 12},\label{eq:minimax risk}
\end{align}
where $\kappa(\cF)$  and $\varepsilon(\cF,\cS)$ are as given in~\eqref{eq:finentrint} and~\eqref{eq:approxcond}, respectively. Moreover, for all $z \geq 0$, we have
\begin{align}
 & \sup_{(\rho,\sigma) \in \cS}    \mathbb{P}\left( \abs{\sD_{\mathsf{M}}(\rho\Vert\sigma) -\hat\sD_{\mathsf{M},\Theta(\delta,\cS),\cF}^n} \mspace{-2 mu} \geq    \xi_{\cF,\cS,\delta,b}+128(b+1) (\log b) \big(n^{-\frac 12} \kappa(\cF)  + z \big) \right) \mspace{-4 mu} \leq \mspace{-2 mu} 2e^{-n  z^2}. \label{eq:tailineQNEmrel}
\end{align}  
\end{theorem}
The proof of Theorem~\ref{Thm:minimaxrisk-QNE} is given in Section~\ref{Sec:Thm:minimaxrisk-QNE-proof}, and relies on upper bounding the (absolute)  error  as the sum of approximation and statistical components, i.e., $\mathbb{E}[\text{error}]\leq \text{approx. err} + \mathbb{E}[\text{stat. err}]$.  This is enabled by the compatible structure of the variational form of measured relative entropy (see \eqref{eq:renyi_rel_ent_variational}) and its QNE (see \eqref{eq:QNE-inner_opt}). The approximation error arises as a consequence of parameterizing the Hermitian operator in the variational form using a hybrid model composed of PQCs and neural nets, while the statistical error is due to replacing expectations therein  by the corresponding empirical averages with respect to measurement outcomes. The approximation error can be made arbitrary small by choosing a sufficiently large  hybrid model, thanks to the universal approximation property of suitable neural nets and PQCs for approximating continuous functions and unitaries, respectively.

Choosing an expressive model to reduce the approximation error, however, comes at the cost of increased statistical  complexity of the model. The associated statistical error can be bounded in terms of the suprema of  separable sub-Gaussian processes indexed by the neural class and quantum circuit parameters, which is a key object of interest in empirical process theory and can be analyzed using the tools therein.   In particular, the expected statistical error can be upper bounded in terms of a covering entropy integral \cite[Corollary 5.25]{VanHandel-book} with the  random fluctuations around the latter exhibiting sub-Gaussian concentration \cite[Theorem 5.29]{VanHandel-book} (see Theorem \ref{thm:tailineq} below). Leveraging the aforementioned results, Theorem \ref{Thm:minimaxrisk-QNE} quantifies the approximation and statistical errors for a given neural net $\cF$  satisfying Assumption \ref{Assump1} and PQC $\cU(\delta,\cS)$.   In \eqref{eq:minimax risk},  the  term $ \xi_{\cF,\cS,\delta,b}$ corresponds to the approximation error component, with $\varepsilon(\cF,\cS)$ and $\delta$ denoting the neural net and PQC approximation parameter, respectively. The second term  in \eqref{eq:minimax risk}, which decays at rate $n^{-\frac 12}$, corresponds to the statistical  component. 
Below, we will  instantiate Theorem \ref{Thm:minimaxrisk-QNE} to shallow (Corollary \ref{Cor-shallownn}) and deep  (Proposition \ref{Prop:deepnn}) nets, where we explicitly quantify $\xi_{\cF,\cS,\delta,b}$ and $\kappa(\cF)$ in terms of neural net parameters.
\begin{remark}[Effective copy complexity]
    Since each inner loop of the QNE consumes $n$ copies of both $\rho$ and $\sigma$, the effective copy complexity of the algorithm that achieves the expected absolute error bound given in Theorem~\ref{Thm:minimaxrisk-QNE} is $O\big(n|\Theta(\delta,\cS)|\big)$. Given this, it is of interest to determine or characterize the classes  $\cS$ for which $|\Theta(\delta,\cS)|$ scales benignly as a function of $\delta$. However, this turns out to be a non-trivial problem, closely related to that of characterizing the class of unitaries which can be efficiently approximated. 
Nevertheless, in Section~\ref{Sec:copycomp} below, we will consider certain classes of  states exhibiting symmetries for which the copy complexity scales efficiently. 
\end{remark}
\medskip
Next, we specialize Theorem~\ref{Thm:minimaxrisk-QNE} by considering specific embeddings and neural classes. We first define the neural class of interest. 
\begin{definition}[Neural network class]\label{Def:NNclass}
    Let $2 \leq L \in \NN$,  $k_0=m$, $k_1, \ldots, k_{L-1} \in \NN$,  and $ k_L=1 $.  
    A neural network with $L$ layers and activation~$\varphi\colon \RR \mapsto \RR$ is a map $f\colon \RR^{k_0} \mapsto \RR$ 
    given by 
    \begin{align}
        f\coloneqq \begin{cases}
            W_2\circ \varphi \circ W_1, &\quad L=2, \\
            W_L \circ \varphi \circ W_{L-1} \circ \varphi \circ \cdots \circ \varphi \circ W_1, & \quad L \geq 3,
        \end{cases}
    \end{align}
    where for $\ell \in \{1,\ldots,L\}$, $W_{\ell}\colon \RR^{k_{\ell}} \mapsto \RR^{k_{\ell-1}}$, $W_{\ell}x\coloneqq A_{\ell}x+b_{\ell}$ are affine maps with matrices $A_{\ell} \in \RR^{k_{\ell} \times k_{\ell-1}}$ and bias vectors $b_{\ell} \in \RR^{k_{\ell}}$, and the activation $\varphi$ acts component-wise as $\varphi(x_1,\ldots,x_k)=(\varphi(x_1),\ldots, \varphi(x_k))$. For $M \in \RR_{\geq 0}$, let $\cF^{\mathrm{NN}}(L,K,M,m,\varphi)$ be the set of all such  $f$ with  $\max_{0 \leq \ell \leq L} k_{\ell} \leq K$ and  $\max_{1 \leq \ell \leq L} \norm{A_{\ell}}_{\max} \vee  \norm{b_{\ell}}_{\max} \leq M$, where $\norm{A}_{\max}\coloneqq \max_{i,j}\abs{A_{i,j}}$ is the maximum of the absolute values of the entries of $A$.  Henceforth, we take $\varphi$ as the ReLU function $\varphi_{\mathrm{R}}(x)\coloneqq x \vee 0$ unless specified otherwise,  and  set $\cF^{\mathrm{NN}}(L,K,M,m)\coloneqq \cF^{\mathrm{NN}}(L,K,M,m,\varphi_{\mathrm{R}})$. 
\end{definition}
\subsubsection*{Performance Guarantees with Shallow Nets}
We start by treating the shallow neural class, $\cF^{\mathrm{SNN}}(K,M,m) \coloneqq \cF^{\mathrm{NN}}(2,K,M,m)$,
with the canonical embedding $e_d\colon  i \mapsto \ket{i}$, where $\ket{i}$ is a $d$-dimensional vector with $1$ in the $i^{\mathrm{th}}$ position and $0$ elsewhere. Consider the following shallow neural class with $d $ neurons in the hidden layer:
    \begin{align}
    \tilde \cF^{\mathrm{SNN}}\mspace{-2 mu}(d,M)&\mspace{-2 mu}\coloneqq \mspace{-2 mu} \left\{f(x)\mspace{-2 mu}=\mspace{-2 mu}\sum_{j=1}^d \beta_j \varphi_{\mathrm{R}}\big( \langle j, x \rangle \big), ~x \in \RR^d:   \abs{\beta_j} \leq M, ~ j \in [1:d] \right\}\mspace{-2 mu} \subseteq \mspace{-2 mu}\cF^{\mathrm{SNN}}(d,M,d). \notag
\end{align}
The next corollary  bounds the expected absolute error achieved by the QNE with such a neural net. 
\begin{cor}[Performance guarantees with shallow nets] \label{Cor-shallownn}
The following hold for $\cS\subseteq \cS_d(b)$ with $a (\delta,b)\coloneqq 2\delta   (b+\log  b )$:
\begin{subequations}
\begin{align}
& \sup_{(\rho,\sigma) \in \cS}   \EE\!\left[ \abs{\sD_{\mathsf{M}}(\rho\Vert\sigma) -\hat\sD_{\mathsf{M},\Theta(\delta,\cS),\tilde \cF^{\mathrm{SNN}}(d,\log  b )}^n}\right] \leq  a (\delta,b) +96(\log  b) (b+1)   d^{\frac 12}n^{-\frac 12}, \label{eq:expubnd-shallownet} \\ 
&\sup_{(\rho,\sigma) \in \cS}  \mathbb{P}\left(\abs{\sD_{\mathsf{M}}(\rho\Vert\sigma) -\hat\sD_{\mathsf{M},\Theta(\delta,\cS),\tilde \cF^{\mathrm{SNN}}(d,\log  b )}^n}\geq  a (\delta,b) +128(\log  b) (b+1) \big(  2d^{\frac 12}n^{-\frac 12}  + z \big)\right) \notag \\
& \qquad \qquad\qquad\qquad\qquad\qquad\qquad\qquad\qquad\qquad\qquad\qquad\leq 2e^{-n  z^2}, ~~~\forall~z \geq 0.\label{eq:devineq-shallownet}
\end{align}
\end{subequations}
\end{cor}
The proof of Corollary~\ref{Cor-shallownn} is given in Section~\ref{Sec:Cor-shallownn-proof}, and it follows from Theorem~\ref{Thm:minimaxrisk-QNE} by showing that $\varepsilon\big(\tilde \cF^{\mathrm{SNN}}(d,\log b),\cS_d(b)\big)=0$ and  $\kappa\big(\tilde \cF^{\mathrm{SNN}}(d,\log b)\big) \leq 2 \sqrt{d}$. 
From \eqref{eq:expubnd-shallownet}, we observe  that given $n$ copies of the quantum states $\rho,\sigma$, the expected absolute error of the QNE for estimating  $\sD_{\mathsf{M}}(\rho\Vert\sigma)$ over  $(\rho,\sigma) \in \cS_d(b)$ with neural class $\cF^{\mathrm{SNN}}(d,\log  b )$ and quantum circuit $\cU(\delta,\cS)$ is $O\big(b(\log  b +1)(\delta+D^{1/2}n^{-1/2})\big)$, where $D \coloneqq d\,\Theta(\delta,\cS_d(b))$.
\begin{remark}[Lower bound on minimax risk]
   It is instructive to compare the expected absolute error of QNE with the following lower bound on the expected minimax risk for estimating classical relative entropy which holds 
when    $d b \lesssim n \log d\,$ and $b \geq (\log d)^2 \vee 2$: 
\begin{align}
R_n^{\star}\big(\cS_d(b)\big) \gtrsim \frac{\sqrt{b}}{\sqrt{n}}+\frac{d b}{n \log d}.\label{eq:lwrbndminmaxrisk}
\end{align}
This bound can be distilled from the results in \cite{Yihong-Yang-2016} and \cite{Yanjun-2020-alphadiv} by slightly modifying the proofs therein to account for the class $\cS_d(b)$ considered herein (see Appendix \ref{App:lowerbndest}).  From \eqref{eq:expubnd-shallownet} and \eqref{eq:lwrbndminmaxrisk}, it thus follows that the expected minimax risk over $\cS_d(b)$ has a  parametric rate  of convergence in $n$ (independent of $d$), and moreover that it is achieved by the QNE with a shallow net. However, there is a mismatch in the dependence on $d$ in \eqref{eq:expubnd-shallownet} and \eqref{eq:lwrbndminmaxrisk} as the second term in the lower bound scales as $\Omega_b(d \log d/n)$ for $d\log d/n \lesssim 1$ whereas the upper bound is $O_b(d^{1/2}n^{-1/2})$ with $\delta=0$ (applicable  in the classical setting). That said, this mismatch is not unexpected as achieving the optimal dependence on $d$ requires explicit bias-correction mechanisms even in the classical case, which is not performed by the QNE. 
\end{remark}
\subsubsection*{Performance Guarantees with Deep Nets}
We next derive performance guarantees for QNEs realized by deep neural nets.
For this purpose, we consider the  one-dimensional embedding into $\RR$ given by $e_1\colon i \mapsto i/d$. Let 
\begin{align}
 \cP_k(a)  \coloneqq  \left\{p(x)=\sum_{i=0}^k a_i x^i, ~ x \in \RR:  \norm{\bm{a}}_{\infty} \leq a \right\},  \notag
\end{align}
denote the class of univariate  polynomials of degree at most $k$ and magnitude of coefficients, $\bm{a} \coloneqq (a_0,\ldots,a_k)$, bounded by $a$. 
We define the following class of density operator pairs for  $\alpha>0$:
\begin{align}
   \tilde{\cS}_{k,d,\alpha}(b,a, \varepsilon)& \coloneqq \left\{(\rho,\sigma) \in \cS_d(b):\begin{aligned}
       & \mbox{ for any enumeration } \Lambda_{\alpha}^{\star}(\rho,\sigma)=\{\lambda_{i,\alpha}^{\star}\}_{i=1}^d, \\& \exists~p \in  \cP_k(a)\mbox{ s.t. }  \max_{i \in [1:d]}\abs{p\big(e_1(i)\big)-\lambda_{i,\alpha}^{\star}}\leq \varepsilon
\end{aligned}\right\}.\label{eq:densityclass}
\end{align}
In other words, $ \tilde{\cS}_{k,d,\alpha}(b,a,\varepsilon)$ is the class of density operator pairs $(\rho,\sigma) \in \cS_d(b)$ such that any ordering of the eigenvalues $\Lambda_{\alpha}^{\star}(\rho,\sigma)$  of the optimizer $H_{\alpha}^{\star}(\rho,\sigma)$,  embedded into the interval $[0,1]$ in a uniformly spaced manner, can be approximated within $\varepsilon$ precision by a real polynomial of degree at most $k$ and absolute value of coefficients bounded by $a$.

While $\tilde{\cS}_{k,d,\alpha}(b,a, \varepsilon) \subseteq \cS_d(b)$ by definition, Proposition~\ref{Prop:densityopclass} below shows that $  \tilde{\cS}_{k,d,\alpha}(b,a,\varepsilon)=\cS_d(b)$ for  $k$, $a$, and $\varepsilon$ sufficiently large.
\begin{prop}[Relation between classes of density operators]\label{Prop:densityopclass}
The following hold for all $\alpha >0$:
\begin{enumerate}[(i)]
    \item     For all $k \geq d-1$, $a \geq \big((3e)^d \log  b \big)/(2\pi)$, and $\varepsilon \geq 0$,
    \[
    \tilde{\cS}_{k,d,\alpha}(b,a,\varepsilon)=\cS_d(b).
    \]
    \item For all $k \in \NN$, $a \geq 2^k k! \log  b $, and $\varepsilon \geq \left(2d+0.5\right) k^{-\frac 13} \log  b  $, \[
    \tilde{\cS}_{k,d,\alpha}(b,a,\varepsilon)= \cS_d(b).
    \]
\end{enumerate}
 
\end{prop}
 See Appendix~\ref{Prop:densityopclass-proof} for a proof. For proving Proposition \ref{Prop:densityopclass} $(i)$, given  $\{\lambda_i\}_{i=1}^d$ such that  $\max_{1 \leq i \leq d} \abs{\lambda_i} \leq \log b$, we simply construct a polynomial $p$ of degree $d-1$ such that $p(i/d)=\lambda_i$ for $i \in [1:d]$ and obtain an upper bound on its coefficients. For a proof of Part $(ii)$, we consider an intermediate Lipschitz continuous function $\hat f$ such that $\hat f(i/d)=\lambda_i$ and quantify the error of approximating this function via Bernstein polynomials of degree at most $k$ whose coefficients are bounded as given therein.  
The next proposition provides an upper bound on the expected absolute error (along with a deviation inequality) for estimating the measured relative entropy,   achieved by a QNE that employs a fixed width deep neural net with appropriately chosen parameters. Let 
 \begin{align}
        \tilde \cF^{\mathrm{NN}}(L,K,M,1)\mspace{-2 mu}\coloneqq \mspace{-2 mu}\left\{ \begin{aligned}
          \mspace{-2 mu}  \tilde f(x)=~&f(x) \ind_{\abs{f(x)} \leq \log  b  }+(\log  b)\big(   \ind_{f(x)> \log  b  }-   \ind_{f(x)< -\log  b  } \big)  \\& : f \in  \cF^{\mathrm{NN}}(L,K,M,1)
    \end{aligned}\right\}. \label{eq:truncnnclass}
 \end{align}
That is, $\tilde \cF^{\mathrm{NN}}(L,K,M,1)$ is the neural network class $\cF^{\mathrm{NN}}(L,K,M,1)$ (see  Definition~\ref{Def:NNclass}) with the function values saturated at $\pm \log  b $. Set $ \tilde{\cS}_{k,d}(b,a,\varepsilon) \coloneqq  \tilde{\cS}_{k,d,1}(b,a,\varepsilon)$.
\begin{prop}[Performance guarantees with deep nets] \label{Prop:deepnn} 
There exists a constant $c>0$ such that setting
\begin{subequations}
   \begin{align}\label{eq:deepnnbndparams}
 L_{\varepsilon,k,a} & \coloneqq c k\left(\log  (\varepsilon^{-1})+ \log  k+\log  a \right),  \\
   B_{\varepsilon,k,a,b}& \coloneqq 24 L_{\varepsilon,k,a} \log b +13\sqrt{L_{\varepsilon,k,a} \log b}, 
\end{align} 
\end{subequations}
$\tilde \cF^{\mathrm{NN}}(L_{\varepsilon,k,a})\equiv\tilde \cF^{\mathrm{NN}}(L_{\varepsilon,k,a},9,1,1)$ and $\cS \subseteq \tilde{\cS}_{k,d}(b,a, \varepsilon)$, the following hold for all $z \geq 0$:
\begin{samepage}
\begin{align}
&\sup_{\substack{(\rho,\sigma)  \in  \cS}}  \EE\!\left[ \abs{\sD_{\mathsf{M}}(\rho\Vert\sigma) -\hat\sD_{\mathsf{M},\Theta(\delta,\cS), \tilde \cF^{\mathrm{NN}}(L_{\varepsilon,k,a})}^n}\right] \notag \\
& \qquad \qquad \qquad \qquad \qquad \qquad  \leq 2(b+1) \varepsilon+ 2\delta   (b+\log  b )+48(b+1)  \mspace{1 mu} B_{\varepsilon,k,a,b}\mspace{1 mu} n^{-\frac 12}. \label{eq:deepnnminriskgen} \\
 & \sup_{\substack{(\rho,\sigma)  \in \mspace{2 mu}\cS}}    \mathbb{P}\bigg( \abs{\sD_{\mathsf{M}}(\rho\Vert\sigma) -\hat\sD_{\mathsf{M},\Theta(\delta,\cS), \tilde \cF^{\mathrm{NN}}(L_{\varepsilon,k,a})}^n} \geq  2(b+1) \varepsilon+ 2\delta   (b+\log  b )  \notag \\
 & \qquad \qquad\qquad\qquad \qquad \qquad \qquad   + 128(b+1)  \mspace{1 mu} \big(B_{\varepsilon,k,a,b}\mspace{1 mu} n^{-\frac 12} +z \log b \big)\bigg) \leq 2 e^{-n  z^2}. \label{eq:tailineqdeepnn} 
\end{align} 
\end{samepage}
\end{prop}
The proof of Proposition~\ref{Prop:deepnn} (see Section~\ref{Sec:Prop:deepnn-proof}) follows from Theorem~\ref{Thm:minimaxrisk-QNE} by using the neural  class $\tilde \cF^{\mathrm{NN}}(L_{\varepsilon,k,a})$, which is rich enough to approximate polynomials of degree $k$ and coefficients uniformly bounded by $a$, up to an error of $\varepsilon$, thanks to  the results in~\cite{Dennis-2021}. The unspecified constant $c$ in the definition of $L_{\varepsilon,k,a}$ is the same one that appears in \cite[Proposition~III.5]{Dennis-2021}. Leveraging this approximation result and upper bounding $\kappa\big(\tilde \cF^{\mathrm{NN}}(L_{\varepsilon,k,a})\big)$ leads to the claim in~\eqref{eq:tailineqdeepnn}. The advantage of Proposition~\ref{Prop:deepnn} compared to Corollary~\ref{Cor-shallownn} is that the bounds are independent of the dimension $d$, albeit over a sub-class, $\tilde{\cS}_{k,d}(b,a, \varepsilon)$,  of $\cS_d(b)$.  
\subsection{Measured R\'{e}nyi Relative  entropy} 
The following analogue of Theorem~\ref{Thm:minimaxrisk-QNE} holds for the measured R\'{e}nyi relative  entropy.

\begin{theorem}[Performance guarantees for measured R\'{e}nyi relative entropy] \label{Thm:minimaxrisk-QNE-Renyi}
Let $\alpha \in (0,1) \cup (1,\infty)$.  Consider a QNE for estimating $\sD_{\mathsf{M,\alpha}}(\rho\Vert\sigma)$ with classical neural net $\cF \subseteq \cF_b$ satisfying Assumption~\ref{Assump1} and a PQC $\cU_{\alpha}(\delta,\cS)$ parametrized by $\Theta_{\alpha}(\delta,\cS)$. Setting 
\begin{align}
 \xi_{\cF,\cS,\delta,b,\alpha} \coloneqq \alpha \left(b^{\abs{\alpha-1}}+b^{\alpha}\right) \varepsilon_{\alpha}(\cF,\cS) +2\left(\frac{\alpha}{\abs{\alpha-1}} b^{\abs{\alpha-1}}+b^{\alpha}\right)\delta, \notag 
\end{align}
we have for $\cS \subseteq \cS_d(b)$ and $z \geq 0$ that
\begin{align}
& \sup_{(\rho,\sigma) \in \cS}   \EE\!\left[ \abs{\sD_{\mathsf{M},\alpha}(\rho\Vert\sigma) -\hat{\sD}_{\mathsf{M},\alpha,\Theta_{\alpha}(\delta,\cS), \cF}^n}\right]  \leq  \xi_{\cF,\cS,\delta,b,\alpha}  + 48\alpha \big(b^{\alpha}+b^{\abs{\alpha-1}}\big)(\log  b) \mspace{2 mu}  \kappa(\cF) n^{-\frac 12},\label{eq:minimax risk:Renyi} \\
&  \sup_{(\rho,\sigma) \in \cS}   \PP\bigg( \abs{\sD_{\mathsf{M},\alpha}(\rho\Vert\sigma) -\hat{\sD}_{\mathsf{M},\alpha,\Theta_{\alpha}(\delta,\cS), \cF}^n} \geq  \xi_{\cF,\cS,\delta,b,\alpha}  + 128\alpha \big(b^{\alpha}+b^{\abs{\alpha-1}}\big)(\log  b) \mspace{2 mu} \notag \\
& \qquad \qquad  \qquad \qquad \qquad \qquad \qquad \qquad \qquad \qquad \qquad \times \big( \kappa(\cF) n^{-\frac 12}+z\big) \bigg) \leq 2 e^{-nz^2}. \label{eq:tailineqmeasrenyi}
\end{align}
\end{theorem}
The proof of Theorem~\ref{Thm:minimaxrisk-QNE-Renyi} is analogous to that of Theorem~\ref{Thm:minimaxrisk-QNE}, and is given in Section~\ref{Sec:Thm:minimaxrisk-QNE-Renyi-proof} below.

\medskip

Specializing Theorem~\ref{Thm:minimaxrisk-QNE-Renyi} to a QNE employing the shallow neural net used in Corollary~\ref{Cor-shallownn}, we obtain the  result stated below (proof is omitted as it is similar to Corollary~\ref{Cor-shallownn}). 
\begin{cor}[Performance guarantees with shallow neural net] \label{Cor-shallownn-Renyi}
Let $\alpha \in (0,1) \cup (1,\infty)$. The following hold with $\cS \subseteq \cS_d(b)$ for all $z \geq 0$:
\begin{subequations}
\begin{align}
 \sup_{(\rho,\sigma) \in \cS}  & \EE\!\left[ \abs{\sD_{\mathsf{M},\alpha}(\rho\Vert\sigma) -\hat\sD_{\mathsf{M},\alpha,\Theta_{\alpha}(\delta,\cS),\tilde \cF^{\mathrm{SNN}}(d,\log  b )}^n}\right] \leq 2 \left(\frac{\alpha}{\abs{\alpha-1}} b^{\abs{\alpha-1}}+b^{\alpha}\right)\delta  \notag \\
 &\qquad \qquad \qquad \qquad \qquad \qquad \qquad \qquad \qquad +96\alpha \big(b^{\alpha}+b^{\abs{\alpha-1}}\big) (\log  b)  \mspace{1 mu}  d^{\frac 12} n^{-\frac 12}, \label{eq:expabserrren-shallow-net} \\
  \sup_{(\rho,\sigma) \in \cS}   &\PP\Bigg(  \abs{\sD_{\mathsf{M},\alpha}(\rho\Vert\sigma) -\hat\sD_{\mathsf{M},\alpha,\Theta_{\alpha}(\delta,\cS),\tilde \cF^{\mathrm{SNN}}(d,\log  b )}^n}  \geq 2\left(\frac{\alpha}{\abs{\alpha-1}} b^{\abs{\alpha-1}}+b^{\alpha}\right)\delta  \notag \\
 &\qquad \qquad \qquad   \qquad \qquad +128\alpha \big(b^{\alpha}+b^{\abs{\alpha-1}} \big) (\log  b)  \big( 2d^{\frac 12} n^{-\frac 12}+z\big)\Bigg) \leq 2e^{-nz^2}. 
\end{align}
\end{subequations}
\end{cor}
To benchmark the above performance guarantees, one may compare it with the following lower bound on the minimax  which holds when $d b^{\alpha} \lesssim n \log d\,$,  $b \geq (\log d)^2 \vee 2$ and $\alpha >1$: 
\begin{align}
R_{n,\alpha}^{\star}\big(\cS_d(b)\big) \gtrsim_{\alpha} \frac{\sqrt{b}}{\sqrt{n}}+\frac{d b}{n \log d}.  \label{eq:lwrbndminmaxrisk-ren} 
\end{align}
This claim (see Appendix \ref{App:lowerbndest} for a proof) follows from the lower bound on the minimax risk for classical $\alpha$-divergence estimation obtained in \cite{Yanjun-2020-alphadiv},  by tweaking the proof therein to account for the class $\cS_d(b)$ and noting that $R_{n,\alpha}^{\star}\big(\cS_d(b)\big)$ can be lower bounded in terms of this minimax risk  (up to multiplicative constants which depend on $\alpha$ and $b$). 
On the other hand, given $n$ copies of the quantum states $\rho,\sigma$, the expected absolute error achieved by QNE for estimating $\sD_{\mathsf{M},\alpha}(\rho\Vert\sigma)$ as quantified in \eqref{eq:expabserrren-shallow-net} is $O_{\alpha,b}\big(D^{1/2}n^{-1/2}\big)$, where $D=\Theta_{\alpha}\big(\delta,\cS_d(b)\big)\,d$.  
Comparing these two bounds, it follows that the QNE achieves the  parametric rate of convergence in $n$, meaning it is minimax optimal for estimating measured R\'{e}nyi relative  entropy over $\cS_d(b)$.  

\medskip
The next proposition also follows from Theorem~\ref{Thm:minimaxrisk-QNE-Renyi} by specializing to a QNE realized by the deep neural net from Proposition~\ref{Prop:deepnn}.

\begin{prop}[Performance guarantees with deep neural net] \label{Prop:deepnn-Renyi}  Let $\alpha \in (0,1) \cup (1,\infty)$, let $\tilde{\cS}_{k,d,\alpha}(b,a,\varepsilon)$ be the class of density operator pairs defined in~\eqref{eq:densityclass}, and let $\tilde \cF^{\mathrm{NN}}(L,K,M,1)$ be the neural class given in~\eqref{eq:truncnnclass}. With $L_{\varepsilon,k,a}$, $B_{\varepsilon,k,a,b} $ as defined in~\eqref{eq:deepnnbndparams} and $\tilde \cF^{\mathrm{NN}}(L_{\varepsilon,k,a})\equiv \tilde \cF^{\mathrm{NN}}(L_{\varepsilon,k,a},9,1,1)$,  the following holds for all $\cS \subseteq \tilde{\cS}_{k,d,\alpha}(b,a,\varepsilon)$ and $z \geq 0$:
\begin{align}
&\sup_{(\rho,\sigma) \in   \cS}   \EE\!\left[ \abs{\sD_{\mathsf{M},\alpha}(\rho\Vert\sigma) -\hat\sD_{\mathsf{M},\alpha,\Theta_{\alpha}(\delta,\cS), \tilde \cF^{\mathrm{NN}}(L_{\varepsilon,k,a})}^n}\right]   \leq 2 \alpha \big(b^{\alpha}+b^{\abs{\alpha-1}}\big) \big(\varepsilon+24 B_{\varepsilon,k,a,b}n^{-\frac 12}\big)\notag \\
& \qquad \qquad\qquad\qquad\qquad\qquad\qquad\qquad\qquad\qquad\qquad\qquad \quad +2\left(\frac{\alpha}{\abs{\alpha-1}} b^{\abs{\alpha-1}}+b^{\alpha}\right)\delta, \notag \\
  &  \sup_{\substack{(\rho,\sigma) \in \cS }} \mspace{-4 mu}  \PP\Bigg( \mspace{-2 mu}\abs{\sD_{\mathsf{M},\alpha}(\rho\Vert\sigma) -\hat\sD_{\mathsf{M},\alpha,\Theta_{\alpha}(\delta,\cS), \tilde \cF^{\mathrm{NN}}(L_{\varepsilon,k,a})}^n} \geq   2\mspace{-2 mu}\left(\frac{\alpha}{\abs{\alpha-1}} b^{\abs{\alpha-1}}+b^{\alpha}\right)\delta   \notag \\
  & \qquad \qquad\qquad\qquad\qquad \quad  +2\alpha \Big(b^{\alpha}+b^{\abs{\alpha-1}}\Big) \left(\varepsilon+ 64 B_{\varepsilon,k,a,b} n^{-\frac 12}+64 z \log b\right)\mspace{-6 mu}\Bigg) \leq 2 e^{-nz^2}.  \notag
\end{align}  
\end{prop}
The proof of Proposition~\ref{Prop:deepnn-Renyi} follows similarly to the proof of Proposition~\ref{Prop:deepnn}, and hence is omitted.
\section{Copy Complexity for Quantum Neural Estimation of Entropies}
\label{Sec:copycomp}
Here, we examine the copy complexity of QNE, focusing on measured relative entropy. The case of measured R\'{e}nyi relative entropy is similar.  As mentioned earlier, the number of copies of $\rho$ and $\sigma$ consumed by QNE in the worst case  scales proportionally to $|\Theta(\delta,\cS)|$. Hence, it is of interest to understand how the number of parameters of efficient PQCs scales as $\delta$ shrinks, since this dependence directly affects the copy complexity. To get more insight into this aspect, consider the setting where   
 $d=2^N$, with $N$ denoting  the number of qubits. An application of the Solovay--Kitaev theorem (see  Appendix~\ref{Sec:eff-imp-unit}) 
shows that, for all $\cS \subseteq \mathbb{S}_d \times \mathbb{S}_d$,  there exists a PQC $\cU(\delta,\cS)$ with $|\Theta(\delta,\cS)| = O\mspace{-1 mu}\Big(\mspace{-1 mu}\big(N^2 4^N \delta^{-1}\big)^{N^24^N}\Big)$. The super-exponential dependence of $|\Theta(\delta,\cS)|$ on $N$ is a consequence of the requirement that the quantum circuit be able to universally approximate  the whole unitary group uniformly. Moreover, when $\cS=\mathbb{S}_d \times \mathbb{S}_d$, a  lower bound based on a covering argument shows that such a dependence on $N$ is unavoidable (see, e.g.,~\cite[Section~4.5.4]{Nielsen_Chuang_2010}).

Given that an exponential dependence on $N$ is undesirable for practically realistic algorithms, it is of interest to determine conditions on $\cS$ under which $|\Theta(\delta,\cS)|$ or the QNE copy complexity scales more efficiently, say polynomially or poly-logarithmically in $N$. The aforementioned determination is challenging for two reasons. First, the projectors  $\{P_i^{\star}\}_{i=1}^d$ appearing in the spectral decomposition of $H^{\star}(\rho,\sigma)$, and consequently in the definition of $\cU(\delta,\cS)$, do not have a closed-form expression. 
Second, even if such a characterization is available, determining the sets $\cS$ for which~\eqref{eq:unitaryapprox} holds with an efficiently realizable PQC is known to be a hard problem. In fact, it is closely related to the problem of characterizing the classes of unitaries which can be efficiently approximated. Given this, a pertinent approach is to fix a PQC $\cU$, as is the case in practice, and consider the classes $\cS$ for which efficient QNEs exist.

To this end, let us recall the necessary and sufficient condition that the optimizer $H^{\star}(\rho,\sigma)= \log \omega^{\star}(\rho,\sigma)$ should satisfy, which can be obtained as the limit $\alpha \to 1$ of~\eqref{eq:optimal-omega-alpha-gen} (see Appendix~\ref{Sec:optmeasrelent}):
\begin{align}
\sigma & =\int_{0}^{\infty}\ \left(\omega^{\star}(\rho,\sigma)+sI\right)^{-1}\rho\left(\omega^{\star}(\rho,\sigma)+sI\right)^{-1}ds. \notag
\end{align}
Solving this equation analytically to determine $\omega^{\star}(\rho,\sigma)$ appears challenging. Nevertheless, we can leverage it to determine a  class of density operator pairs over which our copy complexity results hold.   For a given $\rho \in \mathbb{S}_d$,  unitary $U$ on $\HH_d$ and $b \geq 1$, let 
\begin{align}
\Omega\big(b,\rho, U\big) &\coloneqq \left\{ \omega =\sum_{i=1}^d\frac{\operatorname{Tr}\!\left[ U |i\rangle\!\langle
i|U^{\dag}\rho\right]}{v_i} U |i\rangle\!\langle
i|U^{\dag} : \begin{aligned}
   & \mathbf{v} \in (0,1]^d, ~\sum_{i=1}^d v_i=1 \mbox{ and } \\& \frac{1}{b} \leq \frac{\operatorname{Tr}\!\left[ U |i\rangle\!\langle
i|U^{\dag}\rho\right]}{v_i} \leq b, ~\forall~i \in [1:d]
\end{aligned}\right\}. \notag
\end{align}
Setting $\sigma(\rho, \omega) \coloneqq \int_{0}^{\infty}\ \left(\omega+sI\right)^{-1}\rho\left(\omega+sI\right)^{-1}ds$ for $\omega \in \Omega\big(b,\rho, U\big)$, it is easy to see that $\sigma(\rho, \omega) \in \mathbb{S}_d$ since $\rho \mapsto \sigma(\rho, \omega)$ is a completely positive map and $\operatorname{Tr}\!\left[\sigma(\rho, \omega)\right]= \operatorname{Tr}\!\left[\rho \omega^{-1}\right]=\sum_{i=1}^d v_i=1$. 
Hence, for a given PQC $\cU$, define 
\begin{align}
 & \bar{\cS}_d(\delta, b,\cU) \coloneqq \left\{(\rho,\sigma(\rho,\omega)) \in \cS_d(b): \begin{aligned}
     &  \omega \in \Omega\big(b,\rho, U\big) \mbox{ for }  U \mbox{ s.t. } \exists~ U(\bm{\theta}) \in \cU\\
     &  \mbox{ satisfying } \norm{U-U(\bm{\theta})} \leq \delta
 \end{aligned}  \right\}, \label{eq:density-pair-copycomp} 
\end{align}
i.e., $\bar{\cS}_d(\delta, b,\cU)$ is the class of density operator pairs within $\cS_d(b)$ such that the eigenbasis of the optimizer $H^{\star}(\rho,\sigma)$ (or $\omega^{\star}(\rho,\sigma)$) can be approximated by the PQC $\cU$ within an accuracy of $\delta$.  Given a PQC $\cU$, $\bar{\cS}_d(\delta, b,\cU)$ can be determined in principle by computing  $(\rho,\sigma(\rho,\omega))$ pairs over $\omega \in \Omega\big(b,\rho, U\big)$ as in~\eqref{eq:density-pair-copycomp} and ascertaining if the pair lies within $ \cS_d(b)$.

The next result provides a  copy complexity bound for the QNE with a shallow net and a PQC $\cU$ with parameter set $\Theta(\cU)$. 
\begin{prop}[Copy complexity of QNE]\label{Prop:copycomp-QNE}
 Let $\epsilon \geq 0$ and $b \geq 1$. With probability at least $p \in [0,1)$, QNE with shallow  net $\tilde \cF^{\mathrm{SNN}}(d,\log  b)$ and  PQC $\cU$ estimates $\sD_{\mathsf{M}}(\rho\Vert\sigma)$  between $(\rho,\sigma) \in \bar{\cS}_d(\epsilon_b,b,\cU)$ within an error of $\epsilon$ given $O\big(c_{b,p}^2d|\Theta(\cU)|\epsilon^{-2}\big)$
copies each of $\rho$ and $\sigma$, where $\epsilon_b=\epsilon/(4b+4\log b)$ and  $c_{b,p}\coloneqq \big(1+\sqrt{-\log\! \left(1-p\right)}\big)(b+1)\log b.$
\end{prop}
The proof of Proposition~\ref{Prop:copycomp-QNE}  follows immediately from Corollary~\ref{Cor-shallownn} and can be found in  Section~\ref{Sec:Prop:effqucirccomp-proof}.  An analogous claim for QNE with a deep net and a quantum circuit $\cU$ can be 
obtained by using Proposition~\ref{Prop:deepnn} in place of Corollary~\ref{Cor-shallownn}, after setting $k = d-1$,  $a = \big((3e)^{d} \log  b \big)/(2\pi)$ according to Proposition~\ref{Prop:densityopclass}$(i)$.

Note that the copy complexity bound of QNE  given in Proposition \ref{Prop:copycomp-QNE} is proportional to $d$, which scales exponentially in the number of qudits. We next show that a better copy complexity can be achieved by QNE when the states have additional symmetry, namely, permutation invariance. To state this result, let  $d=m^N$  
for some  $m,N \in \NN$, and take the underlying Hilbert space as  $\HH_d=\HH_{m}^{\otimes N}$.   Let $\mathsf{S}_N$ denote the symmetric group on $\{1,\ldots,N\}$, and let $\mathsf{SU}_d$ denote the special unitary group on $\HH_d$. For a permutation $\pi \in \mathsf{S}_N$, define the permutation operator $V_{m}(\pi)$ to be the unitary operator such that 
\begin{align}
    V_m(\pi) \ket{i_1} \otimes  \cdots  \otimes \ket{i_N} = \ket{i_{\pi^{-1}(1)}} \otimes \cdots \otimes  \ket{i_{\pi^{-1}(N)}}, 
\end{align}
for all $\ket{i_1}, \ldots, \ket{i_N} \in \HH_m$. 
A state $\rho \in \cS_{d}$ is permutation invariant if $V_m(\pi)\rho V_m(\pi)^{\dag }=\rho$ for all $\pi \in \mathsf{S}_N$. Let $\mathfrak{P}\big(\HH_m^{\otimes N}\big)$ denote the class of all permutation-invariant states acting on $\HH_{m}^{\otimes N}$. 

The symmetries within the class $\mathfrak{P}\big(\HH_m^{\otimes N}\big)$ can be leveraged to construct a QNE whose copy complexity scales polynomially in $N$. The intuition for this improvement stems from the fact that the number of degrees of freedom of permutation-invariant operators acting on $\HH_m^{\otimes N}$ scales polynomially in $N$. To utilize the underlying symmetries in the construction of QNE, we will rely on   Schur--Weyl duality, which states that  $\HH_m^{\otimes N}$ can be written as  (see, e.g.,~\cite[Section 6.2]{Hayashi-2017}):
\begin{align}
    \HH_d=\HH_m^{\otimes N}=\bigoplus_{\kappa \in Y_m^N} \cW_{\kappa} \otimes \cV_{\kappa}, \label{eq:SWduality}
\end{align}
where $\kappa$ indexes the Young frames  $Y_m^N$ (with depth at most $m$ and number of boxes $N$),  and $\cW_{\kappa}$ (resp. $\cV_{\kappa}$)  denotes the corresponding irreducible space of $\mathsf{SU}_d$ (resp.  $\mathsf{S}_N$).  
Let $\cK_{m,N}\coloneqq \bigoplus_{\kappa} \cW_{\kappa}$, and let $P_{\kappa}$ denote the projection onto $\cW_{\kappa} \otimes \cV_{\kappa}$. Consider the completely positive trace preserving (CPTP) map  $\bar{\cN}\colon \cL(\HH_d) \to \cL(\cK_{m,N}) $ given by
\begin{align}
 \bar{\cN}(\rho)\coloneqq \sum_{\kappa} \operatorname{Tr}_{\cV_{\kappa}}[P_{\kappa} \rho P_{\kappa}], \label{eq:schurweylproj}
\end{align}
where $\operatorname{Tr}_{\cV_{\kappa}}[\cdot]$ denotes the partial trace with respect to $\cV_{\kappa}$. The quantum channel $ \bar{\cN}$  will be used below as a pre-processing channel for QNE, which can be realized by applying the Schur transform in conjunction with the partial trace operation.

For a PQC $\cU$ acting on $\cK_{m,N}$,  let $\bar{\cS}_{|\cK_{m,N}|}(\delta, b,\cU)$ denote the set of density operator pairs on $\cK_{m,N}$ as defined by~\eqref{eq:density-pair-copycomp}, where $|\cK_{m,N}|$ denotes the dimension of the Hilbert space $\cK_{m,N}$. 
With $\bar{P}_{\kappa}$ denoting the projection onto $\cW_{\kappa}$, let  $\bar{\cM} \colon  \cL(\cK_{m,N}) \to \cL(\HH_d) $ denote the CPTP map defined as
\begin{align}
  \bar{\cM} (\bar \rho)=  \sum_{\kappa} \bar{P}_{\kappa} \bar \rho \bar{P}_{\kappa} \otimes \pi_{\kappa}, \label{eq:invschurweylproj}
\end{align}
where $\pi_{\kappa}=I/|\cV_k|$ denotes the maximally mixed state on $\cV_{\kappa}$. 
Consider the following class of density operator pairs on $\HH_m^{\otimes N}$:
 \begin{align}
  \hat{\cS}_{m,N}(\delta, b,\cU)&\coloneqq \left\{\big(\bar{\cM} (\bar \rho),\bar{\cM} (\bar \sigma)\big):\big(\bar \rho, \bar \sigma\big) \in  \bar{\cS}_{|\cK_{m,N}|}(\delta, b,\cU)\right\}.  \notag 
\end{align}
The next proposition shows that QNE achieves a polynomial (in the number of qudits) copy complexity for density operator pairs within $\hat{\cS}_{m,N}(\delta, b,\cU)$. 
\begin{prop}[Copy complexity  with permutation invariant states]\label{Prop:copcomp-perminv}
 Let $\epsilon \geq 0$ and $b \geq 1$. With probability at least $p \in [0,1)$, QNE with $\cN=\bar{\cN}$ as given in~\eqref{eq:schurweylproj}, shallow neural net $\tilde \cF^{\mathrm{SNN}}(\bar d_{m,N},\log  b)$ and a PQC $\cU$ acting on $\cK_{m,N}$, estimates $\sD_{\mathsf{M}}(\rho\Vert\sigma)$  for $(\rho,\sigma) \in \hat{\cS}_{m,N}(\epsilon_b, b,\cU)$  within an error $\epsilon$  given $O\big(c_{b,p}^2\bar d_{m,N}|\Theta(\cU)|\epsilon^{-2}\big)$
copies each of $\rho$ and $\sigma$, where $ \bar d_{m,N}\coloneqq (N+1)^{m-1} N^{\frac{m(m-1)}{2}}$ and $\epsilon_b,c_{b,p}$ are as given in Proposition~\ref{Prop:copycomp-QNE}.
\end{prop} 

 The proof of Proposition~\ref{Prop:copcomp-perminv} is given in Section~\ref{Sec:Prop:copcomp-perminv-proof} and is based on the observation that the quantum channel $\bar{\cN}$  in~\eqref{eq:schurweylproj} preserves the measured relative entropy between two output states when the corresponding inputs are permutation-invariant. Hence, the QNE can be run on i.i.d. copies of $\bar{\cN}(\rho)$ and $\bar{\cN}(\sigma)$, which are density operators on a much smaller Hilbert space $\cK_{m,N}$ whose dimension is polynomial in $N$. To show that neural net $\tilde \cF^{\mathrm{SNN}}\big(\bar d_{m,N},\log  b\big)$ suffices for QNE on this smaller space,  we will use the fact that the Thompson metric $T(\rho,\sigma)$ contracts under application of a quantum channel. Combining these observations, the claim follows via Corollary~\ref{Cor-shallownn}. Note that the copy complexity in Proposition~\ref{Prop:copcomp-perminv} is proportional to $\bar d_{m,N}$, which is an exponential improvement over that given in Proposition~\ref{Prop:copycomp-QNE}.  We conclude this section by noting that analogous statements as given in Propositions~\ref{Prop:copycomp-QNE} and~\ref{Prop:copcomp-perminv}  can be obtained for measured R\'{e}nyi relative entropy via  Corollary~\ref{Cor-shallownn-Renyi} in place of Corollary~\ref{Cor-shallownn}. 
\section{Proofs}\label{Sec:proofs}
\subsection{Proof of Theorem~\ref{Thm:minimaxrisk-QNE}} 
\label{Sec:Thm:minimaxrisk-QNE-proof}
We will use the following lemma, whose proof is given in Appendix~\ref{Sec:Lem:bndtraceratio-proof}. 
\begin{lemma}[Eigenvalues of optimizer]\label{Lem:bndtraceratio}
 If $(\rho,\sigma) \in \cS_d(b)$, then $\abs{\lambda} \leq \log  b $ for  $\lambda \in \Lambda_{\alpha}^{\star}(\rho,\sigma)$, where $\Lambda_{\alpha}^{\star}(\rho,\sigma)$ is the set of eigenvalues of  $H_{\alpha}^{\star}(\rho,\sigma)$.
\end{lemma}
Proceeding with the proof of Theorem~\ref{Thm:minimaxrisk-QNE}, fix $(\rho,\sigma) \in \cS$.
 Let 
 \begin{align}
   \sD_{\mathsf{M},\Theta(\delta,\cS),\cF}(\rho\Vert\sigma)  \coloneqq \sup_{\bm{\theta} \in \Theta(\delta,\cS)}\sD_{\mathsf{M},\bm{\theta},\cF}(\rho\Vert\sigma), \notag
\end{align}
where 
\begin{align}
    \sD_{\mathsf{M},\bm{\theta},\cF}(\rho\Vert\sigma) \coloneqq \sup_{H \in \cH (\bm{\theta},\cF)}\operatorname{Tr}[H\rho
]-\operatorname{Tr}[e^H\sigma]+1,\label{eq:inneropt}
\end{align}
and
\begin{align}
\cH \big(\bm{\theta},\cF\big)\coloneqq \left\{H=\sum_{i=1}^d f(i) U(\bm{\theta}) |i\rangle\!\langle
i|U^{\dag}(\bm{\theta}), ~f \in \cF\right\}, \label{eq:Hermopfixedqcir}
\end{align}
is the class of Hermitian operators realized by the classical neural net for a fixed quantum circuit with parameter  $\bm{\theta}$.  
Then,  the  triangle inequality implies  that
\begin{align}
  \abs{\sD_{\mathsf{M}}(\rho\Vert\sigma) -\hat\sD_{\mathsf{M},\Theta(\delta,\cS),\cF}^n}  &\leq   \abs{\sD_{\mathsf{M}}(\rho\Vert\sigma) -\sD_{\mathsf{M},\Theta(\delta,\cS),\cF}(\rho\Vert\sigma)}  \notag \\
  & \qquad \qquad \qquad \qquad + \abs{\sD_{\mathsf{M},\Theta(\delta,\cS),\cF}(\rho\Vert\sigma) -\hat\sD_{\mathsf{M},\Theta(\delta,\cS),\cF}^n}, \label{eq:effective-err}
\end{align}
where the first and second terms on the right-hand side are referred to as the approximation error and estimation error, respectively. Note that the estimation error  is random, due to $\hat\sD_{\mathsf{M},\Theta(\delta,\cS),\cF }^n$ being dependent on the measurement outcomes.  

We first upper bound the approximation error. Since $\rho,\sigma>0$, the supremum in~\eqref{eq:rel_ent_variational} is achieved for some $ H^{\star}\coloneqq  H^{\star}(\rho,\sigma)$~\cite{Berta2015OnEntropies}. From the definition of the set $\cU(\delta,\cS)$, there exists a unitary $U(\bm{\theta}^{\star})$ for some $\bm{\theta}^{\star} \in \Theta(\delta,\cS)$ and an enumeration $\mathbf{P}\coloneqq\{P_i^{\star}\}_{i=1}^d$ of rank-one orthogonal eigen-projectors of $ H^{\star}$ such that~\eqref{eq:unitaryapprox} holds with $U^{\star}=U^{\star}(\mathbf{P})$ and $U(\bm{\theta})=U(\bm{\theta}^{\star})$. 
 Let $H^{\star}=\sum_{i=1}^d \lambda_i^{\star}P_i^{\star}$ be a spectral decomposition of $H^{\star}$. 
 By definition of $ \varepsilon(\cF,\cS)$ given in~\eqref{eq:approxcond}, there exists $f^{\star} \in \cF$ such that
\begin{align}
    \max_{i \in [1:d]} \abs{f^{\star}(e_m(i))-\lambda_i^{\star}} \leq \varepsilon(\cF,\cS),\label{eq:apprcnn}
\end{align}
This further yields
\begin{align}
  & \sD_{\mathsf{M}}(\rho\Vert\sigma)= \operatorname{Tr}[H^{\star}\rho
]-\operatorname{Tr}[e^{H^{\star}}\sigma]+1 
= \sum_{i=1}^d \lambda_i^{\star} \operatorname{Tr}[P_i^{\star} \rho] ,\notag
\end{align}
where the last equality follows from \eqref{eq:opteigvalues}.
  Set 
\begin{subequations}\label{eq:paramhermop}
\begin{align}
P_{i,\bm{\theta}^{\star}}&\coloneqq U(\bm{\theta}^{\star})\ket{i}\!\bra{i}U(\bm{\theta}^{\star})^{\dag},\label{eq:projopt}\\
F^{\star} &\coloneqq \sum_{i=1}^d f^{\star}\!\left(e_m(i)\right) \ket{i}\!\bra{i}, \label{eq:nneigvals} \\
H_{f^{\star},\bm{\theta}^{\star}} &\coloneqq \sum_{i=1}^d f^{\star}\!\left(e_m(i)\right) P_{i,\bm{\theta}^{\star}}=U(\bm{\theta}^{\star})F^{\star} U(\bm{\theta}^{\star})^{\dag}. \label{eq:nnqcopt}
\end{align}   
\end{subequations}
We can bound the approximation error as follows: 
\begin{align}
    \abs{\sD_{\mathsf{M}}(\rho\Vert\sigma) -\sD_{\mathsf{M},\Theta(\delta,\cS),\cF}(\rho\Vert\sigma)}
   &\stackrel{(a)}{=}  \sD_{\mathsf{M}}(\rho\Vert\sigma) -\sD_{\mathsf{M},\Theta(\delta,\cS),\cF}(\rho\Vert\sigma) \notag\\
  &=\operatorname{Tr}[H^{\star}\rho
]-\operatorname{Tr}[e^{H^{\star}}\sigma]+1-\sD_{\mathsf{M},\Theta(\delta,\cS),\cF}(\rho\Vert\sigma) \notag\\
    & \stackrel{(b)}{\leq} \operatorname{Tr}\! \big[(H^{\star}-H_{f^{\star},\bm{\theta}^{\star}})\rho
\big]+\left(\operatorname{Tr}\!\big[e^{H_{f^{\star},\bm{\theta}^{\star}}}\sigma\big]-\operatorname{Tr}\!\big[e^{H^{\star}}\sigma\big]\right)\notag \\
& \stackrel{(c)}{\leq} \big\|H^{\star}-H_{f^{\star},\bm{\theta}^{\star}}\big\|+ \norm{e^{H^{\star}}-e^{H_{f^{\star},\bm{\theta}^{\star}}}}, \label{eq:euppbndefferr}
\end{align}
where 
\begin{enumerate}[(a)]
    \item   follows because $ \sD_{\mathsf{M},\Theta(\delta,\cS),\cF}(\rho\Vert\sigma) \leq \sD_{\mathsf{M}}(\rho\Vert\sigma)$ since the supremum  over $H$ (see~\eqref{eq:inneropt}) is over a subset $\cup_{\bm{\theta} \in \Theta(\delta,\cS)} \cH \big(\bm{\theta},\cF\big)$ of the class of all Hermitian operators;
    \item follows by noting  that  $H_{f^{\star},\bm{\theta}^{\star}} \in \cH \big(\bm{\theta}^{\star},\cF\big)$ (since $f^{\star} \in \cF$ and   $\bm{\theta}^{\star} \in \Theta(\delta,\cS)$), thus implying that
    \begin{align}
        \sD_{\mathsf{M},\Theta(\delta,\cS),\cF}(\rho\Vert\sigma) \geq   \sD_{\mathsf{M},\bm{\theta^{\star}},\cF}(\rho\Vert\sigma) \geq \operatorname{Tr}[H_{f^{\star},\bm{\theta}^{\star}}\rho
]-\operatorname{Tr}[e^{H_{f^{\star},\bm{\theta}^{\star}}}\sigma]+1; \notag
    \end{align}
    \item uses operator H\"{o}lder's inequality.
\end{enumerate}
To control the terms in the right hand side of~\eqref{eq:euppbndefferr}, we will use the following lemma.
\begin{lemma}[Bound on operator norm]\label{Lem:opnorm-diff}
Let $H_1$ and $H_2$ be Hermitian operators with spectral decompositions $H_1=U_1 \Lambda_1 U_1^{\dag}$ and $H_2=U_2 \Lambda_2 U_2^{\dag}$. Then
\begin{align}
    \norm{H_1-H_2} \leq \norm{\Lambda_1-\Lambda_2}+\big(\norm{\Lambda_1}+\norm{\Lambda_2}\big)\norm{U_1-U_2}. \notag
\end{align}
\end{lemma}
\noindent For completeness, the proof is given in Appendix~\ref{Sec:Lem:opnorm-diff-proof} and relies on sub-multiplicativity of the spectral norm.

\medskip
\noindent Applying Lemma~\ref{Lem:opnorm-diff} with $H_1=H^{\star}=U^{\star}(\mathbf{P})\Lambda^{\star}U^{\star}(\mathbf{P})^{\dag}$ and $H_2=H_{f^{\star},\bm{\theta}^{\star}}=U(\bm{\theta}^{\star})F^{\star} U(\bm{\theta}^{\star})^{\dag}$ yields
\begin{align}
 \norm{H^{\star}-H_{f^{\star},\bm{\theta}^{\star}}}\leq \varepsilon(\cF,\cS) +  2 \delta \log  b,\label{eq:approx-herm-op}
\end{align}
where we used $\norm{U^{\star}(\mathbf{P})-U(\bm{\theta}^{\star})} \leq \delta$ due to~\eqref{eq:unitaryapprox},  $\norm{\Lambda_1}=\norm{\Lambda^{\star}}\leq \log b$ by Lemma~\ref{Lem:bndtraceratio},  $\norm{\Lambda_2}=\norm{F^{\star}}=
\norm{f^{\star}}_{\infty,\cX} \leq \log b$, and 
\begin{align}
    &\norm{\Lambda_1-\Lambda_2}=\max_{i \in [1:d]} \abs{f^{\star}(e_m(i))-\lambda_i^{\star}} \leq \varepsilon(\cF,\cS),\notag 
\end{align}
which follows from~\eqref{eq:apprcnn}.  
Similarly, applying Lemma~\ref{Lem:opnorm-diff} with $H_1=U^{\star}(\mathbf{P})e^{\Lambda^{\star}}U^{\star}(\mathbf{P})^{\dag}$ and $H_2=U(\bm{\theta}^{\star})e^{F^{\star}} U(\bm{\theta}^{\star})^{\dag}$ yields
\begin{align}
\norm{e^{H^{\star}}-e^{H_{f^{\star},\bm{\theta}^{\star}}}}=\norm{U^{\star}(\mathbf{P})e^{\Lambda^{\star}}U^{\star}(\mathbf{P})^{\dag}-U(\bm{\theta}^{\star})e^{F^{\star}}U^{\dag}(\bm{\theta}^{\star})}  \leq  b\,\varepsilon(\cF,\cS)+2 b\delta,  \label{eq:apprqcexp}
\end{align}
where we used $
\norm{e^{F^{\star}}} =\norm{e^{f^{\star}}}_{\infty,\cX} \leq b$, $
\norm{e^{\Lambda^{\star}}}  \leq b$, and 
\begin{align}
\norm{e^{F^{\star}}-e^{\Lambda^{\star}}}=\max_{i \in [1:d]} \abs{e^{\lambda_i^{\star}}-e^{f^{\star}\!\left(e_m(i)\right)}} 
  &\stackrel{(a)}{\leq} b\mspace{2 mu} \max_{i \in [1:d]} \abs{\lambda_i^{\star}-f^{\star}\!\left(e_m(i)\right)}  \leq b \, \varepsilon(\cF,\cS).\notag
\end{align}
In the above,  $(a)$ is due to the fact that $e^x$ is Lipschitz\footnote{Recall that, by definition, $f$ is $L$-Lipschitz on $[a,b]$ if   $|f(x)-f(y)| \leq L \abs{x-y}$ for all $x,y \in [a,b]$.} for all $x \in [-\log  b ,\log  b ]$, with the Lipschitz constant bounded by $b$.  
Combining~\eqref{eq:euppbndefferr}-\eqref{eq:apprqcexp}  the approximation error can be bounded as
\begin{align}
  \abs{\sD_{\mathsf{M}}(\rho\Vert\sigma) -\sD_{\mathsf{M},\Theta(\delta,\cS),\cF}(\rho\Vert\sigma)} \leq \varepsilon(\cF,\cS)(1+b)+  2\delta  (b+\log  b). \label{eq:approxerr} 
\end{align}

We next obtain an upper bound on the expected estimation error. First observe that
\begin{align}
  & \abs{\sD_{\mathsf{M},\Theta(\delta,\cS),\cF}(\rho\Vert\sigma) -\hat\sD_{\mathsf{M},\Theta(\delta,\cS),\cF}^n}\notag \\
  & \leq \sup_{\bm{\theta} \in \Theta(\delta,\cS)}\abs{\sD_{\mathsf{M},\theta,\cF}(\rho\Vert\sigma) -\hat\sD_{\mathsf{M},\theta,\cF}^n} \notag \\
  &\leq \sup_{\bm{\theta} \in \Theta(\delta,\cS)} \sup_{f  \in \cF}\abs{\frac{1}{n}\sum_{\ell=1}^{n}f(i_\ell(\bm{\theta}))-\frac{1}{n}\sum_{\ell=1}^{n}
e^{f(j_\ell(\bm{\theta}))}-\EE_{p_{\bm{\theta}}^\rho}[f]+\EE_{q_{\bm{\theta}}^\sigma}\big[e^f\big]} \notag \\
&    \leq \sup_{\bm{\theta} \in \Theta(\delta,\cS)}  \sup_{f  \in \cF}\abs{\frac{1}{n}\sum_{\ell=1}^{n}\big(f(i_\ell(\bm{\theta}))-\EE_{p_{\bm{\theta}}^\rho}[f]\big)}+ \sup_{\bm{\theta} \in \Theta(\delta,\cS)} \sup_{f  \in \cF}\abs{\frac{1}{n}\sum_{\ell=1}^{n}
e^{f(j_\ell(\bm{\theta}))}-\EE_{q_{\bm{\theta}}^\sigma}\big[e^f\big]}.\notag \\
&\leq  \sup_{\bm{\theta} \in \Theta(\delta,\cS)}\mspace{-2 mu}\sup_{f  \in \cF}\mspace{-2 mu}\abs{\frac{1}{n}\sum_{\ell=1}^{n}\big(f(i_\ell(\bm{\theta}))-\EE_{p_{\bm{\theta}}^\rho}[f]\big)} +  b\mspace{-2 mu} \sup_{\bm{\theta} \in \Theta(\delta,\cS)} \mspace{-2 mu} \sup_{f  \in\cF} \mspace{-2 mu}\abs{\frac{1}{n}\sum_{\ell=1}^{n}\big(f(j_\ell(\bm{\theta}))-\EE_{q_{\bm{\theta}}^\sigma}[f]\big)}, \label{eq:bnddev-mrel} 
\end{align}
where to obtain the first inequality,  we used that for arbitrary functionals $f\colon \Omega \mapsto \RR$ and $g\colon \Omega \mapsto \RR$ such that at least $\sup_{\omega \in \Omega}f(\omega)$ or $\sup_{\omega \in \Omega}g(\omega)$ is finite, 
\begin{align}
    \abs{\sup_{\omega \in \Omega}f(\omega)-\sup_{\omega \in \Omega}g(\omega)} \leq \sup_{\omega \in \Omega}\abs{f(\omega)-g(\omega)}, \label{eq:ineqsup}
\end{align} 
and the final inequality again follows from the Lipschitz property of $e^x$ mentioned above. 
Now, note that each  term in~\eqref{eq:bnddev-mrel} is a separable random process indexed by $(f, \theta)$ since the neural net class has a finite number of bounded parameters with  continuous activation and each parameter can be approximated by rational-valued parameters, and the unitary group  on $\HH_d$ is separable (e.g., with respect to the trace norm). Hence, the suprema above are measurable. Taking expectations, we obtain
 \begin{align}
  & \EE \!\left[\abs{\sD_{\mathsf{M},\Theta(\delta,\cS),\cF}(\rho\Vert\sigma) -\hat\sD_{\mathsf{M},\Theta(\delta,\cS),\cF}^n}\right] \notag \\
& \leq  \EE \!\left[\sup_{\substack{f  \in \cF,\\\bm{\theta} \in \Theta(\delta,\cS)}} \abs{\frac{1}{n}\sum_{\ell=1}^{n}\big(f(i_\ell(\bm{\theta}))-\EE_{p_{\bm{\theta}}^\rho}[f]\big)}\right] + b~ \EE \!\left[ \sup_{\substack{f  \in \cF,\\\bm{\theta} \in \Theta(\delta,\cS)}}\abs{\frac{1}{n}\sum_{\ell=1}^{n}\big(f(j_\ell(\bm{\theta}))-\EE_{q_{\bm{\theta}}^\sigma}[f]\big)}\right] \notag \\
& \leq  2\mspace{2 mu} \EE \!\left[\sup_{\bm{\theta} \in \Theta(\delta,\cS)} \sup_{f  \in \cF}\abs{\frac{1}{n}\sum_{\ell=1}^{n}R_{\ell}f(i_\ell(\bm{\theta}))}\right] + 2\mspace{2 mu}b \mspace{2 mu}\EE \!\left[ \sup_{\bm{\theta} \in \Theta(\delta,\cS)}\sup_{f  \in\cF}\abs{\frac{1}{n}\sum_{\ell=1}^{n}R_{\ell}\big(f(j_\ell(\bm{\theta}))}\right],
\label{eq:symmproc} 
\end{align}
where the final step 
 is due to the symmetrization inequality~\cite[Lemma 2.3.1]{AVDV-book}  using i.i.d. Rademacher variables $\{R_{\ell}\}_{\ell=1}^n$, i.e., $\PP(R_{\ell}=1)=\PP(R_{\ell}=-1)=0.5$, independent of $\cup_{\bm{\theta} \in \Theta(\delta,\cS)} \{i_{\ell}(\bm{\theta}),j_{\ell}(\bm{\theta})\}_{\ell=1}^n$.

Let 
\begin{align}
    & Z^{(n)}_{f,\bm{\theta}}\coloneqq  \frac{1}{\sqrt{n}}\sum_{\ell=1}^{n}R_{\ell} f(i_\ell(\bm{\theta})) \quad \text{and} \quad  \tilde Z^{(n)}_{f,\bm{\theta}}\coloneqq    \frac{1}{\sqrt{n}}\sum_{\ell=1}^{n}R_{\ell} f(j_\ell(\bm{\theta})). \notag  
\end{align}
Let  $\mu_n$ and $\nu_n $ be the empirical probability measures corresponding to $\{i_{\ell}(\bm{\theta})\}_{\ell=1}^n$ and $\{j_{\ell}(\bm{\theta}) \}_{\ell=1}^n$, respectively. 
Note that  for fixed $\big (f,\bm{\theta},\{i_{\ell}(\bm{\theta}),j_{\ell}(\bm{\theta})\}_{\ell=1}^n\big)$, $Z^{(n)}_{f,\bm{\theta}}$ and $\tilde Z^{(n)}_{f,\bm{\theta}}$, being weighted sums of i.i.d.~Rademacher variables,  are centered and sub-Gaussian random variables  with respect to  pseudo-metrics $\|\cdot\|_{2,\mu_n}$ and $\|\cdot\|_{2,\nu_n}$ on $\cF \times \Theta(\delta,\cS)$, respectively. Hence, they both  are also centered and sub-Gaussian random variables  with respect to the pseudometric $\|\cdot\|_{\infty,\cX}$ on $\cF \times \Theta(\delta,\cS)$, which does not depend on $\bm{\theta}$. For instance, we have via an application of Hoeffding's lemma that for all $(f,\bm{\theta}),~(f',\bm{\theta}') \in \cF \times \Theta(\delta,\cS)$,
\begin{align}
& \EE_{\{R_{\ell}\}_{\ell=1}^n}\left[ Z^{(n)}_{f,\bm{\theta}}\right]=0,~\forall~(f,\bm{\theta}) \in \cF \times \Theta(\delta,\cS), \notag \\
 \text{and}  \quad & \EE_{\{R_{\ell}\}_{\ell=1}^n}\left[ e^{t\big(Z^{(n)}_{f,\bm{\theta}}-Z^{(n)}_{f',\bm{\theta}'}\big)}\right] \leq e^{\frac{t^2 \norm{f-f'}_{2,\mu_n}}{2}} \leq e^{\frac{t^2 \norm{f-f'}_{\infty,\cX}}{2}}. \notag 
\end{align}
Then, using the law of iterated expectations by first taking conditional expectation with respect to $\{R_{\ell}\}_{\ell=1}^n$ given the samples $\cup_{\bm{\theta} \in \Theta(\delta,\cS)} \{i_{\ell}(\bm{\theta}),j_{\ell}(\bm{\theta})\}_{\ell=1}^n$ and then with respect to these samples,  we obtain
\begin{align}
&\EE \!\left[\abs{\sD_{\mathsf{M},\Theta(\delta,\cS),\cF}(\rho\Vert\sigma) -\hat\sD_{\mathsf{M},\Theta(\delta,\cS),\cF}^n}\right] \notag \\
& \leq  2\mspace{2 mu} \EE_{\mathsf{P}_n} \!\left[ \EE_{\big|\cup_{\bm{\theta} \in \Theta(\delta,\cS)} \{i_{\ell}(\bm{\theta}),j_{\ell}(\bm{\theta})\}_{\ell=1}^n} \left[\sup_{\bm{\theta} \in \Theta(\delta,\cS)} \sup_{f  \in \cF}\abs{\frac{1}{n}\sum_{\ell=1}^{n}R_{\ell}f(i_\ell(\bm{\theta}))}\right]\right] \notag \\
&\qquad \qquad     + 2\mspace{2 mu}b \mspace{2 mu}\EE_{\mathsf{Q}_n}\!\left[ \EE_{\big| \cup_{\bm{\theta} \in \Theta(\delta,\cS)} \{i_{\ell}(\bm{\theta}),j_{\ell}(\bm{\theta})\}_{\ell=1}^n}\left[ \sup_{\bm{\theta} \in \Theta(\delta,\cS)}\sup_{f  \in\cF}\abs{\frac{1}{n}\sum_{\ell=1}^{n}R_{\ell}f(j_\ell(\bm{\theta}))}\right]\right] \label{eq:lawiterexp} \\
    & \stackrel{(a)}{\leq} 24 n^{-\frac 12}  \EE_{\mathsf{P}_n} \mspace{-2 mu}\left[\int_{0}^{\infty}\sqrt{ \log  N\big(\epsilon,\cF,\|\cdot\|_{\infty,\cX}\big)}d\epsilon \right]+24 b \EE_{\mathsf{Q}_n} \mspace{-2 mu}\left[\int_{0}^{\infty}\sqrt{ \log  N\big(\epsilon, \cF,\|\cdot\|_{\infty,\cX}\big)} d\epsilon \right] \notag \\
&\leq 48 n^{-\frac{1}{2}} (b+1) \log  b\int_{0}^{1}\sqrt{ \log  N\big(2\epsilon \log  b,\cF,\|\cdot\|_{\infty,\cX}\big)}d\epsilon \notag \\
&\stackrel{(b)}{\leq } 48(b+1) \log  b~   \kappa(\cF) n^{-\frac 12},\label{eq:expempesterrbnd}
\end{align}
where the outer expectations in the first  and second terms in~\eqref{eq:lawiterexp} are with respect to probability measures $\mathsf{P}_n \coloneqq \bigotimes_{\bm{\theta} \in \Theta(\delta,\cS) }(p_{\bm{\theta}}^\rho)^{\otimes n}$ and  $\mathsf{Q}_n \coloneqq \bigotimes_{\bm{\theta} \in \Theta(\delta,\cS) }(q_{\bm{\theta}}^\sigma)^{\otimes n}$, respectively, and 
\begin{enumerate}[(a)]
    \item follows by applying the standard maximal inequality for bounding the expectation of a separable sub-Gaussian process  with respect to the $\|\cdot\|_{\infty,\cX}$ pseudometric  on $\cF \times \Theta(\delta,\cS)$  (see e.g.~\cite[Corollary 2.2.8.]{AVDV-book} and \cite[Corollary 5.25]{VanHandel-book}) as given below:
    \begin{align}
    \EE_{\{R_{\ell}\}_{\ell=1}^n}\left[ Z^{(n)}_{f,\bm{\theta}}\right] \leq 12 \int_{0}^{\infty}\sqrt{ \log  N\big(\epsilon,\cF,\|\cdot\|_{\infty,\cX}\big)}d\epsilon. \notag 
    \end{align}
    \item  follows from~\eqref{eq:coventcond}.
\end{enumerate}
From~\eqref{eq:approxerr} and~\eqref{eq:expempesterrbnd},~\eqref{eq:minimax risk} follows via the triangle inequality and noting that the resulting bound holds uniformly for $(\rho,\sigma) \in \cS$. 

\medskip

\noindent 
To prove~\eqref{eq:tailineQNEmrel}, we will use the  theorem stated below\footnote{While \cite[Theorem 5.29]{VanHandel-book} is stated without specifying explicit constants, a straightforward computation by following the proof therein yields the constants stated in~\eqref{eq:tailineqsubgaussian}. }, which gives a tail probability bound for the deviation of the supremum of a sub-Gaussian process from its associated entropy integral. 
\begin{theorem}{\cite[Theorem 5.29]{VanHandel-book}}\label{thm:tailineq}
Let $(X_{\theta})_{\theta \in \Theta }$ be a separable sub-Gaussian process on a pseudometric space $(\Theta,\mathsf{d})$. Then, for all $\theta_0 \in \Theta$ and $z\geq 0$, we have
\begin{align}
  &  \mathbb{P}\left(\sup_{\theta \in \Theta} X_{\theta}-X_{\theta_0} \geq 32 \int_{0}^{\infty} \sqrt{\log N(\epsilon,\Theta,\mathsf{d})}d\epsilon+z \right) \leq 2\mspace{2 mu}e^{-\frac{z^2}{72\mspace{2 mu}\mathsf{diam}(\Theta,\mathsf{d})^2}}, \label{eq:tailineqsubgaussian}
\end{align}
where $\mathsf{diam}(\Theta,\mathsf{d})\coloneqq  \sup\limits_{\theta,\tilde{\theta} \in \Theta}\mathsf{d}(\theta,\tilde{\theta})$. 
\end{theorem}
From~\eqref{eq:bnddev-mrel}, it follows that
\begin{align}
    & \abs{\sD_{\mathsf{M},\Theta(\delta,\cS),\cF}(\rho\Vert\sigma) -\hat\sD_{\mathsf{M},\Theta(\delta,\cS),\cF}^n}  \leq  \sup_{\bm{\theta} \in \Theta(\delta,\cS)}\mspace{-2 mu}\sup_{f  \in \cF} \abs{\bar Z^{(n)}_{f,\bm{\theta}}},\label{eq:uppbnddesterr} 
    \end{align}
    where the above inequality again used~\eqref{eq:ineqsup} with
    \begin{align}
  & \bar Z^{(n)}_{f,\bm{\theta}} \coloneqq   \frac{1}{n}\sum_{\ell=1}^{n}f(i_\ell(\bm{\theta}))-\frac{1}{n}\sum_{\ell=1}^{n}
e^{f(j_\ell(\bm{\theta}))}-\EE_{p_{\bm{\theta}}^\rho}[f]+\EE_{q_{\bm{\theta}}^\sigma}\big[e^f\big]. \notag 
\end{align}

To apply Theorem~\ref{thm:tailineq}, we will show that $\{\bar Z^{(n)}_{f,\bm{\theta}}\}_{(f,\bm{\theta})\in \cF \times \Theta(\delta,\cS)}$ is a centered  separable sub-Gaussian process with respect to an appropriate pseudo-metric defined on $ \cF \times \Theta(\delta,\cS)$. To see this, observe that $\EE\big[\bar Z^{(n)}_{f,\bm{\theta}}\big]=0$ for all $(f,\bm{\theta})\in \cF \times \Theta(\delta,\cS)$. Moreover,  
\begin{align}
\bar Z^{(n)}_{f,\bm{\theta}}-  \bar Z^{(n)}_{\tilde f,\tilde{\bm{\theta}}}  & = \sum_{\ell=1}^{n} \frac{1}{n} \bigg( \big(f(i_\ell(\bm{\theta}))-f(i_\ell(\tilde{\bm{\theta}})\big)-\big(
e^{f(j_\ell(\bm{\theta}))}-e^{\tilde f(j_\ell(\tilde{\bm{\theta}}))}\big)-\big(\EE_{p_{\bm{\theta}}^\rho}[f]-\EE_{p_{\tilde{\bm{\theta}}}^\rho}[\tilde f]\big) \notag \\
 &\qquad \qquad \qquad +\big(\EE_{q_{\bm{\theta}}^\sigma}\big[e^f\big]-\EE_{q_{\tilde{\bm{\theta}}}^\sigma}\big[e^{\tilde f}\big] \big) \bigg). \notag 
\end{align}
Since $\{(i_\ell(\bm{\theta}),j_\ell(\bm{\theta}))\}_{\ell=1}^n$ are i.i.d. variables, the moment generating function of $\bar Z^{(n)}_{f,\bm{\theta}}-  \bar Z^{(n)}_{\tilde f,\tilde{\bm{\theta}}}$ is equal to the product of the moment generating function of each term within the sum above. Also, by using the  Lipschitz property of $e^x$ mentioned earlier, each term within that sum is upper bounded as follows: 
\begin{align}
&\bigg|\frac{1}{n} \bigg( \big(f(i_\ell(\bm{\theta}))-\tilde f(i_\ell(\tilde{\bm{\theta}})\big)-\big(
e^{f(j_\ell(\bm{\theta}))}-e^{\tilde f(j_\ell(\tilde{\bm{\theta}}))}\big)-\big(\EE_{p_{\bm{\theta}}^\rho}[f]-\EE_{p_{\tilde{\bm{\theta}}}^\rho}[\tilde f]\big) \notag \\
&\qquad \qquad \qquad \qquad \qquad \qquad \qquad \qquad \qquad \qquad \qquad \qquad +\big(\EE_{q_{\bm{\theta}}^\sigma}\big[e^f\big]-\EE_{q_{\tilde{\bm{\theta}}}^\sigma}\big[e^{\tilde f}\big] \big) \bigg) \bigg| \notag \\
& \leq  \frac{2(1+b)}{n}\big\|f-\tilde f\big\|_{\infty,\cX}. \label{eq:bnddiffsep} 
\end{align} 
 Then, it follows from Hoeffding's lemma  that 
\begin{align}
\mathbb{E}\bigg[e^{t\big(\bar Z^{(n)}_{f,\bm{\theta}}-  \bar Z^{(n)}_{\tilde f,\tilde{\bm{\theta}}}\big)} \bigg]\leq   e^{\frac{1}{2}t^2 \mathsf{d}_{b,n}(f,\tilde f)^2}, \notag
\end{align}
where $\mathsf{d}_{b,n}(f,\tilde f)$ is the pseudo-metric on $\cF \times \Theta(\delta,\cS)$ given by
\begin{align}
\mathsf{d}_{b,n}(f,\tilde f)\coloneqq \mathsf{d}_{b,n}\big((f,\theta),(\tilde f,\tilde \theta)\big) \coloneqq 2(1+b)n^{-\frac 12}\big\|f-\tilde f\big\|_{\infty,\cX}. \notag
\end{align}
  Thus, $\{\bar Z^{(n)}_{f,\bm{\theta}}\}_{(f,\bm{\theta})\in \cF \times \Theta(\delta,\cS)}$ is a separable sub-Gaussian process on the pseudo-metric space $\big(\cF \times \Theta(\delta,\cS),\mathsf{d}_{b,n}\big)$, where the separability  follows from \eqref{eq:bnddiffsep} and the fact that $\cF \times \Theta(\delta,\cS)$ is totally bounded  with respect to the pseudometric $\big\|\cdot\big\|_{\infty,\cX}$ by~\eqref{eq:coventcond}  in Assumption~\ref{Assump1}. The total boundedness is because $N (\epsilon,\cF \times \Theta, \norm{\cdot}_{\infty,\cX}) $ is finite for all $\epsilon>0$ due to \eqref{eq:coventcond} and the monotonically non-increasing property of covering number in $\epsilon$. Also, 
  \begin{align}
   \mathsf{diam}(\cF \times \Theta(\delta,\cS),\mathsf{d}_{b,n}) =4n^{-\frac 12}(1+b)\log b , \notag  
  \end{align}
  and 
  \begin{align}
  \int_{0}^{\infty} \sqrt{\log N(\epsilon,\cF \times \Theta(\delta,\cS),\mathsf{d}_{b,n})}d\epsilon  &= \int_{0}^{\mathsf{diam}(\cF \times \Theta(\delta,\cS),\mathsf{d}_{b,n})} \sqrt{\log N(\epsilon,\cF \times \Theta(\delta,\cS),\mathsf{d}_{b,n})}d\epsilon  \notag \\
      &=4n^{-\frac 12}(1+b) (\log b) \int_{0}^{1} \sqrt{\log N(2\epsilon \log b,\cF, \| \cdot \|_{\infty,\cX})}d\epsilon \notag \\
      &= 4 n^{-\frac 12}(1+b) (\log b) \mspace{2 mu} \kappa(\cF), \notag
  \end{align}
where the first equality follows because in an arbitrary pseudometric space $(\Theta,\mathsf{d})$, $N(\epsilon, \Theta,\mathsf{d})=1$ for $\epsilon \geq \mathsf{diam}(\Theta,\mathsf{d})$, while the second equality uses a change of variables along with  $N(a\epsilon, \Theta,a\mathsf{d})=N(\epsilon,\Theta,\mathsf{d})$ for $a>0$. 
Then,~\eqref{eq:tailineqsubgaussian} in Theorem~\ref{thm:tailineq} (with $\bar Z^{(n)}_{0,\bm{\theta}}=0$ which corresponds to $f=0 \in \cF$)   yields 
  \begin{align}
\mathbb{P}\left(\sup_{\bm{\theta} \in \Theta(\delta,\cS)}\mspace{-2 mu}\sup_{f  \in \cF} \bar Z^{(n)}_{f,\bm{\theta}}\geq 128 \mspace{2 mu}n^{-\frac 12}(1+b) \log b \mspace{2 mu} \kappa(\cF)+ \bar z \right) \leq 2\mspace{2 mu}e^{-\frac{n \bar  z^2}{72(1+b)^2(\log b)^2}}, ~\bar z \geq 0,\notag
  \end{align}
  for some $c>0$. 
  Noting that the above also holds with $\bar Z^{(n)}_{f,\bm{\theta}}$ replaced by $-\bar Z^{(n)}_{f,\bm{\theta}}$, we obtain
    \begin{align}
  \mathbb{P}\left(\sup_{\bm{\theta} \in \Theta(\delta,\cS)}\mspace{-2 mu}\sup_{f  \in \cF} \abs{\bar Z^{(n)}_{f,\bm{\theta}}} \geq 128 \mspace{2 mu} n^{-\frac 12}(1+b) \log b \mspace{2 mu} \kappa(\cF)+\bar z \right) \leq  2\mspace{2 mu}e^{-\frac{n \bar z^2}{72(1+b)^2(\log b)^2}}. \notag
  \end{align}
From~\eqref{eq:effective-err} and~\eqref{eq:approxerr}, we have with $\xi_{\cF,\cS,\delta,b}\coloneqq \varepsilon(\cF,\cS)(b+1)+ 2\delta   (b+\log  b )$ that
\begin{align}
&  \mathbb{P}\left( \abs{\sD_{\mathsf{M}}(\rho\Vert\sigma) -\hat\sD_{\mathsf{M},\Theta(\delta,\cS),\cF}^n} \geq \xi_{\cF,\cS,\delta,b}+ 128n^{-\frac 12}(1+b) \log b \mspace{2 mu} \kappa(\cF)+\bar z \right) \leq  2\mspace{2 mu}e^{-\frac{n \bar z^2}{72(1+b)^2(\log b)^2}},  \notag
\end{align}
for some $c \geq 0$. Setting $\bar z= \sqrt{72} \mspace{2 mu}z(1+b) \log b$ and taking the supremum over $(\rho,\sigma) \in \cS$ leads to the claim in~\eqref{eq:tailineQNEmrel},  thus completing the proof. 
\subsection{Proof of Corollary~\ref{Cor-shallownn}}
\label{Sec:Cor-shallownn-proof}
We will apply Theorem~\ref{Thm:minimaxrisk-QNE} with $\cF=\tilde \cF^{\mathrm{SNN}}(d,\log  b ) \subseteq \cF_b$ and $\cS=\cS_d(b)$. To this end, we first verify that Assumption~\ref{Assump1}  is satisfied as required therein. Note  that  $\cX=\{\ket{i}, i \in [1:d]\}$ for the embedding $e_d: i \mapsto \ket{i}$.  By Lemma~\ref{Lem:bndtraceratio}, for all $(\rho,\sigma) \in \cS_d(b)$ and $ \Lambda^{\star}(\rho,\sigma)=\{\lambda_i^{\star}\}_{i \in [1:d]}$, we have 
$\abs{\lambda_i^{\star}} \leq \log  b $ for all $i \in [1:d]$. Hence, \begin{align}
    f^{\star}(x)\coloneqq \sum_{i=1}^d \lambda_i^{\star} \varphi_{\mathrm{R}}\big( \langle i, x \rangle \big) \in \tilde \cF^{\mathrm{SNN}}(d,\log  b ), \notag
\end{align}
  satisfies 
  $\max_{i \in [1:d]}\abs{f^{\star}(e_d(i))-\lambda_i^{\star}}=0$. Further, observe that $\norm{f}_{\infty,\cX} \leq \log  b $ for  $f \in \tilde \cF^{\mathrm{SNN}}(d,\log  b )$ and that for all $f=\sum_{i=1}^d \beta_i \varphi_{\mathrm{R}}\big( \langle i, x \rangle \big)$ and $\tilde f =\sum_{i=1}^d \tilde \beta_i \varphi_{\mathrm{R}}\big( \langle i, x \rangle \big)$, we have
  \begin{align}
    \norm{f-\tilde f}_{\infty,\cX}=\max_{i \in [1:d]} \big |\beta_i-\tilde \beta_i \big|. \notag
  \end{align} 
This implies that
\begin{align}
N\big (\epsilon,\tilde \cF^{\mathrm{SNN}}(d,\log  b ), \norm{\cdot}_{\infty,\cX}\big) =N\big (\epsilon,[-\log  b ,\log  b ]^d, \norm{\cdot}_{\infty}\big) \leq \big(3 (\log  b) \mspace{2 mu} \epsilon^{-1}\big)^{d}, \forall ~\epsilon >0, \notag
\end{align}
where the second inequality used a standard upper bound on the  $\epsilon$-covering number of $\norm{\mspace{1 mu}\cdot\mspace{1 mu}}_{\infty}$ ball of radius $\log  b $ in $\RR^d$.  
Hence
\begin{align}
 \log   N (2\epsilon \log  b,\tilde \cF^{\mathrm{SNN}}(d,\log  b ), \norm{\cdot}_{\infty,\cX}) \leq d \log  (3\epsilon^{-1}) \leq 3d \epsilon^{-1}, \forall ~\epsilon>0. \notag
\end{align}
From the above, it follows   that Assumption~\ref{Assump1}  holds with 
\begin{align}
&\varepsilon\big(\cF^{\mathrm{SNN}}(d,\log  b ),\cS_d(b)\big)=0, \notag \\
&\bar \kappa\big(\cF^{\mathrm{SNN}}(d,\log  b ),\epsilon\big)=d \epsilon^{-1}, \notag \\
\mbox{ and }&\kappa\big(\cF^{\mathrm{SNN}}(d,\log  b )\big)=  \int_{0}^1\sqrt{d \epsilon^{-1}} d\epsilon= 2 \sqrt{d}. \notag    
\end{align}
Substituting these in~\eqref{eq:minimax risk} and~\eqref{eq:tailineQNEmrel} leads to the desired claim.

\subsection{Proof of Proposition~\ref{Prop:deepnn}}
\label{Sec:Prop:deepnn-proof}
We will  apply Theorem~\ref{Thm:minimaxrisk-QNE} with $\cF=\tilde \cF^{\mathrm{NN}}(L_{\varepsilon,k,a})\subseteq \cF_b$ and $\cS= \tilde{\cS}_{k,d}(b,a,\varepsilon)\subseteq \cS_d(b)$. We first verify that Assumption~\ref{Assump1}  is satisfied. Note that for the embedding $e_1\colon i \mapsto i/d$,  we have $\cX=\{1/d,2/d,\ldots,1\}$. By definition, for any $(\rho,\sigma) \in \tilde{\cS}_{k,d}(b,a,\varepsilon)$ and any enumeration $\Lambda^{\star}(\rho,\sigma)=\{\lambda_i^{\star}\}_{i=1}^d$, there exists a    polynomial $p \in  \cP_k(a)$ of degree at most $k$ such that 
\begin{align}
    \max_{i \in [1:d]}\abs{p\big(e_1(i)\big)-\lambda_i^{\star}}\leq \varepsilon. \notag
\end{align}
Now, suppose the neural class 
$\cF^{\mathrm{NN}}(L,K,M,1)$  with depth $L$, width $K$, and parameter  norm $M$ appropriately chosen, is such that there exists $f \in \cF^{\mathrm{NN}}(L,K,M,1)$ satisfying
 \begin{align}
     \norm{f-p}_{\infty,\cX} \leq \varepsilon. \label{eq:polyapprox-deepnn}
 \end{align}
 Then,  we  obtain via the triangle inequality that 
\begin{align}
     \max_{i \in [1:d]} \abs{f(e_1(i))-\lambda_i^{\star}} \leq 2 \varepsilon. \label{eq:deepnnapprox} 
\end{align}
Hence,~\eqref{eq:approxcond} holds with $\varepsilon\big(\cF^{\mathrm{NN}}(L,K,M,1),\tilde{\cS}_{k,d}(b,a,\varepsilon)\big) =2 \varepsilon$. 
For quantifying the error in approximating polynomials via neural nets in~\eqref{eq:polyapprox-deepnn}, we will use the error bounds from~\cite{Dennis-2021}.  
In particular,~\cite[Proposition~III.5]{Dennis-2021}
implies that there exists $f \in \cF^{\mathrm{NN}}(L_{\varepsilon,k,a},K,M,1) $ with $K=9$, $M = 1$ and  
\begin{align}
    L_{\varepsilon,k,a} &\coloneqq c k\left(\log  (\varepsilon^{-1})+ \log  k+\log  a \right), \label{eq:depthnn} 
\end{align}
 such that~\eqref{eq:polyapprox-deepnn} holds,  where  $c>0$ is some constant. 
 By considering the truncated neural class given in~\eqref{eq:truncnnclass}, 
we obtain that 
\begin{align}
    \norm{f}_{\infty,\cX} \leq \log  b ,~ \forall f \in \tilde \cF^{\mathrm{NN}}(L_{\varepsilon,k,a},K,M,1).
\end{align}
Moreover, it  is easy to see that~\eqref{eq:polyapprox-deepnn} continues to hold with $f$ replaced by the corresponding $\tilde f$ since $\abs{p(i/d)} \leq \log  b $ for all $i \in [1:d]$. 

Next, we obtain a bound on the covering entropy of $\tilde \cF^{\mathrm{NN}}(L_{\varepsilon,k,a},K,M,1)$. We have
\begin{align}
  N \Big(\epsilon,\tilde \cF^{\mathrm{NN}}(L,K,M,1), \norm{\mspace{1 mu}\cdot\mspace{1 mu}}_{\infty,\cX}\Big) \leq  N \Big(\epsilon, \cF^{\mathrm{NN}}(L,K,M,1), \norm{\mspace{1 mu}\cdot \mspace{1 mu}}_{\infty,\cX}\Big). \notag 
\end{align}
Applying the covering entropy bound\footnote{This bound applies since covering number of a truncated function class is upper bounded by that of the original function class.}  given in~\cite[Lemma 5]{Schmidt-Hieber-2020}  to the neural class $\tilde \cF^{\mathrm{NN}}(L) \coloneqq \tilde \cF^{\mathrm{NN}}(L,9,1,1)$  
yields (with $p_0=p_{L}=1$ and $p_l=9$ for $1 \leq l \leq L-1 $ in the statement therein) that  
\begin{align}
 \log  N\big(\epsilon,\tilde \cF^{\mathrm{NN}}(L), \norm{\mspace{1 mu}\cdot \mspace{1 mu}}_{\infty,\cX} \big) 
  & \leq (28+81(L-2)) \log \! \Big(32L 10^{2(L-1)} \epsilon^{-1}\Big)   \leq 81 L  \log \!  \left(10^{3L} \epsilon^{-1}\right). \notag
\end{align}
From this, we obtain using $\log x \leq x$ for all $x \geq 0$ and $\sqrt{x+y} \leq \sqrt{x}+\sqrt{y}$ for all $x,y \geq 0$ that
\begin{align} \int_{0}^{1} \sqrt{\log   N (2\epsilon \log  b,\tilde \cF^{\mathrm{NN}}(L), \norm{\cdot}_{\infty,\cX}) } \, d\epsilon &\leq 24 L+ 9\sqrt{\frac{L}{2\log b}}   \int_{0}^{1} \epsilon^{-\frac 12}\, d\epsilon 
  \leq 24L +13\sqrt{\frac{L}{\log b}}. \label{eq:coventbnddeepnn}
\end{align}
Hence, from~\eqref{eq:deepnnapprox} and~\eqref{eq:coventbnddeepnn}, we have that $\tilde \cF^{\mathrm{NN}}( L_{\varepsilon,k,a}) \subseteq \cF_b$ and $\tilde{\cS}_{k,d}(b,a, \varepsilon) \subseteq \cS_d(b)$  satisfies Assumption~\ref{Assump1}  with  $\varepsilon\big(\tilde \cF^{\mathrm{NN}}( L_{\varepsilon,k,a}),\tilde{\cS}_{k,d}(b,a,\varepsilon)\big)=2\varepsilon$ and $\kappa\big(\tilde \cF^{\mathrm{NN}}( L_{\varepsilon,k,a})\big) \leq  B_{\varepsilon,k,a,b}/\log b$. Then, applying Theorem~\ref{Thm:minimaxrisk-QNE} with  $\cF=\tilde \cF^{\mathrm{NN}}( L_{\varepsilon,k,a})$ and $\cS=\tilde{\cS}_{k,d}(b,a, \varepsilon)$  yields the desired claims.
\subsection{Proof of Theorem~\ref{Thm:minimaxrisk-QNE-Renyi}}\label{Sec:Thm:minimaxrisk-QNE-Renyi-proof}
  The proof largely follows via  steps similar to those in the proof of Theorem~\ref{Thm:minimaxrisk-QNE}. So, we only highlight the differences. Fix $(\rho,\sigma) \in \cS$. 
   Let 
 \begin{align}
\sD_{\mathsf{M},\alpha,\Theta(\delta,\cS),\cF}(\rho\Vert\sigma)  \coloneqq \sup_{\bm{\theta} \in \Theta(\delta,\cS)}\sD_{\mathsf{M},\alpha, \bm{\theta},\cF}(\rho\Vert\sigma), \notag
\end{align}
where 
\begin{align}
\sD_{\mathsf{M},\alpha,\bm{\theta},\cF}(\rho\Vert\sigma) \coloneqq \sup_{H \in \cH \big(\bm{\theta},\cF\big)}\frac{\alpha}{\alpha-1}\log  \operatorname{Tr}[e^{\left(  \alpha-1\right)  H}\rho
]-\log  \operatorname{Tr}[e^{\alpha H}\sigma],\label{eq:inneropt-Ren}
\end{align}
and $\cH \big(\bm{\theta},\cF\big)$ is as defined in~\eqref{eq:Hermopfixedqcir}. 
Then,  the  triangle inequality implies that
\begin{align}
  \abs{\sD_{\mathsf{M},\alpha}(\rho\Vert\sigma) -\hat\sD_{\mathsf{M},\alpha,\Theta(\delta,\cS),\cF}^n}  &\leq   \abs{\sD_{\mathsf{M},\alpha}(\rho\Vert\sigma) -\sD_{\mathsf{M},\alpha,\Theta(\delta,\cS),\cF}(\rho\Vert\sigma)}  \notag \\
  & \qquad \qquad \qquad  + \abs{\sD_{\mathsf{M},\alpha,\Theta(\delta,\cS),\cF}(\rho\Vert\sigma) -\hat\sD_{\mathsf{M},\alpha,\Theta(\delta,\cS),\cF}^n}, \label{eq:effective-err-Renyi}
\end{align}
where the first and second terms on the right-hand side are the approximation and estimation errors, respectively.

Let $H_{\alpha}^{\star} \coloneqq H_{\alpha}^{\star}(\rho,\sigma)$ be an optimizer of~\eqref{eq:renyi_rel_ent_variational} satisfying~\eqref{eq:opteigvalues} and~\eqref{eq:optrenmeasrelent}.  Recall from the definition of the set $\cU_{\alpha}(\delta,\cS)$ that there exists an enumeration $\mathbf{P}_{\alpha}\coloneqq\{P_{i,\alpha}^{\star}\}_{i=1}^d $ of  orthogonal rank-one eigenprojectors of $H_{\alpha}^{\star}$, 
and  $U(\bm{\theta}^{\star}) $ for some $\bm{\theta}^{\star} \in \Theta_{\alpha}(\delta,\cS)$ such that~\eqref{eq:unitaryapprox} holds with $U^{\star}=U^{\star}(\mathbf{P}_{\alpha})$ and $U(\bm{\theta})=U(\bm{\theta}^{\star})$, where $U^{\star}(\mathbf{P}_{\alpha})$ is the unitary such that  $U^{\star}(\mathbf{P}_{\alpha})\ket{i}\!\bra{i}(U^{\star}(\mathbf{P}_{\alpha})\big)^{\dag}=P_{i,\alpha}^{\star}$.  Let $H_{\alpha}^{\star}=\sum_{i=1}^D \lambda_{i,\alpha}^{\star} P_{i,\alpha}^{\star}$ be the spectral decomposition of $H_{\alpha}^{\star}$ with $\Lambda_{\alpha}^{\star}\coloneqq \Lambda_{\alpha}^{\star}(\rho,\sigma)$ denoting the set of eigenvalues. By definition of $\varepsilon_{\alpha}(\cF,\cS)$,  there exists $f^{\star} \in \cF$ such that
\begin{align}
    \max_{i \in [1:d]} \abs{f^{\star}(e_m(i))-\lambda_{i,\alpha}^{\star}} \leq \varepsilon_{\alpha}(\cF,\cS).\label{eq:apprcnnren}
\end{align} 
  Let   
$H_{f^{\star},\bm{\theta}^{\star}}$ be as defined~\eqref{eq:paramhermop}. 
Then,  we can bound the approximation error as 
\begin{align}
&\abs{\sD_{\mathsf{M,\alpha}}(\rho\Vert\sigma) -\sD_{\mathsf{M,\alpha},\Theta_{\alpha}(\delta,\cS),\cF}(\rho\Vert\sigma)} \notag \\
& \stackrel{(a)}{\leq} \frac{\alpha}{\abs{\alpha-1}}\Big\|e^{(\alpha-1)H_{\alpha}^{\star}}-e^{(\alpha-1)H_{f^{\star},\bm{\theta}^{\star}}}\Big\|+\Big\|e^{\alpha H_{\alpha}^{\star}}-e^{\alpha H_{f^{\star},\bm{\theta}^{\star}}}\Big\| \notag \\
& \stackrel{(b)}{\leq} \alpha \left(b^{\abs{\alpha-1}}+b^{\alpha}\right) \varepsilon_{\alpha}(\cF,\cS) +2\left(\frac{\alpha}{\abs{\alpha-1}} b^{\abs{\alpha-1}}+b^{\alpha}\right)\delta, \label{eq:approxbnd-Renyi}
\end{align}
where $(a)$ follows via similar steps  leading to~\eqref{eq:euppbndefferr} and $(b)$ is via an application of Lemma~\ref{Lem:opnorm-diff} as leading to~\eqref{eq:approxerr}. For instance,  
\begin{align}
   \norm{e^{\alpha H_{\alpha}^{\star}}-e^{\alpha 
H_{f^{\star},\bm{\theta}^{\star}}}} 
&\leq \norm{e^{\alpha \Lambda_{\alpha}^{\star}}-e^{\alpha 
 F^{\star}}} + \left(\norm{e^{\alpha \Lambda_{\alpha}^{\star}}}+\norm{e^{\alpha 
 F^{\star}}}\right) \norm{U^{\star}(\mathbf{P}_{\alpha})-U(\bm{\theta}^{\star})}\notag \\
 & \leq \alpha  b^{\alpha} \max_{i \in [1:d]} \abs{\lambda_{i,\alpha}^{\star}- f^{\star}\!\left(e_m(i)\right)} + 2 b^{\alpha} \delta\notag \\
 & \leq  b^{\alpha} \big(\alpha \varepsilon_{\alpha}(\cF,\cS) + 2\delta\big), \notag
\end{align}
where the first inequality is by application of Lemma~\ref{Lem:opnorm-diff}, and second inequality uses that the Lipschitz constant of  $e^{\alpha x}$ is  bounded by $\alpha e^{\alpha \log  b }=\alpha b^{\alpha}$ for $\abs{x} \leq \log  b $ along with 
$\norm{e^{\alpha f^{\star}}}_{\infty,\cX} \leq  b^{\alpha} $.  Similarly, 
\begin{align}
    \norm{e^{(\alpha-1) H_{\alpha}^{\star}}-e^{(\alpha-1) 
 H_{f^{\star},\bm{\theta}^{\star}}}} &\leq b^{\abs{\alpha-1}} \big(\abs{\alpha-1} \varepsilon_{\alpha}(\cF,\cS) + 2\delta\big). \notag 
\end{align}
Also, following analogous steps leading to~\eqref{eq:expempesterrbnd}, we have
 \begin{align}
  & \EE \!\left[\abs{\sD_{\mathsf{M},\alpha,\Theta_{\alpha}(\delta,\cS),\cF}(\rho\Vert\sigma) -\hat\sD_{\mathsf{M},\alpha,\Theta_{\alpha}(\delta,\cS),\cF}^n}\right] \notag \\
& \leq 24  n^{-\frac 12} \alpha b^{\abs{\alpha-1}} \int_{0}^{\infty}\sqrt{ \log  N\big(\epsilon,\cF,\|\cdot\|_{\infty,\cX}\big)}d\epsilon + 24 n^{-\frac 12} \alpha b^{\alpha}\int_{0}^{\infty}\sqrt{ \log  N\big(\epsilon, \cF,\|\cdot\|_{\infty,\cX}\big)} d\epsilon \notag \\
&\leq 48n^{-\frac{1}{2}} \alpha \big(b^{\alpha}+b^{\abs{\alpha-1}}\big) \log  b \int_{0}^{1}\sqrt{ \log  N\big(2\epsilon \log  b,\cF,\|\cdot\|_{\infty,\cX}\big)}d\epsilon \notag \\
&\stackrel{(c)}{\leq } 48\alpha \big(b^{\alpha}+b^{\abs{\alpha-1}}\big) \log  b ~   \kappa(\cF) n^{-\frac 12}.\label{eq:expempesterrbnd-Renyi}
\end{align}
The claim in~\eqref{eq:minimax risk:Renyi} then follows by substituting~\eqref{eq:approxbnd-Renyi} and~\eqref{eq:expempesterrbnd-Renyi} in~\eqref{eq:effective-err-Renyi}, and taking supremum over $(\rho,\sigma) \in \cS$. The proof of~\eqref{eq:tailineqmeasrenyi}  follows analogous to that of~\eqref{eq:tailineQNEmrel} and hence is omitted.

\subsection{Proof of Proposition~\ref{Prop:copycomp-QNE}} 
\label{Sec:Prop:effqucirccomp-proof}
Recalling that $a (\delta,b)\coloneqq 2\delta   (b+\log  b )$ and applying Corollary~\ref{Cor-shallownn} yields that for every  $(\rho,\sigma) \in \bar{\cS}_d(\delta, b,\cU)$, 
\begin{align}
   & \mathbb{P}\left(\abs{\sD_{\mathsf{M}}(\rho\Vert\sigma) -\hat\sD_{\mathsf{M},\Theta(\cU),\tilde \cF^{\mathrm{SNN}}(d,\log  b )}^n} \geq a (\delta,b) +128(\log  b) (b+1) \left( 2d^{\frac 12}n^{-\frac 12}  + z \right)\right) \leq 2e^{-n  z^2}.\notag 
\end{align}
 Take $\delta=\epsilon_b\coloneqq \epsilon/(4b+4 \log b)$.  Setting $2e^{-nz^2}=1-p$ and solving for $z$ yields $z=n^{-\frac 12} \sqrt{\log \big(2/(1-p)\big)}$. Then, taking 
\begin{align}
    0.5 \epsilon=c' n^{-\frac 12}   (b+1) \Big( d^{\frac 12}  +  \sqrt{\log \big(2/(1-p)\big)}\Big)\log  b, \notag
\end{align}
for some constant $c'>0$, solving for $n$ and multiplying this with the upper bound for $\big|\Theta\big(\cU\big)\big|$  yields the desired upper bound  on copy complexity. 
\subsection{Proof of Proposition~\ref{Prop:copcomp-perminv}}\label{Sec:Prop:copcomp-perminv-proof}
Consider the Schur--Weyl decomposition, $\HH_d=\HH_m^{\otimes N}=\bigoplus_{\kappa \in Y_m^N} \cW_{\kappa} \otimes \cV_{\kappa}$,  given in~\eqref{eq:SWduality}. Recall that  $P_{\kappa}$ denotes the projection onto $\cW_{\kappa} \otimes \cV_{\kappa}$ and  $\cK_{m,N}\coloneqq \bigoplus_{\kappa} \cW_{\kappa}$. Consider the  quantum channels $\bar{\cN}(\cdot)$ and $\bar{\cM}(\cdot)$ given in~\eqref{eq:schurweylproj} and~\eqref{eq:invschurweylproj}, respectively. 
 Note that any $\rho \in \mathfrak{P}\big(\HH_m^{\otimes N}\big)$  has the representation 
\begin{align}
    \rho= \bigoplus_{\kappa \in Y_m^N} \bar{\rho}_{\kappa} \otimes \pi_{\kappa}, \notag
\end{align}
for some endomorphism  $\bar{\rho}_{\kappa} $ on $\cW_{\kappa}$. Hence, $(\rho,\sigma) \in \hat{\cS}_{m,N}(\delta, b,\cU)$ implies  that $\rho,\sigma \in \mathfrak{P}\big(\HH_m^{\otimes N}\big)$. Also observe that  if $\rho >0$, then $\bar{\rho}_{\kappa} >0$ and for $\rho \in \mathfrak{P}\big(\HH_m^{\otimes N}\big) $, $\bar{\cM} \circ \bar{\cN} (\rho)=\rho$. By the data-processing inequality for measured relative entropy, we have for any $\rho,\sigma \in \mathfrak{P}\big(\HH_m^{\otimes N}\big) $:
\begin{align}
  \sD_{\mathsf{M}}( \bar{\cN} (\rho)\Vert  \bar{\cN} (\sigma))  \leq \sD_{\mathsf{M}}(\rho\Vert\sigma)= \sD_{\mathsf{M}}(\bar{\cM} \circ \bar{\cN} (\rho)\Vert\bar{\cM} \circ \bar{\cN} (\sigma)) \leq \sD_{\mathsf{M}}( \bar{\cN} (\rho)\Vert  \bar{\cN} (\sigma)).\label{eq:datprocarg}
\end{align}
This yields
\begin{align}
 \sD_{\mathsf{M}}( \bar{\cN} (\rho)\Vert  \bar{\cN} (\sigma))=\sD_{\mathsf{M}}(\rho \Vert  \sigma), ~\forall ~(\rho,\sigma) \in \hat{\cS}_{m,N}(\delta, b,\cU).   \notag
\end{align}
The above observations provide a way to estimate the measured relative entropy between permutation invariant operators efficiently. Note that $\bar{\cN} (\rho)$ and $\bar{\cN} (\sigma)$ are density operators over a much smaller space $\cK_{m,N}$, whose dimension $\bar d_{m,N}$ is bounded from above by $(N+1)^{m-1} N^{\frac{m(m-1)}{2}}$ (see e.g.,~\cite[Equation 6.16]{Hayashi-2017}). By the variational form for measured relative entropy, we have
\begin{align}
\sD_{\mathsf{M}}(\rho\Vert\sigma)  &=\sup_{H}\operatorname{Tr}[H\bar{\cN}(\rho)
]-\operatorname{Tr}[e^H\bar{\cN}(\sigma)]+1, \notag 
\end{align}
where the supremum is taken over Hermitian operators  $H$ on $\cK_{m,N}$. Since  $(\rho,\sigma) \in \hat{\cS}_{m,N}(\delta, b,\cU)$ implies that $\big(\bar{\cN}(\rho),\bar{\cN}(\sigma)\big) \in \bar{\cS}_{|\cK_{m,N}|}(\delta, b,\cU)$, we have  $ T\big(\bar{\cN}(\rho),\bar{\cN}(\sigma)\big) \leq  \log b$ by definition of $\bar{\cS}_{|\cK_{m,N}|}(\delta, b,\cU)$. From these,  the claim in Proposition~\ref{Prop:copcomp-perminv} follows by an application of Corollary~\ref{Cor-shallownn} via similar steps as in the proof of Proposition~\ref{Prop:copycomp-QNE}. We also remark that since the Thompson metric contracts under application of a quantum channel, we have $T\big(\bar{\cN} (\rho),\bar{\cN} (\sigma)\big) =T(\rho,\sigma)$ for all   $(\rho,\sigma) \in \hat{\cS}_{m,N}(\delta, b,\cU)$ via analogous  arguments as in~\eqref{eq:datprocarg}.

\section{Concluding Remarks}

\label{Sec:concl-remarks}

We studied the performance of the quantum neural estimator proposed in~\cite{QNE-24} for estimating  measured relative entropy and its R\'{e}nyi analogue. We obtained upper bounds on its expected absolute error and showed exponential deviation inequalities under regularity conditions on the underlying states, the neural class, and the parametrized quantum circuit. We then specialized our bounds to  shallow and deep neural nets and derived concrete bounds using known results from the approximation theory of neural nets along with bounds on the covering entropy of the associated neural classes. We also established that, although the quantum neural estimator has a worst case copy complexity that scales exponentially in the number of qudits, this improves  to polynomial when the underlying states are permutation invariant.

Several open questions remain in the context of quantum neural estimation of entropies. One such issue pertains to the quantification of the optimization error and an analysis of optimization dynamics such as the investigation of barren plateaus and local optima, which was neglected here. This would call for an in-depth study of the hybrid optimization landscape induced by the neural net and quantum circuit, and an  analysis of specific parameterizations or architectures that lead to relatively fast convergence of the algorithm. Also relevant is the question of estimating the quantum relative entropy and its R\'{e}nyi generalizations such as the sandwiched and Petz--R\'{e}nyi  relative entropies using quantum neural estimators. Recent progress on estimating quantum relative entropy has been reported in \cite{lu2025estimatingquantumrelativeentropies}.

\section*{Acknowledgements}
S. Sreekumar acknowledges support from the Centre National de la Recherche Scientifique (CNRS) and the CentraleSup\'elec$\mspace{1 mu}$-$\mspace{1 mu}$L2S funding  WRP623. Z. Goldfeld is partially supported by National Science Foundation (NSF) grants CCF-2046018 and CCF-2308446, and the
IBM Academic Award. MMW acknowledges support from the NSF under Grant No.~2329662.

\addcontentsline{toc}{section}{References}

\bibliographystyle{quantum}
\bibliography{ref}

\begin{thebibliography}{10}

\bibitem{NeumannThermodynamikGesamtheiten}
John von Neumann.
\newblock ``Thermodynamik quantenmechanischer gesamtheiten''.
\newblock Nachrichten von der Gesellschaft der Wissenschaften zu G{\"{o}}ttingen, Mathematisch-Physikalische Klasse {\bf 1927}, 273--291~(1927).
\newblock  url:~\url{http://eudml.org/doc/59231}.

\bibitem{Shannon1948ACommunication}
Claude~E. Shannon.
\newblock ``A mathematical theory of communication''.
\newblock \href{https://dx.doi.org/10.1002/J.1538-7305.1948.TB01338.X}{Bell System Technical Journal {\bf 27}, 379--423}~(1948).

\bibitem{Renyi1961OnInformation}
Alfr\'ed R\'enyi.
\newblock ``On measures of entropy and information''.
\newblock Proceedings of the Fourth Berkeley Symposium on Mathematical Statistics and Probability, Volume 1: Contributions to the Theory of Statistics {\bf 4.1}, 547--561~(1961).
\newblock  url:~\url{https://digicoll.lib.berkeley.edu/record/112906}.

\bibitem{Kullback1951OnSufficiency}
Solomon Kullback and Richard~A. Leibler.
\newblock ``On information and sufficiency''.
\newblock \href{https://dx.doi.org/10.1214/aoms/1177729694}{The Annals of Mathematical Statistics {\bf 22}, 79--86}~(1951).

\bibitem{U62}
Hisaharu Umegaki.
\newblock ``Conditional expectations in an operator algebra {IV} (entropy and information)''.
\newblock \href{https://dx.doi.org/10.2996/kmj/1138844604}{Kodai Mathematical Seminar Reports {\bf 14}, 59--85}~(1962).

\bibitem{Petz1985Quasi-entropiesAlgebra}
Dénes Petz.
\newblock ``Quasi-entropies for states of a von {N}eumann algebra''.
\newblock \href{https://dx.doi.org/10.2977/PRIMS/1195178929}{Publications of the Research Institute for Mathematical Sciences {\bf 21}, 787--800}~(1985).

\bibitem{Petz1986Quasi-entropiesSystems}
Dénes Petz.
\newblock ``Quasi-entropies for finite quantum systems''.
\newblock \href{https://dx.doi.org/10.1016/0034-4877(86)90067-4}{Reports on Mathematical Physics {\bf 23}, 57--65}~(1986).

\bibitem{muller2013quantum}
Martin M{\"u}ller-Lennert, Fr{\'e}d{\'e}ric Dupuis, Oleg Szehr, Serge Fehr, and Marco Tomamichel.
\newblock ``On quantum {R\'e}nyi entropies: A new generalization and some properties''.
\newblock \href{https://dx.doi.org/10.1063/1.4838856}{Journal of Mathematical Physics {\bf 54}, 122203}~(2013).

\bibitem{wilde2014strong}
Mark~M. Wilde, Andreas Winter, and Dong Yang.
\newblock ``Strong converse for the classical capacity of entanglement-breaking and {H}adamard channels via a sandwiched {R\'e}nyi relative entropy''.
\newblock \href{https://dx.doi.org/10.1007/s00220-014-2122-x}{Communications in Mathematical Physics {\bf 331}, 593--622}~(2014).

\bibitem{ogawa-nagaoka-2000}
Tomohiro Ogawa and Hiroshi Nagaoka.
\newblock ``Strong converse and {Stein's} lemma in quantum hypothesis testing''.
\newblock \href{https://dx.doi.org/10.1109/18.887855}{IEEE Transactions on Information Theory {\bf 46}, 2428--2433}~(2000).

\bibitem{ogawa-hayashi-2004}
Tomohiro Ogawa and Masahito Hayashi.
\newblock ``On error exponents in quantum hypothesis testing''.
\newblock \href{https://dx.doi.org/10.1109/TIT.2004.828155}{IEEE Transactions on Information Theory {\bf 50}, 1368--1372}~(2004).

\bibitem{Nussbaum-Szkola-2006}
Michael Nussbaum and Arleta Szkoła.
\newblock ``{The Chernoff lower bound for symmetric quantum hypothesis testing}''.
\newblock \href{https://dx.doi.org/10.1214/08-AOS593}{The Annals of Statistics {\bf 37}, 1040--1057}~(2009).

\bibitem{Audenaert-2008}
Koenraad M.~R. Audenaert, Michael Nussbaum, Arleta Szkola, and Frank Verstraete.
\newblock ``{Asymptotic error rates in quantum hypothesis testing}''.
\newblock \href{https://dx.doi.org/https://doi.org/10.1007/s00220-008-0417-5}{Communications in Mathematical Physics {\bf {279}}, {251--283}}~({2008}).

\bibitem{BBH-2021}
Mario Berta, Fernando Brandao, and Christoph Hirche.
\newblock ``On composite quantum hypothesis testing''.
\newblock \href{https://dx.doi.org/10.1007/s00220-021-04133-8}{Communications in Mathematical Physics {\bf 385}, 55--77}~(2021).

\bibitem{Li2017QuantumEstimation}
Tongyang Li and Xiaodi Wu.
\newblock ``Quantum query complexity of entropy estimation''.
\newblock \href{https://dx.doi.org/10.1109/TIT.2018.2883306}{IEEE Transactions on Information Theory {\bf 65}, 2899--2921}~(2017).

\bibitem{Subasi2019EntanglementCircuit}
Yigit Subasi, Lukasz Cincio, and Patrick~J. Coles.
\newblock ``Entanglement spectroscopy with a depth-two quantum circuit''.
\newblock \href{https://dx.doi.org/10.1088/1751-8121/aaf54d}{Journal of Physics A: Mathematical and Theoretical {\bf 52}, 44001}~(2019).

\bibitem{AISW-2020}
Jayadev Acharya, Ibrahim Issa, Nirmal~V. Shende, and Aaron~B. Wagner.
\newblock ``Estimating quantum entropy''.
\newblock \href{https://dx.doi.org/10.1109/JSAIT.2020.3015235}{IEEE Journal on Selected Areas in Information Theory {\bf 1}, 454--468}~(2020).

\bibitem{Subramanian2019QuantumStates}
Sathyawageeswar Subramanian and Min-Hsiu Hsieh.
\newblock ``Quantum algorithm for estimating $\ensuremath{\alpha}$-{R\'enyi} entropies of quantum states''.
\newblock \href{https://dx.doi.org/10.1103/PhysRevA.104.022428}{Physical Review A {\bf 104}, 022428}~(2021).

\bibitem{Gilyen2019DistributionalWorld}
Andr{\'a}s Gily{\'e}n and Tongyang Li.
\newblock ``Distributional property testing in a quantum world''.
\newblock In Thomas Vidick, editor, 11th Innovations in Theoretical Computer Science Conference (ITCS 2020).
\newblock \href{https://dx.doi.org/https://doi.org/10.4230/LIPIcs.ITCS.2020.25}{Volume 151 of Leibniz International Proceedings in Informatics (LIPIcs), pages 25:1--25:19}.
\newblock Dagstuhl, Germany~(2020). Schloss Dagstuhl--Leibniz-Zentrum fuer Informatik.

\bibitem{Gur2021SublinearEntropy}
Tom Gur, Min-Hsiu Hsieh, and Sathyawageeswar Subramanian.
\newblock ``Sublinear quantum algorithms for estimating von {N}eumann entropy''~(2021).
\newblock  \href{http://arxiv.org/abs/2111.11139}{arXiv:2111.11139}.

\bibitem{Wang2022QuantumEntropies}
Youle Wang, Benchi Zhao, and Xin Wang.
\newblock ``Quantum algorithms for estimating quantum entropies''.
\newblock \href{https://dx.doi.org/10.1103/PhysRevApplied.19.044041}{Physical Review Applied {\bf 19}, 044041}~(2023).

\bibitem{Wang2022NewDistances}
Qisheng Wang, Ji~Guan, Junyi Liu, Zhicheng Zhang, and Mingsheng Ying.
\newblock ``New quantum algorithms for computing quantum entropies and distances''.
\newblock \href{https://dx.doi.org/10.1109/TIT.2024.3399014}{IEEE Transactions on Information Theory {\bf 70}, 5653–5680}~(2024).
\newblock  \href{http://arxiv.org/abs/2203.13522}{arXiv:2203.13522}.

\bibitem{Gilyen2022ImprovedEstimation}
András Gilyén and Alexander Poremba.
\newblock ``Improved quantum algorithms for fidelity estimation''~(2022).
\newblock  \href{http://arxiv.org/abs/2203.15993}{arXiv:2203.15993}.

\bibitem{Wang2022QuantumSystems}
Youle Wang, Lei Zhang, Zhan Yu, and Xin Wang.
\newblock ``Quantum phase processing and its applications in estimating phase and entropies''.
\newblock \href{https://dx.doi.org/10.1103/PhysRevA.108.062413}{Physical Review A {\bf 108}, 062413}~(2023).

\bibitem{Wang2023QuantumEstimation}
Qisheng Wang, Zhicheng Zhang, Kean Chen, Ji~Guan, Wang Fang, Junyi Liu, and Mingsheng Ying.
\newblock ``Quantum algorithm for fidelity estimation''.
\newblock \href{https://dx.doi.org/10.1109/TIT.2022.3203985}{IEEE Transactions on Information Theory {\bf 69}, 273--282}~(2023).

\bibitem{Wang-Zhang-Li-2024}
Xinzhao Wang, Shengyu Zhang, and Tongyang Li.
\newblock ``A quantum algorithm framework for discrete probability distributions with applications to {Rényi} entropy estimation''.
\newblock \href{https://dx.doi.org/10.1109/TIT.2024.3382037}{IEEE Transactions on Information Theory {\bf 70}, 3399--3426}~(2024).

\bibitem{Wang-Zhang-2025}
Qisheng Wang and Zhicheng Zhang.
\newblock ``Time-efficient quantum entropy estimator via samplizer''.
\newblock \href{https://dx.doi.org/10.1109/TIT.2025.3576137}{IEEE Transactions on Information Theory {\bf 71}, 9569--9599}~(2025).

\bibitem{lu2025estimatingquantumrelativeentropies}
Yuchen Lu and Kun Fang.
\newblock ``Estimating quantum relative entropies on quantum computers''.
\newblock \href{https://dx.doi.org/10.1038/s42005-026-02655-y}{Communications Physics}~(2026).
\newblock  \href{http://arxiv.org/abs/2501.07292}{arXiv:2501.07292}.

\bibitem{QNE-24}
Ziv Goldfeld, Dhrumil Patel, Sreejith Sreekumar, and Mark~M. Wilde.
\newblock ``Quantum neural estimation of entropies''.
\newblock \href{https://dx.doi.org/10.1103/PhysRevA.109.032431}{Physical Review A {\bf 109}, 032431}~(2024).

\bibitem{Shin2024}
Myeongjin Shin, Junseo Lee, and Kabgyun Jeong.
\newblock ``Estimating quantum mutual information through a quantum neural network''.
\newblock \href{https://dx.doi.org/10.1007/s11128-023-04253-1}{Quantum Information Processing {\bf 23}, 57}~(2024).

\bibitem{lee2023estimating}
Sangyun Lee, Hyukjoon Kwon, and Jae~Sung Lee.
\newblock ``Estimating entanglement entropy via variational quantum circuits with classical neural networks''.
\newblock \href{https://dx.doi.org/10.1103/PhysRevE.109.044117}{Physical Review E {\bf 109}, 044117}~(2024).

\bibitem{Petz1988}
Dénes Petz.
\newblock ``A variational expression for the relative entropy''.
\newblock \href{https://dx.doi.org/10.1007/BF01225040}{Communications in Mathematical Physics {\bf 114}, 345--349}~(1988).

\bibitem{petz2007quantum}
D{\'e}nes Petz.
\newblock ``Quantum information theory and quantum statistics''.
\newblock \href{https://dx.doi.org/10.1007/978-3-540-74636-2}{Springer Science \& Business Media}. ~(2007).

\bibitem{Berta2015OnEntropies}
Mario Berta, Omar Fawzi, and Marco Tomamichel.
\newblock ``On variational expressions for quantum relative entropies''.
\newblock \href{https://dx.doi.org/10.1007/s11005-017-0990-7}{Letters in Mathematical Physics {\bf 107}, 2239--2265}~(2015).

\bibitem{thompson_1963}
Anthony~C. Thompson.
\newblock ``On certain contraction mappings in a partially ordered vector space''.
\newblock \href{https://dx.doi.org/10.1090/S0002-9939-1963-0149237-7}{Proceedings of the American Mathematical Society {\bf 14}, 438--443}~(1963).

\bibitem{Donald1986}
Matthew~J. Donald.
\newblock ``On the relative entropy''.
\newblock \href{https://dx.doi.org/10.1007/BF01212339}{Communications in Mathematical Physics {\bf 105}, 13--34}~(1986).

\bibitem{P09}
Marco Piani.
\newblock ``Relative entropy of entanglement and restricted measurements''.
\newblock \href{https://dx.doi.org/10.1103/PhysRevLett.103.160504}{Physical Review Letters {\bf 103}, 160504}~(2009).

\bibitem{Hiai-Petz-1991}
Fumio Hiai and D{\'e}nes Petz.
\newblock ``{The proper formula for relative entropy and its asymptotics in quantum probability}''.
\newblock \href{https://dx.doi.org/10.1007/BF02100287}{Communications in Mathematical Physics {\bf 143}, 99--114}~(1991).

\bibitem{Hayashi_2002}
Masahito Hayashi.
\newblock ``Optimal sequence of quantum measurements in the sense of {Stein's} lemma in quantum hypothesis testing''.
\newblock \href{https://dx.doi.org/10.1088/0305-4470/35/50/307}{Journal of Physics A: Mathematical and General {\bf 35}, 10759}~(2002).

\bibitem{Mosonyi2015QuantumHT}
Mil{\'a}n Mosonyi and Tomohiro Ogawa.
\newblock ``Quantum hypothesis testing and the operational interpretation of the quantum {R{\'e}nyi} relative entropies''.
\newblock \href{https://dx.doi.org/10.1007/s00220-014-2248-x}{Communications in Mathematical Physics {\bf 334}, 1617--1648}~(2015).

\bibitem{Brando2020AdversarialHT}
Fernando Brand{\~a}o, Aram~W. Harrow, James~R. Lee, and Yuval Peres.
\newblock ``Adversarial hypothesis testing and a quantum {Stein’s} lemma for restricted measurements''.
\newblock \href{https://dx.doi.org/10.1109/TIT.2020.2979704}{IEEE Transactions on Information Theory {\bf 66}, 5037--5054}~(2020).

\bibitem{RSB-IT-2025}
Tobias Rippchen, Sreejith Sreekumar, and Mario Berta.
\newblock ``Locally-measured {R}ényi divergences''.
\newblock \href{https://dx.doi.org/10.1109/TIT.2025.3571527}{IEEE Transactions on Information Theory {\bf 71}, 6105--6133}~(2025).

\bibitem{SHCB-2025}
Sreejith Sreekumar, Christoph Hirche, Hao-Chung Cheng, and Mario Berta.
\newblock ``Distributed quantum hypothesis testing under zero-rate communication constraints''.
\newblock \href{https://dx.doi.org/10.1007/s00023-025-01623-6}{Annales Henri Poincaré}~(2025).

\bibitem{Verdu-2019}
Sergio Verdú.
\newblock ``Empirical estimation of information measures: A literature guide''.
\newblock \href{https://dx.doi.org/10.3390/e21080720}{Entropy {\bf 21}, 720}~(2019).

\bibitem{SK-2025}
Sreejith Sreekumar and Kengo Kato.
\newblock ``Deviation inequalities for {R\'{e}nyi} divergence estimators via variational expression''~(2025).
\newblock  \href{http://arxiv.org/abs/2508.09382}{arXiv:2508.09382}.

\bibitem{HHJWY16}
Jeongwan Haah, Aram~W. Harrow, Zhengfeng Ji, Xiaodi Wu, and Nengkun Yu.
\newblock ``Sample-optimal tomography of quantum states''.
\newblock \href{https://dx.doi.org/10.1109/TIT.2017.2719044}{IEEE Transactions on Information Theory {\bf 63}, 5628--5641}~(2017).

\bibitem{wright2016learn}
John Wright.
\newblock ``How to learn a quantum state''.
\newblock PhD thesis.
\newblock Carnegie Mellon University.
\newblock ~(2016).
\newblock  url:~\url{https://people.eecs.berkeley.edu/~jswright/papers/thesis.pdf}.

\bibitem{OW16}
Ryan O'Donnell and John Wright.
\newblock ``Efficient quantum tomography''.
\newblock In Proceedings of the Forty-Eighth Annual ACM Symposium on Theory of Computing.
\newblock \href{https://dx.doi.org/10.1145/2897518.2897544}{Page 899–912}.
\newblock STOC '16,{ }New York, NY, USA~(2016). Association for Computing Machinery.

\bibitem{Yuen2023improvedsample}
Henry Yuen.
\newblock ``An improved sample complexity lower bound for (fidelity) quantum state tomography''.
\newblock \href{https://dx.doi.org/10.22331/q-2023-01-03-890}{Quantum {\bf 7}, 890}~(2023).

\bibitem{Biamonte2017QuantumLearning}
Jacob Biamonte, Peter Wittek, Nicola Pancotti, Patrick Rebentrost, Nathan Wiebe, and Seth Lloyd.
\newblock ``Quantum machine learning''.
\newblock \href{https://dx.doi.org/10.1038/nature23474}{Nature {\bf 549}, 195--202}~(2017).

\bibitem{Cerezo2020variationalquantum}
Marco Cerezo, Alexander Poremba, Lukasz Cincio, and Patrick~J. Coles.
\newblock ``Variational {Q}uantum {F}idelity {E}stimation''.
\newblock \href{https://dx.doi.org/10.22331/q-2020-03-26-248}{{Quantum} {\bf 4}, 248}~(2020).

\bibitem{Cerezo2021VariationalAlgorithms}
M.~Cerezo, Andrew Arrasmith, Ryan Babbush, Simon~C. Benjamin, Suguru Endo, Keisuke Fujii, Jarrod~R. McClean, Kosuke Mitarai, Xiao Yuan, Lukasz Cincio, and Patrick~J. Coles.
\newblock ``Variational quantum algorithms''.
\newblock \href{https://dx.doi.org/10.1038/s42254-021-00348-9}{Nature Reviews Physics {\bf 3}, 625--644}~(2021).

\bibitem{Jaeger2023}
Jonas Jäger and Roman~V. Krems.
\newblock ``Universal expressiveness of variational quantum classifiers and quantum kernels for support vector machines''.
\newblock \href{https://dx.doi.org/10.1038/s41467-023-36144-5}{Nature Communications {\bf 14}, 576}~(2023).

\bibitem{NWJ10}
XuanLong Nguyen, Martin~J. Wainwright, and Michael~I. Jordan.
\newblock ``Estimating divergence functionals and the likelihood ratio by convex risk minimization''.
\newblock \href{https://dx.doi.org/10.1109/TIT.2010.2068870}{IEEE Transactions on Information Theory {\bf 56}, 5847--5861}~(2010).

\bibitem{pmlr-v70-arora17a}
Sanjeev Arora, Rong Ge, Yingyu Liang, Tengyu Ma, and Yi~Zhang.
\newblock ``Generalization and equilibrium in generative adversarial nets ({GAN}s)''.
\newblock In Proceedings of the 34th International Conference on Machine Learning.
\newblock Volume~70 of Proceedings of Machine Learning Research, pages 224--232.
\newblock PMLR~(2017).
\newblock  url:~\url{https://proceedings.mlr.press/v70/arora17a.html}.

\bibitem{Zhang-2018}
Pengchuan Zhang, Qiang Liu, Dengyong Zhou, Tao Xu, and Xiaodong He.
\newblock ``On the discrimination-generalization tradeoff in {GAN}s''.
\newblock In 6th International Conference on Learning Representations.
\newblock ~(2018).
\newblock  url:~\url{https://openreview.net/forum?id=Hk9Xc_lR-}.

\bibitem{belghazi-2018}
Mohamed~Ishmael Belghazi, Aristide Baratin, Sai Rajeshwar, Sherjil Ozair, Yoshua Bengio, Aaron Courville, and Devon Hjelm.
\newblock ``Mutual information neural estimation''.
\newblock In Proceedings of the 35th International Conference on Machine Learning.
\newblock Volume~80 of Proceedings of Machine Learning Research, pages 531--540.
\newblock PMLR~(2018).
\newblock  url:~\url{https://proceedings.mlr.press/v80/belghazi18a.html}.

\bibitem{SS-2021-aistats}
Sreejith Sreekumar, Zhengxin Zhang, and Ziv Goldfeld.
\newblock ``Non-asymptotic performance guarantees for neural estimation of $f$-divergences''.
\newblock In Proceedings of the 24th International Conference on Artificial Intelligence and Statistics.
\newblock Volume 130, pages 3322--3330.
\newblock ~(2021).
\newblock  url:~\url{https://proceedings.mlr.press/v130/sreekumar21a.html}.

\bibitem{SG-2021-neural}
Sreejith Sreekumar and Ziv Goldfeld.
\newblock ``Neural estimation of statistical divergences''.
\newblock Journal of Machine Learning Research {\bf 23}, 1--75~(2022).
\newblock  url:~\url{https://www.jmlr.org/papers/v23/21-1212.html}.

\bibitem{Birell-2021}
Jeremiah Birrell, Paul Dupuis, Markos~A. Katsoulakis, Luc Rey-Bellet, and Jie Wang.
\newblock ``Variational representations and neural network estimation of {R}ényi divergences''.
\newblock \href{https://dx.doi.org/10.1137/20M1368926}{SIAM Journal on Mathematics of Data Science {\bf 3}, 1093--1116}~(2021).

\bibitem{VQA-entropies-fidelity-2021}
Kok~Chuan Tan and Tyler Volkoff.
\newblock ``Variational quantum algorithms to estimate rank, quantum entropies, fidelity, and {Fisher} information via purity minimization''.
\newblock \href{https://dx.doi.org/10.1103/PhysRevResearch.3.033251}{Physical Review Research {\bf 3}, 033251}~(2021).

\bibitem{McClean_2018}
Jarrod~R. McClean, Sergio Boixo, Vadim~N. Smelyanskiy, Ryan Babbush, and Hartmut Neven.
\newblock ``Barren plateaus in quantum neural network training landscapes''.
\newblock \href{https://dx.doi.org/10.1038/s41467-018-07090-4}{Nature Communications {\bf 9}, 4812}~(2018).

\bibitem{Cerezo2021CostDependent}
Marco Cerezo, Akira Sone, Tyler Volkoff, Lukasz Cincio, and Patrick~J. Coles.
\newblock ``Cost-function-dependent barren plateaus in shallow parameterized quantum circuits''.
\newblock \href{https://dx.doi.org/10.1038/s41467-021-21728-w}{Nature Communications {\bf 12}, 1791}~(2021).

\bibitem{Holmes2022Trainability}
Zo{\"e} Holmes, Kunal Sharma, Marco Cerezo, and Patrick~J. Coles.
\newblock ``Connecting ansatz expressibility to gradient magnitudes and barren plateaus''.
\newblock \href{https://dx.doi.org/10.1103/PRXQuantum.3.010313}{PRX Quantum {\bf 3}, 010313}~(2022).

\bibitem{Meyer2023ExploitingSymmetry}
Johannes~Jakob Meyer, Marian Mularski, Elies Gil‐Fuster, Antonio~Anna Mele, Francesco Arzani, Alissa Wilms, and Jens Eisert.
\newblock ``Exploiting symmetry in variational quantum machine learning''.
\newblock \href{https://dx.doi.org/10.1103/PRXQuantum.4.010328}{PRX Quantum {\bf 4}, 010328}~(2023).

\bibitem{Nguyen2024TheoryEquivariantQNN}
Quynh~T. Nguyen, Louis Schatzki, Paolo Braccia, Michael Ragone, Patrick~J. Coles, Frédéric Sauvage, Martín Larocca, and Maciej Cerezo.
\newblock ``Theory for equivariant quantum neural networks''.
\newblock \href{https://dx.doi.org/10.1103/PRXQuantum.5.020328}{PRX Quantum {\bf 5}, 020328}~(2024).

\bibitem{Grant2019Initialization}
Edward Grant, Leonard Wossnig, Mateusz Ostaszewski, and Marcello Benedetti.
\newblock ``An initialization strategy for addressing barren plateaus in parameterized quantum circuits''.
\newblock \href{https://dx.doi.org/10.22331/q-2019-12-09-214}{Quantum {\bf 3}, 214}~(2019).

\bibitem{Skolik2021Layerwise}
Andrea Skolik, Jarrod~R. McClean, Masoud Mohseni, Patrick van~der Smagt, and Matthew Leib.
\newblock ``Layerwise learning for quantum neural networks''.
\newblock \href{https://dx.doi.org/10.1007/s42484-020-00036-4}{Quantum Machine Intelligence {\bf 3}, 1--19}~(2021).

\bibitem{Cerezo2025ProvableAbsence}
Marco Cerezo, Martin Larocca, Diego Garc{\'\i}a-Mart{\'\i}n, N.~L. Diaz, Paolo Braccia, Enrico Fontana, Manuel~S. Rudolph, Pablo Bermejo, Aroosa Ijaz, Supanut Thanasilp, Eric~R. Anschuetz, and Zo\"{e} Holmes.
\newblock ``Does provable absence of barren plateaus imply classical simulability?''.
\newblock \href{https://dx.doi.org/10.1038/s41467-025-63099-6}{Nature Communications {\bf 16}, 7907}~(2025).

\bibitem{F96}
Christopher Fuchs.
\newblock ``Distinguishability and accessible information in quantum theory''.
\newblock PhD thesis.
\newblock University of New Mexico.
\newblock ~(1996).
\newblock  \href{http://arxiv.org/abs/quant-ph/9601020}{arXiv:quant-ph/9601020}.

\bibitem{huang2025semidefiniteoptimizationmeasuredrelative}
Zixin Huang and Mark~M. Wilde.
\newblock ``Semi-definite optimization of the measured relative entropies of quantum states and channels''~(2025).
\newblock  url:~\url{https://arxiv.org/abs/2406.19060}.

\bibitem{Paninski-2003}
Liam Paninski.
\newblock ``Estimation of entropy and mutual information''.
\newblock \href{https://dx.doi.org/10.1162/089976603321780272}{Neural Computation {\bf 15}, 1191–1253}~(2003).

\bibitem{Datta-2009}
Nilanjana Datta.
\newblock ``Min- and max-relative entropies and a new entanglement monotone''.
\newblock \href{https://dx.doi.org/10.1109/TIT.2009.2018325}{IEEE Transactions on Information Theory {\bf 55}, 2816--2826}~(2009).

\bibitem{bacon2006quantum}
Dave Bacon, Isaac~L. Chuang, and Aram~W. Harrow.
\newblock ``{Efficient Quantum Circuits for Schur and Clebsch-Gordan Transforms}''.
\newblock \href{https://dx.doi.org/10.1103/PhysRevLett.97.170502}{Physical Review Letters {\bf 97}, 170502}~(2006).

\bibitem{harrow2005applications}
Aram~W. Harrow.
\newblock ``Applications of coherent classical communication and the {S}chur transform to quantum information theory''.
\newblock PhD thesis.
\newblock Massachusetts Institute of Technology.
\newblock ~(2005).
\newblock  \href{http://arxiv.org/abs/quant-ph/0512255}{arXiv:quant-ph/0512255}.

\bibitem{Kirby2018schur}
William~M. Kirby and Frederick~W. Strauch.
\newblock ``{A practical quantum algorithm for the {S}chur transform}''.
\newblock \href{https://dx.doi.org/10.26421/QIC18.9-10-1}{Quantum Information and Computation {\bf 18}, 0721--0742}~(2018).

\bibitem{Krovi2019efficienthigh}
Hari Krovi.
\newblock ``An efficient high dimensional quantum {S}chur transform''.
\newblock \href{https://dx.doi.org/10.22331/q-2019-02-14-122}{{Quantum} {\bf 3}, 122}~(2019).

\bibitem{VanHandel-book}
Ramon {Van Handel}.
\newblock ``Probability in high dimension: Lecture notes-princeton university''.
\newblock [Online]. Available: \url{https://web.math.princeton.edu/~rvan/APC550.pdf}~(2016).

\bibitem{tikhomirov1993varepsilon}
Andrey~N. Kolmogorov and Vladimir~M. Tikhomirov.
\newblock ``$\varepsilon$-entropy and $\varepsilon$-capacity of sets in functional spaces''.
\newblock Uspekhi Matematicheskikh Nauk {\bf 14}, 3--86~(1959).
\newblock  url:~\url{http://mi.mathnet.ru/eng/rm7289}.

\bibitem{kandala2017hardware}
Abhinav Kandala, Antonio Mezzacapo, Kristan Temme, Maika Takita, Matthias Brink, Jerry~M. Chow, and Jay~M. Gambetta.
\newblock ``Hardware-efficient variational quantum eigensolver for small molecules and quantum magnets''.
\newblock \href{https://dx.doi.org/10.1038/nature23879}{Nature {\bf 549}, 242--246}~(2017).

\bibitem{lloyd2018universalqaoa}
Seth Lloyd.
\newblock ``Quantum approximate optimization is computationally universal''~(2018).
\newblock  \href{http://arxiv.org/abs/1812.11075}{arXiv:1812.11075}.

\bibitem{morales2020universality}
Mauro E.~S. Morales, Jacob~D. Biamonte, and Zolt\'{a}n Zimbor\'{a}s.
\newblock ``On the universality of the quantum approximate optimization algorithm''.
\newblock \href{https://dx.doi.org/10.1007/s11128-020-02748-9}{Quantum Information Processing {\bf 19}, 291}~(2020).

\bibitem{Yihong-Yang-2016}
Yihong Wu and Pengkun Yang.
\newblock ``Minimax rates of entropy estimation on large alphabets via best polynomial approximation''.
\newblock \href{https://dx.doi.org/10.1109/TIT.2016.2548468}{IEEE Transactions on Information Theory {\bf 62}, 3702--3720}~(2016).

\bibitem{Yanjun-2020-alphadiv}
Yanjun Han, Jiantao Jiao, and Tsachy Weissman.
\newblock ``Minimax estimation of divergences between discrete distributions''.
\newblock \href{https://dx.doi.org/10.1109/JSAIT.2020.3041036}{IEEE Journal on Selected Areas in Information Theory {\bf 1}, 814--823}~(2020).

\bibitem{Dennis-2021}
Dennis Elbr\"{a}chter, Dmytro Perekrestenko, Philipp Grohs, and Helmut B{\"o}lcskei.
\newblock ``Deep neural network approximation theory''.
\newblock \href{https://dx.doi.org/10.1109/TIT.2021.3062161}{IEEE Transactions on Information Theory {\bf 67}, 2581--2623}~(2021).

\bibitem{Nielsen_Chuang_2010}
Michael~A. Nielsen and Isaac~L. Chuang.
\newblock ``Quantum computation and quantum information: 10th anniversary edition''.
\newblock \href{https://dx.doi.org/10.1017/CBO9780511976667}{Cambridge University Press}. ~(2010).

\bibitem{Hayashi-2017}
Masahito Hayashi.
\newblock ``{A Group Theoretic Approach to Quantum Information}''.
\newblock \href{https://dx.doi.org/10.1007/978-3-319-45241-8}{Springer}. ~(2017).

\bibitem{AVDV-book}
Aad~W. van~der Vaart and Jon~A. Wellner.
\newblock ``Weak convergence and empirical processes''.
\newblock \href{https://dx.doi.org/10.1007/978-1-4757-2545-2}{Springer, New York}. ~(1996).

\bibitem{Schmidt-Hieber-2020}
Johannes Schmidt-Hieber.
\newblock ``Nonparametric regression using deep neural networks with {ReLU} activation function''.
\newblock \href{https://dx.doi.org/10.1214/19-AOS1875}{The Annals of Statistics {\bf 48}, 1875--1897}~(2020).

\bibitem{wilde2025-QIM}
Mark~M. Wilde.
\newblock ``{Quantum Fisher information matrices from R\'enyi relative entropies}''~(2025).
\newblock  \href{http://arxiv.org/abs/2510.02218}{arXiv:2510.02218}.

\bibitem{SB-IT-2025}
Sreejith Sreekumar and Mario Berta.
\newblock ``Limit distribution theory for quantum divergences''.
\newblock \href{https://dx.doi.org/10.1109/TIT.2024.3482168}{IEEE Transactions on Information Theory {\bf 71}, 459--484}~(2025).

\bibitem{Carlen-2010}
Eric~A. Carlen.
\newblock ``Trace inequalities and quantum entropy: An introductory course''.
\newblock In Contemporary Mathematics.
\newblock Volume 529 of Entropy and the Quantum.
\newblock American Mathematical Society~(2010).

\bibitem{Davis-1975}
Philip~J. Davis.
\newblock ``Interpolation and approximation''.
\newblock Dover Publications, Inc. New York. ~(1975).

\bibitem{Berry-41}
Andrew~C. Berry.
\newblock ``The accuracy of the {Gaussian} approximation to the sum of independent variates''.
\newblock \href{https://dx.doi.org/10.1090/s0002-9947-1941-0003498-3}{Transactions of the American Mathematical Society {\bf 49}, 122--136}~(1941).

\bibitem{Esseen-42}
Carl-Gustav Esseen.
\newblock ``On the {Liapunoff} limit of error in the theory of probability''.
\newblock Arkiv f\"{o}r matematik, astronomi och fysik {\bf 28A}, 1--19~(1942).

\bibitem{Ross-Selinger-2016}
Neil~J. Ross and Peter Selinger.
\newblock ``Optimal ancilla-free {Clifford + $T$} approximation of $z$-rotations''.
\newblock \href{https://dx.doi.org/10.26421/QIC16.11-12-1}{Quantum Information and Computation {\bf 16}, 901–953}~(2016).

\bibitem{Dawson-Nielsen-2006}
Christopher~M. Dawson and Michael~A. Nielsen.
\newblock ``{The Solovay-Kitaev algorithm}''.
\newblock \href{https://dx.doi.org/10.26421/QIC6.1-6}{Quantum Information and Computation {\bf 6}, 81–95}~(2006).

\end{thebibliography}

\appendix
\section{Optimization of Measured Relative Entropies}
\label{Sec:optmeasrelent}
For $\rho,\sigma>0$,  the optimizers $\omega_{1}^{\star}(\rho,\sigma)$ and $\omega_{\alpha}^{\star}(\rho,\sigma)$ of  the RHSs of~\eqref{eq:measured-rel-ent-opt} and~\eqref{eq:varformrenent2}, respectively,  are given by 
\begin{align} \label{eq:optrelvarform}
\omega_{\alpha}^{\star}(\rho,\sigma)= \begin{cases}
e^{(\alpha-1) H_{\alpha}^{\star}(\rho,\sigma)}=\sum_{i=1}^d\left(\frac{\Tr[P_{i,\alpha}^{\star} \rho]}{\Tr[P_{i,\alpha}^{\star} \sigma]} \right)^{\alpha-1} P_{i,\alpha}^{\star}, &\qquad\text{ for }\alpha  \in \left(0,\frac{1}{2}\right),  \\
e^{\alpha H_{\alpha}^{\star}(\rho,\sigma)}=\sum_{i=1}^d\left(\frac{\Tr[P_{i,\alpha}^{\star} \rho]}{\Tr[P_{i,\alpha}^{\star} \sigma]} \right)^{\alpha} P_{i,\alpha}^{\star}, &\qquad\text{ for }\alpha  \in\left[\frac{1}{2},\infty\right ), 
\end{cases}
\end{align}
where $\{P_{i,\alpha}^{\star}\}_{i=1}^d$ are the orthogonal projectors in the spectral decomposition of $H_{\alpha}^{\star}(\rho,\sigma)$ from~\eqref{eq:opteigvalues}.  This can be verified by substituting $\omega_{\alpha}^{\star}(\rho,\sigma)$ in~\eqref{eq:varformrenent2} and noting that the RHS of~\eqref{eq:varformrenentmain2} coincides with that obtained by substituting $H_{\alpha}^{\star}(\rho,\sigma)$ in the RHS of~\eqref{eq:renyi_rel_ent_variational}.  
The variational form in~\eqref{eq:varformrenentmain2} and the relation~\eqref{eq:optrelvarform} can be leveraged to show that $H_{\alpha}^{\star}(\rho,\sigma)$ is unique when $\rho,\sigma>0$, and should satisfy the  necessary (and sufficient) condition for optimality derived from it.  \begin{prop}[Necessary and sufficient condition for optimality]
Suppose $\rho,\sigma>0$. 
Then the optimal $\omega>0$ in~\eqref{eq:measured-rel-ent-opt} is unique and determined by the following equation:
\begin{align}
\sigma & =\int_{0}^{\infty}\ \left(\omega+sI\right)^{-1}\rho\left(\omega+sI\right)^{-1}ds.\label{eq:opt-omega-measure-rel-ent}
\end{align}
The
optimal $\omega>0$ in~\eqref{eq:varformrenent2} is unique and determined by the equations:
\begin{align}
\rho & =\left(\frac{1-\alpha}{\alpha}\right)\frac{\sin\!\left(\left(\frac{\alpha}{1-\alpha}\right)\pi\right)}{\pi}\int_{0}^{\infty}\ t^{\frac{\alpha}{\alpha-1}}\left(\omega+tI\right)^{-1}\sigma\left(\omega+tI\right)^{-1}dt,~~~ \forall~ \alpha\in\left(0,\frac{1}{2}\right),\label{eq:optimal-omega-alpha-0-half} \\
\sigma \mspace{-2 mu}& =\mspace{-2 mu}\left(\frac{\alpha}{1-\alpha}\right)\frac{\sin\!\left(\left(\frac{1-\alpha}{\alpha}\right)\pi\right)}{\pi}\mspace{-2 mu}\int_{0}^{\infty}\mspace{-2 mu} t^{\frac{\alpha-1}{\alpha}}\mspace{-2 mu}\left(\omega+tI\right)^{-1}\rho\left(\omega+tI\right)^{-1}dt, ~ \forall~ \alpha\in\left[\frac{1}{2},1\right)\cup\left(1,\infty\right). \label{eq:optimal-omega-alpha-half-infty}
\end{align}
\end{prop}
\begin{proof}
Let $A \in \cL(\HH_d)$, $\omega>0$ and $r\in\left(-1,0\right)\cup\left(0,1\right)$. We will use the following expressions  for the relevant matrix derivatives of interest  (see \cite{wilde2025-QIM} for a detailed derivation from first principles):  
\begin{align}
\frac{\partial}{\partial \omega}\Tr[A\log \omega] & =\int_{0}^{\infty}\ \left(\omega+sI\right)^{-1}A\left(\omega+sI\right)^{-1}ds,\label{eq:deriv-log}\\
\frac{\partial}{\partial \omega}\Tr[A\omega^{r}] & =\frac{\sin(r\pi)}{\pi}\int_{0}^{\infty}\ t^{r}\left(\omega+tI\right)^{-1}A\left(\omega+tI\right)^{-1}dt, \quad \forall r\in\left(-1,0\right)\cup\left(0,1\right).\label{eq:power-func-deriv-minus-1-to-plus-1}
\end{align}
The expressions in \eqref{eq:deriv-log} and \eqref{eq:power-func-deriv-minus-1-to-plus-1} also equal the  directional derivatives of $\log \omega$ and $\omega^r$, respectively, in the direction $A$ (see \cite[Proof of Theorem 1 and 5]{SB-IT-2025} for expressions of directional derivatives based on integral expressions given in \cite[Lemma 2.8]{Carlen-2010}).

To determine  $\omega_{1}^{\star}(\rho,\sigma)$,
we set the matrix derivative  with respect to $\omega$ of the expression to be optimized in~\eqref{eq:measured-rel-ent-opt} to zero:
\begin{align}
0=\frac{\partial}{\partial\omega}\left(\Tr[\rho\log\omega]+1-\Tr[\sigma\omega]\right) & =\frac{\partial}{\partial\omega}\Tr[\rho\log\omega]-\frac{\partial}{\partial\omega}\Tr[\sigma\omega] \notag\\
 & =\int_{0}^{\infty}\ \left(\omega+sI\right)^{-1}\rho\left(\omega+sI\right)^{-1}ds-\sigma, \notag
\end{align}
where we applied~\eqref{eq:deriv-log}. We conclude that  $\omega_{1}^{\star}(\rho,\sigma)$ should satisfy~\eqref{eq:opt-omega-measure-rel-ent}. 
Uniqueness of $\omega_{1}^{\star}(\rho,\sigma)$ and sufficiency of the above condition for optimality follows from the strict concavity of the function $\omega\mapsto\Tr[\rho\log\omega]+1-\Tr[\sigma\omega]$
for $\rho,\sigma>0$,  which in turn follows from strict operator concavity of $\log $ and the fact that $\Tr[AB]> 0$ for $A,B >0$. 

To determine the optimal choice of $\omega$ in~\eqref{eq:varformrenent2} for $\alpha\in\left(0,\frac{1}{2}\right)$,
we set the  matrix derivative with respect to $\omega$ of the relevant expression to zero as follows:
\begin{align}
0&=  \frac{\partial}{\partial\omega}\left(\alpha\Tr[\rho\omega]+\left(1-\alpha\right)\Tr[\sigma\omega^{\frac{\alpha}{\alpha-1}}]\right) \nonumber \\
 &  =\alpha\rho+\left(\alpha-1\right)\frac{\sin\!\left(\left(\frac{\alpha}{1-\alpha}\right)\pi\right)}{\pi}\int_{0}^{\infty}\ t^{\frac{\alpha}{\alpha-1}}\left(\omega+tI\right)^{-1}\sigma\left(\omega+tI\right)^{-1}dt.\notag
\end{align}
Here, we used~\eqref{eq:power-func-deriv-minus-1-to-plus-1} with  $r=\frac{\alpha}{\alpha-1}\in\left(-1,0\right)$. Hence, $\omega_{\alpha}^{\star}(\rho,\sigma)$ satisfies~\eqref{eq:optimal-omega-alpha-0-half}.  Uniqueness of $\omega_{\alpha}^{\star}(\rho,\sigma)$ and sufficiency of the above equation for optimality follows from the strict convexity of the function $\omega\mapsto\alpha\Tr[\rho\omega]+\left(1-\alpha\right)\Tr[\sigma\omega^{\frac{\alpha}{\alpha-1}}]$
for $\rho,\sigma>0$, which in turn follows from strict operator convexity of $x \mapsto x^{r}$ for $r \in (-1,0)$ and $\Tr[AB]> 0$ for $A,B >0$.

To determine the optimal choice of $\omega$ in~\eqref{eq:varformrenent2} for $\alpha\in\left[\frac{1}{2},1\right)$, 
we set the following matrix derivative to zero:
\begin{align}
 0&= \frac{\partial}{\partial\omega}\left(\alpha\Tr[\rho\omega^{1-\frac{1}{\alpha}}]+\left(1-\alpha\right)\Tr[\sigma\omega]\right)\nonumber \\ 
  & =\left(1-\alpha\right)\sigma-\alpha\frac{\sin\!\left(\left(\frac{1}{\alpha}-1\right)\pi\right)}{\pi}\int_{0}^{\infty}\ t^{\frac{\alpha-1}{\alpha}}\left(\omega+tI\right)^{-1}\rho\left(\omega+tI\right)^{-1}dt. \notag 
\end{align}
Here, we applied~\eqref{eq:power-func-deriv-minus-1-to-plus-1} with $r=\frac{\alpha-1}{\alpha} \in (-1,0)$.
Hence, we conclude that $\omega_{\alpha}^{\star}(\rho,\sigma)$ should satisfy~\eqref{eq:optimal-omega-alpha-half-infty}. Uniqueness of $\omega_{\alpha}^{\star}(\rho,\sigma)$ and sufficiency of the condition above for optimality follows from the strict convexity of $\omega\mapsto\alpha\Tr[\rho\omega^{1-\frac{1}{\alpha}}]+\left(1-\alpha\right)\Tr[\sigma\omega]$, which is a consequence of strict operator convexity of $x \mapsto x^{r}$ for $r \in (-1,0)$ and $\Tr[AB]> 0$ for $A,B >0$.

Finally, note that $\omega_{\alpha}^{\star}(\rho,\sigma)$ for $\alpha \in (1,\infty)$ should satisfy 
\begin{align}
 0&= \frac{\partial}{\partial\omega}\left(\alpha\Tr[\rho\omega^{1-\frac{1}{\alpha}}]+\left(1-\alpha\right)\Tr[\sigma\omega]\right)\nonumber \\
  & =\alpha\frac{\sin\!\left(\left(\frac{\alpha-1}{\alpha}\right)\pi\right)}{\pi}\int_{0}^{\infty}\ t^{1-\frac{1}{\alpha}}\left(\omega+tI\right)^{-1}\rho\left(\omega+tI\right)^{-1}dt+\left(1-\alpha\right)\sigma, \notag
\end{align}
where we again applied~\eqref{eq:power-func-deriv-minus-1-to-plus-1} with $r=1-\frac{1}{\alpha}\in\left(0,1\right)$. Hence, 
 $\omega_{\alpha}^{\star}(\rho,\sigma)$ should satisfy~\eqref{eq:optimal-omega-alpha-half-infty}. 
Uniqueness of  $\omega_{\alpha}^{\star}(\rho,\sigma)$ and sufficiency of the above condition for characterizing $\omega_{\alpha}^{\star}(\rho,\sigma)$ follows from the strict concavity of the
function $\omega\mapsto\alpha\Tr[\rho\omega^{1-\frac{1}{\alpha}}]+\left(1-\alpha\right)\Tr[\sigma\omega]$
for $\rho,\sigma>0$, which in turn follows from strict operator concavity of  $x \mapsto  x^r$ for $r \in (0,1)$. 
\end{proof}
\begin{remark}[Limit of optimality condition]
The equation in~\eqref{eq:opt-omega-measure-rel-ent} arises as the
limit of~\eqref{eq:optimal-omega-alpha-half-infty} as $\alpha\to1$.
Indeed, this follows because
\begin{align}
\lim_{\alpha\to1}\left(\frac{\alpha}{1-\alpha}\right)\frac{\sin\!\left(\left(\frac{1-\alpha}{\alpha}\right)\pi\right)}{\pi} & =1 \quad \mbox{and} \quad \notag 
\lim_{\alpha\to1}t^{\frac{\alpha-1}{\alpha}}  =1, \notag 
\end{align}
so that
\begin{multline}
\lim_{\alpha\to1}\left(\frac{\alpha}{1-\alpha}\right)\frac{\sin\!\left(\left(\frac{1-\alpha}{\alpha}\right)\pi\right)}{\pi}\int_{0}^{\infty}\ t^{\frac{\alpha-1}{\alpha}}\left(\omega+tI\right)^{-1}\rho\left(\omega+tI\right)^{-1}dt\\
=\int_{0}^{\infty}\ \left(\omega+sI\right)^{-1}\rho\left(\omega+sI\right)^{-1}ds.\notag 
\end{multline}
Here, the equality is justified by an application of the dominated convergence theorem.
\end{remark}

\section{Lower Bound on Minimax Risk for Estimating Measured Relative Entropies}
\label{App:lowerbndest}
The proof of   \eqref{eq:lwrbndminmaxrisk} and \eqref{eq:lwrbndminmaxrisk-ren} relies  on an application of  Le Cam's two-point method and its generalized version with composite hypothesis,  as done in the proof of \cite[Proposition~2 and 3]{Yihong-Yang-2016} and \cite[Theorems~1 and~2]{Yanjun-2020-alphadiv} (see the arXiv version). Since the proofs are similar to those of the lower bounds in \cite{Yanjun-2020-alphadiv}, we will only provide a sketch highlighting the differences in our case. 

To prove \eqref{eq:lwrbndminmaxrisk}, consider the following discrete probability distributions supported on $\{1,2,\ldots,d\}$  with $\epsilon \in (0,0.5)$ and $b \geq 2$:
\begin{subequations}\label{eq:distlecam1}
\begin{align}
 & Q_1\mspace{-2 mu}=\mspace{-2 mu}\left(\mspace{-2 mu}\frac{1-\epsilon}{b(d-1)},\ldots,\frac{1-\epsilon}{b(d-1)}, 1-\frac{(1-\epsilon)}{b}\mspace{-2 mu} \right), ~~ Q_2\mspace{-2 mu}=\mspace{-2 mu}\left(\mspace{-2 mu}\frac{1+\epsilon}{b(d-1)},\ldots,\frac{1+\epsilon}{b(d-1)}, 1-\frac{(1+\epsilon)}{b} \mspace{-2 mu}\right),\\
 &\qquad \qquad \qquad \qquad \qquad \qquad \qquad  P=\left(\frac{1}{2(d-1)},\ldots,\frac{1}{2(d-1)}, \frac{1}{2} \right). 
\end{align}
\end{subequations}
Denoting the $i^{th}$ entry of $P,Q_1,Q_2$ by $p_i$, $q_{1,i}$ and $q_{2,i}$, respectively, we have
\begin{align}
    \max_{1 \leq i \leq d} \frac{p_i}{q_{1,i}} \vee \frac{p_i}{q_{2,i}} \vee \frac{q_{1,i}}{p_i} \vee \frac{q_{2,i}}{p_i} \leq b. \notag
\end{align}
Hence, $(P,Q_1), (P,Q_2) \in \cS_d(b)$ (when interpreted as quantum states diagonal in the same eigen-basis, say, the computational basis). 
Consequently, via an application of Le Cam's lemma as in the proof of \cite[Theorem 2]{Yanjun-2020-alphadiv} yields $R_n^{\star}(\cS_d(b)) \gtrsim b^{1/2}n^{-1/2}$, where $R_n^{\star}(\cS_d(b)) $ is the minimax risk  defined in \eqref{eq:minimaxrisk}. This yields the first term on the RHS of \eqref{eq:lwrbndminmaxrisk}. 

It remains to show that $R_n^{\star}(\cS_d(b)) \gtrsim d b/(n \log d)$. 
For this purpose, consider the following distributions\footnote{We note that the distributions in \eqref{eq:distlecam2} differs slightly from those used in the proof of lower bound in \cite[Theorem 2]{Yanjun-2020-alphadiv}, in that $\log d$ here is replaced by $\log n$ therein. This necessitates an additional condition $\log d \gtrsim \log n$ under which their lower bound holds. However, by using \eqref{eq:distlecam2}, it can be seen via essentially the same steps as in their proof that this constraint can be avoided.} supported on $\{1,2,\ldots,d\}$ similar to that in \cite[Theorem 2]{Yanjun-2020-alphadiv}:
\begin{subequations}\label{eq:distlecam2}
\begin{align}
 & P=\left(\frac{cb}{n \log d},\ldots,\frac{cb}{n \log d}, 1-\frac{c(d-1)b}{n \log d} \right),  \\
  & Q_1=\left(q_{1,i},\ldots,q_{1,d-1}, 1-\sum_{i=1}^{d-1}q_{1,i}\right), \quad Q_2=\left(q_{2,i},\ldots,q_{2,d-1}, 1-\sum_{i=1}^{d-1}q_{2,i}\right), 
\end{align}
\end{subequations}
where $c>0$ is chosen such that $P$ is a probability distribution (possible under the condition $db \lesssim n \log d$), and each of the first $d-1$ entries of  $Q_i$, $i=1,2$,  are chosen independently according to probability measure $\nu_i$ supported on $I=\left[c_1/(n \log d), c_2 \log d/n \right]$ satisfying $\EE_{\nu_1}[\log X]-\EE_{\nu_2}[\log X] \gtrsim d/(n\log d)$ and additional matching moment conditions as in \cite{Yihong-Yang-2016}.  The existence of such probability measures $\nu_1$ and $\nu_2$ has been shown in \cite{Yihong-Yang-2016}. Since $Q_1$ and $Q_2$ are randomly generated, they need not be  probability distributions, i.e., the entries need not sum to 1.  However, one may restrict to approximate probability distributions with strictly positive entries such that $\abs{\norm{q_i}_1-1} \leq 1+\epsilon$, with $\epsilon$ chosen appropriately and the constants $c,c_1,c_2$ adjusted, such that the probability distributions $Q_1',Q_2'$ obtained after normalization by sum of its entries satisfy (almost surely)
\begin{align}
    \max_{1 \leq i \leq d} \frac{p_i}{q'_{1,i}} \vee \frac{p_i}{q'_{2,i}} \vee \frac{q'_{1,i}}{p_i} \vee \frac{q'_{2,i}}{p_i} \leq b \vee \frac{(\log d)^2}{b} \leq b. \notag
\end{align} 
In the last inequality, we used that $b \geq  (\log d)^2$. The minimax risk can then be lower bounded by the minimax risk over the set of such approximate probability distributions under a Poisson sampling model up to an additive factor that vanishes as $\epsilon$ tends to zero. The minimax risk under the Poisson sampling model can in turn be further lower bounded  (see \cite[Proof of Lemma 2]{Yihong-Yang-2016}) as $\Omega\big(db/(n\log d)\big)$ via generalized Le Cam's method. Specifically, this can be done by taking the prior distributions used in generalized Le Cam's method to be the truncated measures obtained (via change of measure)  by restricting support to distributions $(P,Q_i')$ for $i=1,2$, such that $\abs{\sD_{\mathsf{M}}(P\Vert Q_1')-\sD_{\mathsf{M}}(P\Vert Q_2')}=\Omega\big(db/(n\log d)\big)$ for every  $(P,Q_1')$ and $(P,Q_2')$. This change of measure incurs an additive penalty on the lower bound that vanishes with an appropriately chosen large enough $\epsilon$ (that tends to zero at an appropriate rate). Finally, the total variation distance between the sufficient statistics (histogram) under these priors in the Poisson sampling model decays exponentially (see \cite[Lemma 3]{Yihong-Yang-2016}) due to the matching moments conditions under the priors $\nu_1$ and $\nu_2$, which can be used to obtain the desired lower bound via \cite[Lemma 2]{Yihong-Yang-2016}. To summarize, the above steps lead to  $R_n^{\star}(\cS_d(b)) \gtrsim d b/(n \log d)$ given that $b \geq (\log d)^2$ and $db \lesssim n \log d$, thus concluding the proof sketch for the case of measured relative entropy. 

\medskip
To prove  \eqref{eq:lwrbndminmaxrisk-ren}, we first observe that the minimax risk for estimating R\'{e}nyi divergence of order $\alpha$ can be lower bounded in terms of the minimax risk for estimating the classical $\alpha$-divergence, where $\alpha$-divergence (see \cite{Yanjun-2020-alphadiv}) between two discrete distributions $P$ and $Q$ supported on $\{1,\ldots,d\}$ is 
\begin{align}
A_{\alpha}(P\|Q)\coloneqq \frac{1}{\alpha (\alpha-1)} \left(-1+\sum_{i=1}^d p_i^{\alpha} q_i^{1-\alpha}\right). \notag
\end{align}
Note that 
\begin{align}
    A_{\alpha}(P\|Q)=\frac{e^{(\alpha-1)\sD_{\mathsf{M},\alpha}(P\Vert Q)}-1}{\alpha(\alpha-1)}. \notag 
\end{align}
Consider $\alpha >1$.  For $(P,Q) \in \cS_d(b)$, we have 
\begin{align}
   \sD_{\mathsf{M},\alpha}(P\Vert Q)\coloneqq \frac{1}{\alpha-1}\log \left(\sum_{i=1}^d p_i^{\alpha} q_i^{1-\alpha}\right) \leq \frac{1}{\alpha-1}\log \left(\sum_{i=1}^d q_i b^{\alpha-1}\right) \leq \log b.\notag
\end{align}
Hence, without loss of generality, it is sufficient to consider  estimators $\hat{\sD}_{n,\alpha}$  satisfying $\hat{\sD}_{n,\alpha} \leq \log b$. Given such an estimator $\hat{\sD}_{n,\alpha} $,  a natural estimator for $A_{\alpha}(P\|Q)$ is  
\begin{align}
    \hat{A}_{n,\alpha} =\frac{e^{(\alpha-1)\hat{\sD}_{n,\alpha}}-1}{\alpha(\alpha-1)}.\notag
\end{align}
We can upper bound the absolute error for $\hat{A}_{n,\alpha}$ in terms of $\hat{\sD}_{n,\alpha}$ as 
\begin{align}
    \abs{\hat{A}_{n,\alpha}-A_{\alpha}(P\|Q)}\leq \frac{\abs{e^{(\alpha-1)\hat{\sD}_{n,\alpha}}-e^{(\alpha-1)\sD_{\mathsf{M},\alpha}(P\Vert Q)}}}{\alpha(\alpha-1)}\leq \frac{b^{\alpha-1}\abs{\hat{\sD}_{n,\alpha}-\sD_{\mathsf{M},\alpha}(P\Vert Q)}}{\alpha}.\notag 
\end{align}
Hence
\begin{align}
   \abs{\hat{\sD}_{n,\alpha}-\sD_{\mathsf{M},\alpha}(P\Vert Q)} \geq \alpha b^{1-\alpha} \abs{\hat{A}_{n,\alpha}-A_{\alpha}(P\|Q)}.\notag 
\end{align}
Denoting  the minimax risk for estimating classical $\alpha$-divergence over the class $\cS_d^{(\mathsf{c})}(b)$ of distributions $(P,Q)$ satisfying $\max_{1 \leq i \leq d} (p_i/q_i)\vee (q_i/p_i) \leq b$ by $\bar R_{n,\alpha}^{\star}\big(\cS_d^{(\mathsf{c})}(b)\big)$, the above equation implies that
$R_{n,\alpha}^{\star}\big(\cS_d(b)\big) \geq \alpha b^{1-\alpha} \bar R_{n,\alpha}^{\star}\big(\cS_d^{(\mathsf{c})}(b)\big)$. Now, using Le Cam's two-point method and its generalized version with the distributions given in \eqref{eq:distlecam1} and \eqref{eq:distlecam2}, it can be shown using similar steps as in the proof of the lower bound in \cite[Theorem 1]{Yanjun-2020-alphadiv} that given $b \geq (\log d)^2 \vee 2$ and $db^{\alpha} \lesssim n \log d$:
\begin{align}
    \bar R_{n,\alpha}^{\star}\mspace{-4 mu}\left(\cS_d^{(\mathsf{c})}(b)\right) \gtrsim_{\alpha} \frac{b^{\alpha-\frac 12}}{\sqrt{n}}+\frac{db^{\alpha}}{n \log d}, \notag 
\end{align}
which yields the desired lower bound
\begin{align}
 R_{n,\alpha}^{\star}\big(\cS_d(b)\big) \gtrsim_{\alpha} \frac{\sqrt{b}}{\sqrt{n}}+\frac{db}{n \log d}. \notag 
\end{align}

\section{Proof of Proposition~\ref{Prop:densityopclass}}
\label{Prop:densityopclass-proof}
We first consider Part $(i)$. Since $\tilde{\cS}_{k,d,\alpha}(b,a,\varepsilon) \subseteq \cS_d(b)$ by definition, we only need to show that $\cS_d(b) \subseteq 
 \tilde{\cS}_{k,d,\alpha}(b,a,\varepsilon) $ for $k$, $a$, and $\varepsilon$ as in the statement of the proposition. We will only consider $\alpha=1$ as the proof is exactly same for other $\alpha$. 
Consider any $(\rho,\sigma) \in \cS_d(b)$ and let $ \Lambda^{\star}(\rho,\sigma)=\{\lambda_j^{\star}\}_{j=1}^d$ be any enumeration. By Lemma~\ref{Lem:bndtraceratio}, we have  $\abs{\lambda_j^{\star}} \leq \log  b $ for all $ j \in [1:d]$.
Let   
\begin{align}
    p(x)= \sum_{j=1}^d \lambda_j^{\star} \prod_{i \neq j} \frac{dx-i}{j-i}. \notag
\end{align}
Note that $p(i/d)=\lambda_i^{\star}$ and that $p(x)$ is a polynomial of degree $d-1$,  which can be written as
\begin{align}
  p(x)  &= \sum_{j=1}^d \tilde \lambda_j  \tilde p_j(x), \quad 
   \mbox{where} \quad   
  \tilde p_j(x) \coloneqq  \prod_{i \neq j} \left(x-\frac{i}{d}\right) \quad 
 \mbox{and} \quad \tilde \lambda_j \coloneqq \frac{ \lambda_j^{\star} d^{d-1}}{\prod_{i \neq j} (j-i)}. \notag
\end{align}
Further, observe that  $\tilde p_j(x)$ is a polynomial of the form $\sum_{l=0}^{d-1} \tilde a_{l,j} x^l$ with $\abs{\tilde a_{l,j}} \leq 1$ for all $j \in [1:d]$ and $l \in [0:d-1]$. Hence, $p(x)$ is a polynomial of the form $\sum_{l=0}^{d-1} a_l x^l$
with 
\begin{align}
\abs{a_l} \leq \sum_{j=1}^d \big|\tilde \lambda_j\big| =  \sum_{j=1}^d \frac{ \abs{\lambda_j^{\star}} d^{d-1}}{\prod_{i \neq j} \abs{j-i}} \stackrel{(a)}{\leq } \sum_{j=1}^d\frac{ d^{d-1} \log  b  }{\prod_{i \neq j} \abs{j-i}} \stackrel{(b)}{\leq } \frac{(3e)^d \log  b }{2\pi},\notag
\end{align}
where 
\begin{enumerate}[(a)]
    \item  used that $\abs{\lambda_j^{\star}} \leq \log  b $ for all $j \in [1:d]$;
    \item  follows because for all $j \in [1:d]$,  
\begin{align}
    \prod_{i \neq j} \abs{j-i} \geq \min_{1 \leq m \leq d} (m-1)! (d-m)!=\begin{cases}
        \left(\left(\frac{d-1}{2}\right)!\right)^2, & \quad \mbox{ if } d \mbox{ is odd},\\
        \left(\frac{d-2}{2}\right)!\left(\frac{d}{2}\right)! , & \quad \mbox{ if } d \mbox{ is even,}
    \end{cases} \geq 2 \pi  \left( \frac{d}{3e} \right)^d; \notag
\end{align}
  The last inequality can be verified to hold by computation for $1 \leq d \leq 4$, and for $d>4$, we used   
    $d-2 \geq 0.5 d$, $(d-1)/2 \geq (d/3)$, and the  DeMoivre--Stirling lower bound  for the factorial given by $n! > \sqrt{2 \pi n}\left(\frac{n}{e}\right)^n$.
\end{enumerate}   
This implies that $\cS_d(b) \subseteq 
 \tilde{\cS}_{k,d}(b,a,\varepsilon) $ for $k$, $a$ as in the statement of the proposition and $\varepsilon=0$. Since $\tilde{\cS}_{k,d}(b,a,\varepsilon) \subseteq \tilde{\cS}_{k,d}(b,a,\varepsilon') $ for $\varepsilon' \geq \varepsilon$,  the claim in Part $(i)$ follows.

 \medskip

 To prove  Part $(ii)$, let  $\Lambda^{\star}(\rho,\sigma)=\{\lambda_i^{\star}\}_{i=1}^d$ be an  (arbitrary) enumeration. We will consider approximation of a Lipschitz continuous function $f$ satisfying $f(i/d)=\lambda_i^{\star}$  by polynomials of bounded degree and bounded coefficients. Specifically, with $\lambda_0^{\star}=0$, let $\hat f\colon [0,1] \mapsto \left[-\log  b ,\log  b \right]$ be the piecewise linear function given by 
 $\hat f(x)=\lambda^{\star}_{i+1}(dx-i)-\lambda_i^{\star}(dx-i-1)$ for $i/d \leq x \leq (i+1)/d$ and $0 \leq i \leq d$. Note that $\hat f(i/d)=\lambda_i^{\star}$, $\big\|\hat f\big\|_{\infty,[0,1]} \leq \log  b $,  and the Lipschitz constant of $\hat f$ is upper bounded by $2d \log  b $ for $(\rho,\sigma) \in \cS_d(b)$. 
 We next consider Bernstein's approximation of $\hat f$ using  polynomials of degree $k$, i.e., using
 \begin{subequations}\label{eq:Bernsteinpolyapprox}
      \begin{align}
    & p_{k,\hat f}(x)= \sum_{m=0}^k B_{m,k}(x) \hat f\!\left(\frac{m}{k}\right), \quad x \in [0,1], \label{eq:bernsteinapprox}\\
     \mbox{where} \quad & 
     B_{m,k}(x)\coloneqq \binom{k}{m} x^m (1-x)^{k-m}. \label{eq:bernsteinmonom}
 \end{align}
 \end{subequations}
Let $\ell \in \NN \cup \{0\}$. By straightforward computation,   the $\ell^{th}$ derivative of $B_{m,k}(x)$ at $x=0$ is 
 \begin{align}
     B_{m,k}^{(\ell)}(0)=\frac{k !}{(k-\ell)!} \binom{\ell}{m} (-1)^{m+\ell}, ~\forall \ell \in \NN. \label{eq:derexpberstein} 
 \end{align}
 Further,  for any  $\ell$-times differentiable function $f$,
       \begin{align}
\norm{p_{k,f}^{(\ell)}}_{\infty,(0,1)} \leq \norm{f^{(\ell)}}_{\infty,(0,1)}. \label{eq:derbndbernstein}
 \end{align}
We will use the following fact that  bounds the  error of approximating a Lipschitz $f$ using its Bernstein's approximation. 
\begin{lemma}[Error in Bernstein's approximation, see e.g., \cite{Davis-1975}] \label{Lem:Bernstein-approxerr}
For all $f\colon [0,1] \mapsto \RR$ and $k \in \NN $, we have
\begin{align}
\norm{p_{k,f}-f}_{\infty,[0,1]} \leq \left( \norm{f}_{\mathrm{Lip},[0,1]}+0.5\norm{f}_{\infty,[0,1]}\right) k^{-\frac 13},\label{eq:approxbndbernstein}
\end{align}
where 
\begin{align}
    \norm{f}_{\mathrm{Lip},[0,1]}=\sup_{\substack{x,y \in [0,1] \\ x \neq y}} \frac{\abs{f(x)-f(y)}}{\abs{x-y}}. \notag
\end{align}
\end{lemma}
Lemma~\ref{Lem:Bernstein-approxerr}  follows by optimizing $\delta$ in the expression stated below as given in \cite[Page 111]{Davis-1975}:
\begin{align}
\norm{p_{k,f}-f}_{\infty,[0,1]} \leq \norm{f}_{\mathrm{Lip},[0,1]} \delta+ \frac{\norm{f}_{\infty,[0,1]}}{2k\delta^2}, ~\forall~ \delta \geq 0. \label{eq:bernsteinapprx}
\end{align}
An independent proof via an application of the Berry-Esseen theorem with the same approximation error rate in terms of $k$ (albeit with slightly larger constants) is provided below, which was developed as we were initially unaware of~\eqref{eq:bernsteinapprx}.

\medskip

Applying Lemma~\ref{Lem:Bernstein-approxerr} to $\hat f$ as above, we obtain that 
\begin{align}
    \norm{p_{k,\hat f}-\hat f}_{\infty,[0,1]} \leq \left(2d+0.5\right) (\log  b)  k^{-\frac 13}. \notag 
\end{align}
Moreover, the coefficient $a_{\ell}\big(p_{k,\hat f}\big)$ of $x^{\ell}$ in $p_{k,\hat f}$ is upper bounded by
\begin{align}
  \abs{a_{\ell}\big(p_{k,\hat f}\big)} =  \abs{p^{(\ell)}_{k,\hat f}(0)} \leq  \sum_{m=0}^k \big\|\hat f\big\|_{\infty,[0,1]} \abs{  B_{m,k}^{(\ell)}(0)} \leq \log  b  \sum_{m=0}^k \abs{  B_{m,k}^{(\ell)}(0)}, \notag 
\end{align}
which yields via~\eqref{eq:derexpberstein} that
\begin{align}
  \max_{\ell \in [0: k]}  \abs{a_{\ell}\big(p_{k,\hat f}\big)} \leq  2^k k! \log  b . \notag 
\end{align}
This implies by the definition of $\tilde{\cS}_{k,d}(b,a,\varepsilon)$ that Part $(ii)$ holds.
\subsection*{Alternative Proof of Lemma~\ref{Lem:Bernstein-approxerr} with  Different Constants} 
 Since the bound is trivial if $\norm{f}_{\infty,[0,1]}=\infty$ or $\norm{f}_{\mathrm{Lip},[0,1]}=\infty$, we may assume that $\norm{f}_{\infty,[0,1]} \vee \norm{f}_{\mathrm{Lip},[0,1]} <\infty $.
Recalling the definitions of $p_{k,f}$ and $B_{m,k}$ given in~\eqref{eq:Bernsteinpolyapprox},  note that $B_{m,k}(0)=0$ for $m \in [1:k]$ and $B_{0,k}(0)=1$. Similarly, $B_{m,k}(1)=0$ for $m \in [0:k-1]$ and $B_{k,k}(1)=1$ (by interpreting $0^0=1$).  Hence, $p_{k,f}(0)=f(0)$ and $p_{k,f}(1)=f(1)$. Consequently, $\abs{p_{k,f}(x)-f(x)}=0$ for $x \in \{0,1\}$. Next, consider any $x \in (0,1)$  and  $\epsilon>0$.  Then, we can bound the approximation error   as follows:
\begin{align}
  &  \abs{p_{k,f}(x)-f(x)} \notag \\
  &=  \abs{ \sum_{m=0}^k B_{m,k}(x) f\!\left(\frac{m}{k}\right)-f(x)} \notag \\
    & \stackrel{(a)}{=}  \sum_{m=0}^k  B_{m,k}(x) \abs{f\!\left(\frac{m}{k}\right)-f(x)} \notag \\
        & =  \sum_{\substack{0 \leq m \leq k:\\ \abs{x-\frac{m}{k}} \leq \epsilon   }}  B_{m,k}(x) \abs{f\!\left(\frac{m}{k}\right)-f(x)}+\sum_{\substack{0 \leq m \leq k:\\ \abs{x-\frac{m}{k}} > \epsilon   }}  B_{m,k}(x) \abs{f\!\left(\frac{m}{k}\right)-f(x)} \notag \\
        &\stackrel{(b)}{ \leq }   \epsilon \norm{f}_{\mathrm{Lip},[0,1]} \sum_{\substack{0 \leq m \leq k: \abs{x-\frac{m}{k}} \leq \epsilon   }}  B_{m,k}(x) +2 \norm{f}_{\infty,[0,1]} \sum_{\substack{0 \leq m \leq k: \abs{x-\frac{m}{k}} > \epsilon   }}  B_{m,k}(x) \notag \\
        &\stackrel{(c)}{ \leq }  \epsilon \norm{f}_{\mathrm{Lip},[0,1]} +2 \norm{f}_{\infty,[0,1]} \sum_{\substack{0 \leq m \leq k: \abs{x-\frac{m}{k}} > \epsilon   }}  B_{m,k}(x), \label{eq:errboundbern}  
\end{align}
where 
\begin{enumerate}[(a)]
    \item and $(c)$ are due to $B_{m,k}(x)\geq 0$ and $\sum_{m=0}^k  B_{m,k}(x)=1$ for all $x \in [0,1]$;
    \item follows  by the definition of $\norm{f}_{\mathrm{Lip},[0,1]}$ and $\norm{f}_{\infty,[0,1]}$, and $|y-z| \leq |y|+|z|$ for $y,z \in \RR$.
\end{enumerate}
To upper bound the summation in the right hand side of~\eqref{eq:errboundbern}, we use the Berry-Esseen theorem~\cite{Berry-41,Esseen-42} (with refined constants) applied to an i.i.d.~sum of Bernoulli random variables, $\{X_i\}_{i=1}^k$ with $\PP\!\left(X_i=1\right)=x$. This yields 
     \begin{align}
      \PP\left(  \abs{\sum_{i=1}^k X_i -kx}> y \sqrt{kx(1-x)}\right) \leq \sqrt{\frac{2}{\pi}} \int_{y}^{\infty} e^{-\frac{z^2}{2}}dz +\frac{\EE\big[\abs {X_1-x}^3 \big]}{2\sqrt{n} \EE \big[(X_1-x)^2\big]}, \notag
     \end{align}
     which in turn implies that
     \begin{align}
     \sum_{\substack{0 \leq m \leq k: \abs{x-\frac{m}{k}} > \epsilon   }}  B_{m,k}(x)&=\PP\left(  \abs{\sum_{i=1}^k X_i -kx}> k \epsilon\right)    \notag \\
     &\leq \sqrt{\frac{2}{\pi}} \int_{\frac{\epsilon \sqrt{k}}{\sqrt{x(1-x)}}}^{\infty} e^{-\frac{z^2}{2}}dz +\frac{\EE\big[\abs {X_1-x}^3 \big]}{2\sqrt{k} \big(\EE \big[(X_1-x)^2\big]\big)^{\frac 32}} \notag \\
     &\leq \sqrt{\frac{2x(1-x)}{k\pi}} \frac{1}{\epsilon} e^{-\frac{k \epsilon^2}{2x(1-x)}} +\frac{1-2x+2x^2}{2\sqrt{kx(1-x)}}.\notag 
     \end{align}
  Let $k \geq 2$.    For $x \in \big(k^{-1/3},1-k^{-1/3}\big)$ and $\epsilon \geq k^{-\frac 13}$, observe that the right hand side of the above equation is upper bounded by
     \begin{align}
         \frac{k^{\frac 16}}{\sqrt{2\pi }}e^{-2k^{\frac 13}}+2 k^{-\frac 13} \leq 3k^{-\frac 13}. \label{approxsum1}
     \end{align}
     On the other hand, for $x \in (0,k^{-1/3}] $, we have since $p_{k,f}(0)=f(0)$ that
     \begin{align}
    \abs{p_{k,f}(x)-f(x)} & = \abs{p_{k,f}(x)-p_{k,f}(0)+f(0)-f(x)} \notag \\
    &\leq \abs{p_{k,f}(x)-p_{k,f}(0)}+\abs{f(0)-f(x)} \notag \\
    &\leq \left( \norm{p_{k,f}^{(1)}}_{\infty,(0,1)}+\norm{f}_{\mathrm{Lip},[0,1]}  \right)k^{-1/3} \notag \\
    & \leq 2 \norm{f}_{\mathrm{Lip},[0,1]}  k^{-1/3},\label{approxsum2}
     \end{align}
  where in the final inequality, we used $\big\|p_{k,f}^{(1)}\big\|_{\infty,(0,1)} \leq \norm{f}_{\mathrm{Lip},[0,1]}$, which follows from~\eqref{eq:derbndbernstein} with $\ell=1$ and the fact that a $L$-Lipschitz continuous function can be approximated arbitrarily well by a first-order differentiable function with derivative bounded by $L$. 
Similarly, for $x \in  [1-k^{-1/3},1)$,~\eqref{approxsum2} holds.
From~\eqref{eq:errboundbern},~\eqref{approxsum1},~\eqref{approxsum2} and the above discussion, it follows that for all $x \in [0,1]$ and $k \geq 2$, we have
\begin{align}
    \norm{p_{k,f}-f}_{\infty,[0,1]} \leq \left( 2\norm{f}_{\mathrm{Lip},[0,1]}+6\norm{f}_{\infty,[0,1]}\right) k^{-\frac 13}. \notag  
\end{align}
When $k=1$, we obtain using~\eqref{eq:derbndbernstein} with $\ell=0$ that
    \begin{align}
    \norm{p_{k,f}-f}_{\infty,[0,1]}  \leq \norm{p_{k,f}}_{\infty,[0,1]} +\norm{f}_{\infty,[0,1]} \leq 2\norm{f}_{\infty,[0,1]}. \notag  
\end{align}
Hence,~\eqref{eq:approxbndbernstein} holds for all $k \in \NN$,
     thus completing the proof.

\section{Proof of Lemma~\ref{Lem:bndtraceratio}} \label{Sec:Lem:bndtraceratio-proof}
Consider the case $\alpha=1$. 
As shown in~\cite{Berta2015OnEntropies}, for $\rho,\sigma>0$, the optimizer (Hermitian operator) $H^{\star}(\rho,\sigma)$ of~\eqref{eq:rel_ent_variational} has the form 
    \begin{align}
&H^{\star}(\rho,\sigma)= \sum_{i=1}^d \lambda_i^{\star} P_i^{\star}, \quad \mbox{ with } \lambda_i^{\star} \coloneqq\log  \!\left(\frac{\operatorname{Tr}[P_i^{\star} \rho]}{\operatorname{Tr}[P_i^{\star} \sigma]}\right), \notag 
\end{align}
 where  $\{P_i^{\star}\}_{i=1}^d $ is some set of  orthogonal rank-one projectors, so that
\begin{align}
  & \sD_{\mathsf{M}}(\rho\Vert\sigma)= \operatorname{Tr}[H^{\star}(\rho,\sigma)\rho
]-\operatorname{Tr}[e^{H^{\star}(\rho,\sigma)}\sigma]+1 
= \sum_{i=1}^d \operatorname{Tr}[P_i^{\star} \rho] \log  \!\left(\frac{\operatorname{Tr}[P_i^{\star} \rho]}{\operatorname{Tr}[P_i^{\star} \sigma]}\right).\notag
\end{align}
Hence, the set of eigenvalues of $H^{\star}(\rho,\sigma)$ is
\begin{align}
\Lambda^{\star}(\rho,\sigma)=\{\lambda_i^{\star}\}_{i=1}^d= \left\{\log  \!\left(\frac{\operatorname{Tr}[P_i^{\star} \rho]}{\operatorname{Tr}[P_i^{\star} \sigma]}\right)\right\}_{i=1}^d. \notag 
\end{align}
To obtain an upper bound on $\abs{\lambda_i^{\star}}$,   we will use the variational characterization of max-relative entropy given in~\cite[Lemma A.4]{Mosonyi2015QuantumHT}) as stated below:
\begin{align}
    D_{\max}(\rho\|\sigma)\coloneqq \inf \left\{\lambda \in \RR: \rho \leq e^{\lambda}\sigma \right\}= \sup \left\{ \log\! \left(\frac{\operatorname{Tr}[M\rho]}{\operatorname{Tr}[M\sigma]}\right): 0 \leq M \leq I\right\}. \notag
\end{align}
From the last equation above, for  every projector $P$, we have
\begin{align}
 \frac{\operatorname{Tr}[\rho P]}{\operatorname{Tr}[\sigma P]} \leq   e^{D_{\max}(\rho\|\sigma)}=\inf \left\{e^{\lambda} : \sigma^{-\frac 12}\rho \sigma^{-\frac 12} \leq e^{\lambda}  I \right\} = \norm{\sigma^{-\frac 12}\rho \sigma^{-\frac 12}}. \notag
\end{align}
Similarly, 
\begin{align}
 \frac{\operatorname{Tr}[\sigma P]}{\operatorname{Tr}[\rho P]} \leq   e^{D_{\max}(\sigma\|\rho)}= \norm{\rho^{-\frac 12}\sigma \rho^{-\frac 12}}. \notag
\end{align}
Since $(\rho,\sigma) \in \cS_d(b)$, 
we have for every (rank-one) projector $P$ that     
\begin{align} 
\frac{1}{b}\leq \frac{1}{\norm{\rho^{-\frac{1}{2}} \sigma \rho^{-\frac{1}{2}}}} \leq  \frac{\operatorname{Tr}[\rho P]}{\operatorname{Tr}[\sigma P]} \leq \norm{\sigma^{-\frac{1}{2}}\rho \sigma^{-\frac{1}{2}}} \leq b. \notag 
\end{align}
Taking logarithms, we obtain the desired claim for $\alpha=1$. The claim for  $0< \alpha  \neq 1$ follows  similarly by noting that  $H_{\alpha}^{\star}(\rho,\sigma)= \sum_{i=1}^d \lambda_{i,\alpha}^{\star} P_{i,\alpha}^{\star}$ with $ \lambda_{i,\alpha}^{\star} \coloneqq\log  \!\left(\frac{\operatorname{Tr}[P_{i,\alpha}^{\star} \rho]}{\operatorname{Tr}[P_{i,\alpha}^{\star} \sigma]}\right)$.

\section{Solovay-Kitaev Theorem: Bound on Quantum Circuit Parameter Size}
\label{Sec:eff-imp-unit}
Let $\{P_i\}_{i=1}^d$ be an arbitrary set of orthogonal projectors on $\HH_d$ such that 
$\sum_{i=1}^d P_i=I$. Let $U$ be a unitary such that $P_i=U\ket{i}\!\bra{i}U^{\dag}$. 
Note that the determinant of $U$,  $\det(U)=e^{i \omega}$, for some $\omega \in \RR$. Hence, for $\tilde U=e^{-i \omega/d} U$, we have $\det(\tilde U)=1$. Moreover, $\tilde U\ket{i}\!\bra{i}\tilde U^{\dag}= U\ket{i}\!\bra{i} U^{\dag}= P_i$.  This implies that $U$ satisfying $P_i=U\ket{i}\!\bra{i}U^{\dag}$ can be restricted (without loss of generality) to belong to the  special unitary group $\mathsf{SU}(d)$, i.e., the subset of unitary operators on $\HH_d$ with determinant one.
Assume $d=2^N$ for some $N \in \NN$. Consider the universal gate set  $\mathrm{Univ}$ composed of the Clifford group, i.e., the Hadamard gate, $S$ gate, and CNOT gate, along with the $T$ gate. 
Then, it is known from~\cite{Ross-Selinger-2016} that, for every $U \in \mathsf{SU}(d)$ acting on $N$ qubits,  there exists a quantum circuit $\bar U $ constructed by composition of at most $\ell(d,\delta)=O\big(N^2 4^N \log  (N^2 4^N /\delta)\big)$ gates selected from $\mathrm{Univ}$  such that $\big\|U-\bar  U\big\| \leq \delta$. 
Since the number of possible configurations of quantum circuit composed of $\ell$ gates selected from $\mathrm{Univ}$ is $4^{\ell(d,\delta)}$, we obtain that  there exists $\cU(\delta,\cS)$ with $|\cU(\delta,\cS)| \leq 4^{\ell(d,\delta)}$ such that~\eqref{eq:unitaryapprox} is satisfied. Similar claims can also be obtained for qudits, where $d=m^N$ for some positive integers $m,N$, by using a generalization of the Solovay-Kitaev algorithm to qudits as given in \cite{Dawson-Nielsen-2006}.  
\section{Proof of Lemma~\ref{Lem:opnorm-diff}}\label{Sec:Lem:opnorm-diff-proof}
\vspace{-20 pt}
\begin{align}
    \norm{H_1-H_2} 
&= \norm{U_1 \Lambda_1 U_1^{\dag}-U_2 \Lambda_2 U_2^{\dag}} \notag \\
    &\stackrel{(a)}{\leq}   \norm{U_1 \Lambda_1 U_1^{\dag}-U_1 \Lambda_1 U_2^{\dag}}+\norm{U_1 \Lambda_1 U_2^{\dag}-U_1 \Lambda_2 U_2^{\dag}}+\norm{U_1 \Lambda_2 U_2^{\dag}-U_2 \Lambda_2 U_2^{\dag}} \notag \\
    & \stackrel{(b)}{\leq}  \norm{U_1} \norm{\Lambda_1}\norm{U_1^{\dag}- U_2^{\dag}}+\norm{U_1} \norm{\Lambda_1- \Lambda_2} \norm{U_2^{\dag}}+\norm{U_1-U_2} \norm{\Lambda_2} \norm{U_2^{\dag}} \notag \\
    &\stackrel{(c)}{=}  \big(\norm{\Lambda_1}+\norm{\Lambda_2}\big)\norm{U_1- U_2} +\norm{\Lambda_1- \Lambda_2}, \notag
\end{align}
where $(a)$ is via triangle inequality for Schatten norms, $(b)$ is due to sub-multiplicativity of Schatten norms, and $(c)$ is because $\norm{U}=\norm{U^{\dag}}=1$ for a unitary $U$.

\end{document}